\documentclass{article}

\def\arxivmode{1}

\usepackage{microtype}
\usepackage{graphicx}
\usepackage{booktabs} %

\usepackage{hyperref}

\ifdefined\arxivmode
    \usepackage[accepted]{icml2018_arxiv}
\else
    \usepackage[accepted]{icml2018}
\fi

\icmltitlerunning{Stein Points}

\usepackage{amsmath}
\usepackage{amssymb}
\usepackage{amsthm}
\usepackage{bm}
\usepackage{placeins}
\usepackage{subcaption}
\usepackage{MnSymbol}
\usepackage[inline]{enumitem}
\usepackage{xspace}

\newtheorem{theorem}{Theorem}
\newtheorem{definition}{Definition}

\DeclareMathOperator*{\argmin}{arg\,min}
\DeclareMathOperator*{\argmax}{arg\,max}

\begin{document}

\twocolumn[
\icmltitle{Stein Points}

\begin{icmlauthorlist}
\icmlauthor{Wilson Ye Chen}{uts}
\icmlauthor{Lester Mackey}{mic}
\icmlauthor{Jackson Gorham}{od}
\icmlauthor{Fran\c{c}ois-Xavier Briol}{war,imp,ati}
\icmlauthor{Chris. J. Oates}{ncl,ati}
\end{icmlauthorlist}

\icmlaffiliation{uts}{School of Mathematical and Physical Sciences, University of Technology Sydney, Australia}
\icmlaffiliation{mic}{Microsoft Research New England, USA}
\icmlaffiliation{od}{Opendoor Labs, Inc., USA}
\icmlaffiliation{war}{Department of Statistics, University of Warwick, UK}
\icmlaffiliation{imp}{Department of Mathematics, Imperial College London, UK}
\icmlaffiliation{ncl}{School of Mathematics, Statistics and Physics, Newcastle University, UK}
\icmlaffiliation{ati}{Alan Turing Institute, UK}

\icmlcorrespondingauthor{Wilson Ye Chen}{ye.chen@uts.edu.au}
\icmlcorrespondingauthor{Lester Mackey}{lmackey@microsoft.com}

\icmlkeywords{Discrepancy, Kernel, Stein's Method}

\vskip 0.3in
]

\printAffiliationsAndNotice{}  %

\begin{abstract}
An important task in computational statistics and machine learning is to approximate a posterior distribution $p(x)$ with an empirical measure supported on a set of representative points $\{x_i\}_{i=1}^n$.
This paper focuses on methods where the selection of points is essentially deterministic, with an emphasis on achieving accurate approximation when $n$ is small.
To this end, we present {\it Stein Points}.
The idea is to exploit either a greedy or a conditional gradient method to iteratively minimise a kernel Stein discrepancy between the empirical measure and $p(x)$.
Our empirical results demonstrate that Stein Points enable accurate approximation of the posterior at modest computational cost.
In addition, theoretical results are provided to establish convergence of the method.
\end{abstract}

\section{Introduction}

This paper is motivated by approximation of a Borel distribution $P$, defined on a topological space $X$, with deterministic point sets or sequences $\{x_i\}_{i=1}^n \subset X$ for $n \in \mathbb{N}$, such that
\begin{eqnarray}
\textstyle
\frac{1}{n} \sum_{i=1}^n h(x_i) & \rightarrow &\textstyle
 \int h \; \mathrm{d}P
\label{eq: weak convergence}
\end{eqnarray}
as $n \rightarrow \infty$ for all functions $h : X \rightarrow \mathbb{R}$ in a specified set $\mathcal{H}$.
Throughout it will be assumed that $P$ admits a density $p$, with respect to a reference measure, available in a form that is un-normalised (i.e., we know $q(x)$ in closed form where $p(x) = q(x)/C$ for some $C>0$).
Such problems occur in Bayesian statistics where $P$ represents a posterior distribution, and the integral represents a posterior expectation of interest.
Markov chain Monte Carlo (MCMC) methods are extensively used for this task but suffer (in terms of accuracy) from `clustering' of the points $\{x_i\}_{i=1}^n$ when $n$ is small.
This observation motivates us to instead consider a range of goal-oriented discrete approximation methods that are designed with un-normalised densities in mind.
   
The problem of discrete approximation of a distribution, given its normalised density, has been considered in detail and relevant methods include quasi-Monte Carlo (QMC) \citep{Dick2010}, kernel herding \citep{Chen2010,Lacoste-julien2015}, support points \citep{mak2016support,mak2017projected}, transport maps \citep{Marzouk2016}, and minimum energy methods \citep{Johnson1990}.
On the other hand, the question of how to proceed with un-normalised densities has been primarily answered with increasingly sophisticated MCMC.

At the same time, recent work had led to theoretically-justified measures of sample quality in the case of an un-normalised target.
In \cite{Gorham2015,Mackey2016} it was shown that Stein's method can be used to construct discrepancy measures that control weak convergence of an empirical measure to a target.
This was later extended in \cite{Gorham2017} to encompass a family of discrepancy measures indexed by a reproducing kernel.
In the latter case, the discrepancy measure can be recognised as a maximum mean discrepancy \cite{Smola2007}.
As such, one can consider discrete approximation as an optimisation problem in a Hilbert space and attempt to optimise this objective with either a greedy or a conditional gradient method.
The resulting method -- {\it Stein Points} -- and its variants are proposed and studied in this work.

\paragraph{Our Contribution}

This paper makes the following contributions:
\begin{itemize}
\item Two algorithms are proposed for minimisation of the kernel Stein discrepancy \citep[KSD;][]{Chwialkowski2016,Liu2016c,Gorham2017}; a greedy algorithm and a conditional gradient method. 
In each case, a convergence result of the form in Eqn. \ref{eq: weak convergence} is established.
\item Novel kernels are proposed for the KSD, and we prove that, with these kernels, the KSD controls weak convergence of the empirical measure to the target.
In other words, the test functions $h$ for which our results hold constitute a rich set $\mathcal{H}$.
\end{itemize}

\paragraph{Outline}

The paper proceeds as follows.
In Section \ref{sec: background} we provide background, and in Section \ref{sec: methods} we present the approximation methods that will be studied.
Section \ref{sec: results} applies these methods to both simulated and real approximation problems and provides a extensive empirical comparison.
All technical material is contained in Section \ref{sec: theory}, where we derive novel theoretical results for the methods we proposed.
Finally we summarise our findings in Section \ref{sec: conclusion}.
\section{Background} \label{sec: background}

Throughout this section it will be assumed that $X$ is a metric space, and we let $\mathcal{P}(X)$ denote the collection of Borel distributions on $X$.
In this context, weak convergence of the empirical measure to $P$ corresponds to taking the set $\mathcal{H}$ in Eqn. \ref{eq: weak convergence} to be the set $\mathcal{H}_{\text{CB}}$ of functions which are continuous and bounded.
In this work we also consider sets $\mathcal{H}$ that correspond to stronger modes of convergence in $\mathcal{P}(X)$. 

First, in \ref{subsec: disc meas}, we recall how discrepancy measures are constructed.
Then we recall the use of Stein's method in this context in \ref{subsec: ksd}.
Formulae for KSD are presented in \ref{subsec: stein kernels}.

\subsection{Discrepancy Measures} \label{subsec: disc meas}

A \emph{discrepancy} is a quantification of how well the points $\{x_i\}_{i=1}^n$ cover the domain $X$ with respect to the distribution $P$.
This framework will be developed below in reproducing kernel Hilbert spaces \citep[RKHS;][]{Hickernell1998}, but the general theory of discrepancy can be found in \cite{Dick2010}.
Note that we focus on unweighted point sets for ease of presentation, but our discussions and results generalise straightforwardly to point sets that are weighted.

Let $k : X \times X \rightarrow \mathbb{R}$ be the reproducing kernel of a RKHS $\mathcal{K}$ of functions $\mathcal{X} \rightarrow \mathbb{R}$.
That is, $\mathcal{K}$ is a Hilbert space of functions with inner product $\langle \cdot,\cdot\rangle_{\mathcal{K}}$ and induced norm $\| \cdot \|_{\mathcal{K}}$ such that, for all $x \in X$, $k(x,\cdot) \in \mathcal{K}$ and $f(x) = \left\langle f , k(x,\cdot) \right\rangle_{\mathcal{K}}$ whenever $f \in \mathcal{K}$. The Cauchy-Schwarz inequality in $\mathcal{K}$ gives that
\begin{eqnarray*}
\textstyle
\left| \frac{1}{n}\sum_{i=1}^n f(x_i) - \int f \mathrm{d}P \right| 
& \leq & 
\textstyle
\|f\|_{\mathcal{K}} \; D_{\mathcal{K},P}\left(\{x_i\}_{i=1}^n\right) 
\end{eqnarray*}
where the final term
\begin{eqnarray*}
\textstyle
D_{\mathcal{K},P}\left(\{x_i\}_{i=1}^n \right) \; := \; \left\| \frac{1}{n} \sum_{i=1}^n k(x_i,\cdot) - \int k(x,\cdot) \mathrm{d}P(x) \right\|_{\mathcal{K}}  \nonumber
\end{eqnarray*}
is the canonical discrepancy measure for the RKHS. 
The Bochner integral $k_P := \int k(x,\cdot) \mathrm{d}P(x) \in \mathcal{K}$ is known as the \emph{mean embedding} of $P$ into $\mathcal{K}$ \citep{Smola2007}.
Thus, if $\mathcal{H} = B(\mathcal{K}) := \{f \in \mathcal{K} : \|f\|_{\mathcal{K}} \leq 1\}$ is the unit ball in $\mathcal{K}$, then $D_{\mathcal{K},P}\left(\{x_i\}_{i=1}^n\right) \rightarrow 0$ implies the convergence result in Eqn. \ref{eq: weak convergence}.

The RKHS framework is now standard for QMC analysis \citep{Dick2010}.
Its popularity derives from the fact that, when both $k_P$ and $k_{P,P} :=
\int k_P \mathrm{d}P$ are explicit, the canonical discrepancy measure is
also explicit:
\begin{eqnarray}
\textstyle
D_{\mathcal{K},P}(\{x_i\}_{i=1}^n) & = & \label{eq: Discrep} \\
& & \hspace{-60pt}\textstyle
 \sqrt{ k_{P,P} - \frac{2}{n} \sum_{i=1}^n k_P(x_i) + \frac{1}{n^2} \sum_{i,j=1}^n k(x_i,x_j) } \nonumber
\end{eqnarray}
Table 1 in \cite{Briol2016} collates pairs $(k,P)$ for which $k_P$ and $k_{P,P}$ are explicit.

If $P$ is a posterior distribution, so that $p$ has unknown normalisation
constant, it is unclear how the terms $k_P$ and $k_{P,P}$ can be computed in closed form,
and so similarly for the discrepancy $D_{\mathcal{K},P}$. 
This has so far prevented QMC and related methods such as kernel herding \citep{Chen2010} from being used to compute posterior integrals. 
A solution to this problem can be found in Stein's method, presented next.

\subsection{Kernel Stein Discrepancy} \label{subsec: ksd}

The method of \citet{Stein1972} was introduced as an analytical tool for establishing convergence in distribution of random variables, but its potential for generating and analyzing computable discrepancies was developed in \cite{Gorham2015}.
In what follows, we recall the kernelised version of the \emph{Stein discrepancy}, first presented for an optimally-weighted point set in 2.3.3 of \cite{Oates2017} and later generalised to an arbitrarily-weighted point set in \cite{Chwialkowski2016,Liu2016c,Gorham2017}.

Suppose that $X$ carries the structure of a smooth manifold, and consider a linear differential operator $\mathcal{T}_P$ on $X$, together with a set $\mathcal{F}$ of sufficiently differentiable functions, with the following property:
\begin{eqnarray}
\textstyle
\int \mathcal{T}_P[f] \; \mathrm{d}P & = & 0 \quad \forall f \in \mathcal{F}. \label{eq: stein property}
\end{eqnarray}
Then $\mathcal{T}_P$ is called a \emph{Stein operator} and $\mathcal{F}$ a \emph{Stein set}. In the
kernelised version of Stein's method, the set $\mathcal{F}$ is either an
RKHS $\mathcal{K}$ with reproducing kernel $k : X \times X \rightarrow
\mathbb{R}$, or the product $\mathcal{K}^d$, which contains vector-valued
functions $f = (f_1,\dots,f_d)$ with $f_j \in \mathcal{K}$ and is equipped
with a norm\footnote{For what follows, any vector norm can be used to combine the component norms $\|f_j\|_{\mathcal{K}}$ \citep[Prop. 3]{Gorham2017}.} $\|f\|_{\mathcal{K}^d} = (\sum_{j=1}^d \|f_j\|_{\mathcal{K}}^2
)^{1/2}$.
For the case $\mathcal{F} = \mathcal{K}$, the image of $\mathcal{K}$ under a Stein operator $\mathcal{T}_P$ is denoted $\mathcal{K}_0 = \mathcal{T}_P \mathcal{K}$.
The notation can be justified since, under appropriate regularity assumptions, the set $\mathcal{T}_P \mathcal{K}$ admits structure from the reproducing kernel $k_0(x,x') = \mathcal{T}_P \overline{\mathcal{T}_P} k(x,x')$ \citep{Oates2017}.
Here $\overline{\mathcal{T}_P}$ is the adjoint of the operator $\mathcal{T}_P$ and acts on the second argument $x'$ of the kernel.
If instead $\mathcal{F} = \mathcal{K}^d$, then we suppose that $\mathcal{T}_P f = \sum_{j=1}^d \mathcal{T}_{P,j} f_j$ so that the set $\mathcal{K}_0 = \mathcal{T}_P \mathcal{K}^d$ admits structure from the reproducing kernel $k_0(x,x') = \sum_{j=1}^d \mathcal{T}_{P,j} \overline{\mathcal{T}_{P,j}} k(x,x')$.
In either case, we will call the reproducing kernel $k_0$ of $\mathcal{K}_0$ a \emph{Stein reproducing kernel}.

Stein reproducing kernels possess the useful property that $k_{0,P} = \int k_0(x,\cdot) \mathrm{d}P = 0$ and $k_{0,P,P} = \int k_{0,P} \mathrm{d}P = 0$, so in particular both are explicit.
Thus, if $k_0$ is a Stein reproducing kernel, then Eqn. \ref{eq: Discrep} can be simplified:
\begin{eqnarray}\label{eqn:ksd}
D_{\mathcal{K}_0,P}\left(\{x_i\}_{i=1}^n \right) & = &\textstyle \sqrt{\frac{1}{n^2} \sum_{i,j = 1}^n k_0(x_i,x_j)} .
\end{eqnarray}
We call this quantity a \emph{kernel Stein discrepancy} (KSD).
Next, we exhibit some differential operators for which Eqn. \ref{eq: stein property} is satisfied and Eqn. \ref{eqn:ksd} can be computed.

\subsection{Stein Operators and Their Reproducing Kernels} \label{subsec: stein kernels}

The divergence theorem can be used to construct Stein operators on a manifold.
For $P$ supported on $X = \mathbb{R}^d$, \cite{Oates2017,Gorham2015,Chwialkowski2016,Liu2016c,Gorham2017} considered the \emph{Langevin Stein operator}
\begin{eqnarray}
\mathcal{T}_P f& := &\textstyle \frac{ \nabla \cdot (p f) }{p} \label{eq: langevin}
\end{eqnarray}
where $\nabla \cdot$ is the usual divergence operator and $f \in \mathcal{K}^d$.
Thus, for the Langevin Stein operator, we obtain a Stein reproducing kernel
\begin{eqnarray}
k_0(x,x') & = & \nabla_x \cdot \nabla_{x'} k(x,x') \label{eq: explicit kernel} \\
& & + \nabla_x k(x,x') \cdot \nabla_{x'} \log p(x') \nonumber \\
& & + \nabla_{x'} k(x,x') \cdot \nabla_x \log p(x) \nonumber \\
& & + k(x,x') \nabla_x \log p(x) \cdot \nabla_{x'} \log p(x'). \nonumber
\end{eqnarray}
To evaluate this kernel, the normalisation constant for $p$ is not required.
Other Stein operators for the Euclidean case were developed in \cite{Gorham2016}.
For $P$ supported on a closed Riemannian manifold $X$, \cite{Oates2017b,Liu2017a} proposed the second order Stein operator $\mathcal{T}_P f := \textstyle \frac{1}{p} \nabla \cdot (p\nabla f) $ where $\nabla$ and $\nabla \cdot$ are, respectively, the gradient and divergence operators on the manifold and $f \in \mathcal{K}$.
Other Stein operators for the general case are proposed in the supplement of \cite{Oates2017b}.

The theoretical results in \cite{Gorham2017} established that certain combinations of Stein operator $\mathcal{T}_P$ and base kernel $k$ ensure that KSD controls weak convergence; that is, $D_{\mathcal{K}_0 , P}(\{x_i\}_{i=1}^n) \rightarrow 0$ implies that Eqn. \ref{eq: weak convergence} holds with $\mathcal{H} = \mathcal{H}_{\text{CB}}$.
This important result motivates our next contribution, where numerical optimisation methods are used to select points $\{x_i\}_{i=1}^n$ to approximately minimise KSD.
Theoretical analysis of the proposed methods is reserved for Section \ref{sec: theory}.

\section{Methods} \label{sec: methods}

In this paper, two algorithms to select points $\{x_i\}_{i=1}^n$ are studied in detail.
The first of these is a greedy algorithm, which at each iteration attempts to minimise the KSD, whilst the second is a conditional gradient algorithm, known as \emph{herding}, which also targets the KSD.
In \ref{subsec: greedy} and \ref{subsec: herding} the two algorithms are described, whilst in \ref{subsec: other alg} some alternative approaches are briefly discussed.

\subsection{Greedy Algorithm} \label{subsec: greedy}

The simplest algorithm that we consider follows a greedy strategy, whereby
the first point $x_1$ is taken to be a global maximum of $p$ (an operation
which does not require the normalisation constant) and each subsequent point
$x_n$ is taken to be a global maximum of $D_{\mathcal{K}_0 ,
  P}(\{x_i\}_{i=1}^n)$, with the KSD being viewed as a function of $x_n$ holding $\{x_i\}_{i=1}^{n-1}$ fixed.
Equivalently, at iteration $n > 1$ of the greedy algorithm, we select
\begin{eqnarray}
x_n & \in &\textstyle \argmin_{x \in X} \quad \frac{k_0(x,x)}{2} + \sum_{i=1}^{n-1} k_0(x_i,x) .  \label{eq: greedy update}
\end{eqnarray}
Note that each iteration of the algorithm requires the solution of a global optimisation problem over $X$; in practice we employed a numerical optimisation method, and this choice is discussed in detail in connection with the empirical results in Section \ref{sec: results} and the theoretical results in Section \ref{sec: theory}.

If a user has a budget of at most $n$ points, the greedy algorithm can be run for $n$ iterations and thereafter improved using (block) coordinate descent on the KSD objective to update an existing point $x_i$ instead of introducing a new point. 
The cost of each update is equal to the cost of adding the $n$-th greedy Stein Point.
This budget-constrained variant of the method will be called \emph{Stein Greedy-$n$} in the sequel (see Section \ref{subset: block describe} for more details).

\subsection{Herding Algorithm} \label{subsec: herding}

The definition of discrepancy in Section \ref{subsec: disc meas} suggests that selection of $\{x_i\}_{i=1}^n$ can be elegantly formulated as a single global optimisation problem over $\mathcal{K}_0$.
Let $M(\mathcal{K}_0)$ be the \emph{marginal polytope} of $\mathcal{K}_0$; i.e. the convex hull of the set $\{k_0(x,\cdot)\}_{x \in X}$ \citep{Wainwright2008}.
The mean embedding $Q \mapsto k_Q$, as a map $\mathcal{P}(X) \rightarrow M(\mathcal{K})$, is injective whenever the kernel $k$ is universal and $X$ is compact \citep{Smola2007}, so that in this case $k_Q$ fully characterises $Q$. 
Results in a similar direction for Stein reproducing kernels were established in \citet[][Theorem 2.1]{Chwialkowski2016} and \citet[][Proposition 3.3]{Liu2016c}.
Thus, as $P$ is mapped to $0$ under the embedding, we are motivated to consider non-trivial solutions to
\begin{eqnarray}
\textstyle
\argmin_{f \in M(\mathcal{K}_0)} J(f), \qquad J(f) \; := \; \frac{1}{2} \| f \|_{\mathcal{K}_0}^2  . \label{eq: herding ob}
\end{eqnarray}
As might be expected, the objective function is closely related to KSD; for $f(\cdot) = \frac{1}{n} \sum_{i=1}^n k_0(x_i,\cdot)$ we have $J(f) = \frac{1}{2} D_{\mathcal{K}_0,P}(\{x_i\}_{i=1}^n)^2$.
An iterative algorithm, called \emph{kernel herding}, was proposed in \cite{Chen2010} to solve problems in the form of Eqn. \ref{eq: herding ob}.
This was later shown to be equivalent to a conditional gradient algorithm, the \emph{Frank-Wolfe algorithm}, in \cite{Bach2012}.
The canonical Frank-Wolfe algorithm, which results in an unweighted point set \citep[as opposed to a more general weighted point set;][]{Bach2012}, is presented next.

The first point $x_1$ is again taken to be a global maximum of $p$; this corresponds to an element $f_1 = k_0(x_1,\cdot) \in M(\mathcal{K}_0)$.
Then, at iteration $n > 1$, the convex combination $f_n = \frac{n-1}{n} f_{n-1} + \frac{1}{n} f_n^* \in M(\mathcal{K}_0)$ is constructed where the element $f_n^*$ encodes a direction of steepest descent:
\begin{eqnarray*}
f_n & \in & \textstyle
\argmin_{f \in M(\mathcal{K}_0)} \left\langle f , \mathrm{D}J(f_{n-1}) \right\rangle_{\mathcal{K}_0} ,
\end{eqnarray*}
where $\mathrm{D}J(f)$ is the representer of the Fr\'{e}chet derivative of $J$ at $f$.
Given that minimisation of a linear objective over a convex set can be restricted to the boundary of that set, it follows that $f_n^* = k(x_n,\cdot)$ for some $x_n \in X$.
Thus, at iteration $n > 1$ of the proposed algorithm, we select
\begin{eqnarray} \label{eqn:frank-wolfe-point-sequence}
x_n & \in & \textstyle
\argmin_{x \in X} \quad \sum_{i=1}^{n-1} k_0(x_i,x) 
\end{eqnarray}
to obtain $f_n(\cdot) = \frac{1}{n} \sum_{i=1}^n k_0(x_i,\cdot)$, the embedding of the empirical distribution of $\{x_i\}_{i=1}^n$.
As in the standard kernel herding algorithm of \cite{Chen2010}, each iteration in practice requires the solution of a global optimisation problem over $X$.

Compared to Eqn. \ref{eq: greedy update}, the greedy algorithm is seen to be a regularised version of herding with regulariser $\frac{1}{2} k_0(x,x)$.
The two algorithms coincide if $k_0(x,x)$ is independent of $x$; however, this is typically not true for a Stein reproducing kernel.
The computational cost of either method is $O(n^2)$; thus we anticipate applications in which evaluation of $p(x)$ (and its gradient) constitute the principal computational bottleneck. 
The performance of both algorithms is studied empirically in Section \ref{sec: results} and theoretically in Section \ref{sec: theory}.
In a similar manner to \emph{Stein Greedy-$n$}, a budget-constrained variant of the above method can be considered, which we call \emph{Stein Herding-$n$} in the sequel.

\subsection{Other Algorithms} \label{subsec: other alg}

The output of either of our algorithms will be called \emph{Stein Points}.
These are \emph{extensible} point sequence $S_n = (x_i)_{i=1}^n$, meaning that $S_n$ can be incrementally extended $S_n = (S_{n-1},x_n)$ as required.
Another recently proposed extensible method is the (sequential) minimum energy design (MED) of \cite{Joseph2015,Joseph2017}, here used as a benchmark.

For some problems the number of points $n$ will be fixed in advance and the aim will instead be to select a single optimal point set $\{x_i\}_{i=1}^n$.
This alternative problem demands different methodologies, and a promising method in this direction is Stein variational gradient descent \citep[SVGD-$n$;][]{Liu2016a,Liu2017}.
A natural point set analogue of our approach would be to optimise KSD for $n$ fixed.
This approach was considered for other discrepancy measures in \cite{Oettershagen2017}, where the Newton method was used.
We instead employ our budget-constrained algorithms Stein Greedy-$n$ and Stein Herding-$n$ for this use case.
\section{Results} \label{sec: results}

In this section, the proposed greedy and herding algorithms are empirically assessed and compared.
In \ref{subsec: mixture} a Gaussian mixture problem is studied in detail, whilst in \ref{subsec: GP model} and \ref{subsec: IGARCH}, respectively, the methods are applied to approximate the parameter posterior in a non-parametric regression model and an IGARCH model.
First, in \ref{sec: implementation} we provide details on the experimental protocol.

\subsection{Experimental Protocol} \label{sec: implementation}

Here we describe the parameters and settings that were varied in the experiments that are presented.

\paragraph{Stein Operator}

To limit scope, we focus on the case $X = \mathbb{R}^d$ and always take $\mathcal{T}_P$ to be the Langevin Stein operator in Eqn. \ref{eq: langevin}.

\paragraph{Choice of Kernel}

For the kernel $k$ in Eqn. \ref{eq: explicit kernel} we considered one standard choice -- the inverse multi-quadratic (IMQ) kernel -- together with two novel alternatives:
\begin{enumerate}[label=({$k_\arabic*$})]
\item \label{enum:imq} (IMQ) $k_1(x,x') \; = \; \left(\alpha + \|x - x'\|_2^2 \right)^{\beta}$
\item \label{enum:log-imq}(inverse log) $k_2(x,x') \; = \; \left(\alpha + \log(1 + \|x - x'\|_2^2)\right)^{-1}$
\item \label{enum:imq-score}(IMQ score) \\ $k_3(x,x') \; = \; \left(\alpha + \| \nabla \log p(x) - \nabla \log p(x')\|_2^2\right)^{\beta}$.
\end{enumerate}
In all cases $\alpha > 0$ and $\beta \in (-1,0)$.
To limit scope, in what follows we considered a finite number of judiciously selected configurations for $\alpha,\beta$, though in principle these could be optimised as in \cite{Jitkrittum2017}.
The best set of parameter values was selected for each algorithm and each target distribution, where the possible values were $\alpha \in \{0.1\eta, 0.5\eta, \eta, 2\eta, 4\eta, 8\eta\}$ and $\beta \in \{0.1, 0.3, 0.5, 0.7, 0.9\}$, with $\eta > 0$ problem-dependent (see the Supplement).
The IMQ kernel, together with the Langevin Stein operator, was proven in \citet[][Theorem 8]{Gorham2017} to provide a KSD that controls weak convergence.
Similar results for novel kernels $k_2$ and $k_3$ are established in Section \ref{sec: theory}.

\paragraph{Numerical Optimisation Method}
Any optimisation procedure could be used to (approximately) solve the global optimisation problem embedded in each iteration of the proposed algorithms.
In our experiments, we considered the following numerical methods, for which full details appear in the Section \ref{sup: optimisation methods}.
\begin{enumerate}
\item Nelder-Mead (\verb+NM+): At iteration $n$, parallel runs of Nelder-Mead were employed, initialised at draws from a Gaussian mixture proposal centred on the current point set $\Pi = \frac{1}{n-1}\sum_{i=1}^{n-1}\mathcal{N}(x_i,\lambda I)$ with problem-specific $\lambda > 0$.
\item Monte Carlo (\verb+MC+): The optimisation problem at iteration $n$ was solved over a sample of points drawn from the same Gaussian mixture proposal $\Pi$.
\item Grid search (\verb+GS+): Through brute force, the optimisation problem at iteration $n$ was solved over a regular grid of width $\frac{1}{\sqrt{n}}$. This required $O(n^{-\frac{d}{2}})$ points; if required, the domain was first truncated with a large bounding box.
\end{enumerate}

\paragraph{Performance Assessment}
To obtain a reasonably objective assessment, we focused on the $1$-Wasserstein distance between the empirical measure and $P$:
\begin{eqnarray*}
W_P(\{x_i\}_{i=1}^n) & = &\textstyle
 \sup_{h \in \mathcal{H}_{\text{Lip}}} \left| \frac{1}{n} \sum_{i=1}^n h(x_i) - \int h \mathrm{d}P \right| ,
\end{eqnarray*}
where $\mathcal{H}_{\text{Lip}}$ is the set of all function $h : X \rightarrow \mathbb{R}$ with Lipschitz constant $\text{Lip}(h) \leq 1$.
By replacing $P$ with the empirical measure $P_N = \frac{1}{N}\sum_{i=1}^N
\delta_{y_i}$ for $y_i \stackrel{\text{iid}}{\sim} P$, the expected error
from using $W_{P_N}(\{x_i\}_{i=1}^n)$ in lieu of $W_P(\{x_i\}_{i=1}^n)$
converges at a $N^{-\frac{1}{2}} \log N$ rate for $d=2$ and $N^{-\frac{1}{d}}$ rate for
$d>2$ \cite{fournier2015rate}.
By employing $L_1$-spanners, the approximation $W_{P_N}(\{x_i\}_{i=1}^n)$ can be
computed in $O((n+N)^2 \log^{2d-1}(n+N))$ time \cite{gudmundsson2007small}.
For all reported results, the $\{y_i\}_{i=1}^N$ were obtained by brute-force Monte Carlo methods applied to $P$, with $N$ sufficiently large that approximation error can be neglected.

The computational cost associated to any given method was quantified as the total number $n_{\text{eval}}$ of times either the log-density $\log p$ or its gradient $\nabla \log p$ were evaluated.
This can be justified since in most applications the `parameter to data' map dominates the computational cost associated with the likelihood.

\paragraph{Benchmarks}

Two existing methods were used as a benchmark:
\begin{enumerate}
\item The MED method of \cite{Joseph2015,Joseph2017} relies on numerical optimisation methods to minimise an energy measure $\mathcal{E}_{\delta,P}(\{x_i\}_{i=1}^n)$, adapted to $P$.
This measure has one tuning parameter $\delta \in [1,\infty)$.
See Section \ref{subsec: MED describe} of the Supplement for full detail.
\item The SVGD method of \cite{Liu2016a,Liu2017} performs a version of gradient descent on the Kullback-Leibler divergence, described in Section \ref{subset: SVGD describe} of the Supplement.
\end{enumerate}
To avoid confounding of the empirical results by incomparable algorithm parameters, (1) the collection of numerical optimisation methods used for KSD were also used for MED, and (2) the same collection of kernels $k_1, \ldots, k_3$ was considered for SVGD as was used for KSD.
Note that, apart from standard Monte Carlo, none of the methods considered in these experiments are re-parametrisation invariant.

\subsection{Gaussian Mixture Test} \label{subsec: mixture}

For our first test, we considered a Gaussian mixture model
\begin{eqnarray*}
P & = & \textstyle
\frac{1}{2} \mathcal{N}(\mu_1 , \Sigma_1) + \frac{1}{2} \mathcal{N}(\mu_2 , \Sigma_2)
\end{eqnarray*}
defined on $X = \mathbb{R}^2$.
Full settings for each of the methods considered are detailed in Section \ref{subsec: ad numerics GMM} in the Supplement.
Typical point sets are displayed  over the contours of $P$ for $\mu_1 = (-1.5,0)$, $\mu_2 = (1.5,0)$, $\Sigma_1 = \Sigma_2 = I$ in Figure \ref{fig: GMM points}.
Additionally, point sets for the $n$ point budget-constrained algorithms Stein Greedy-$n$ and Stein Herding-$n$ are presented in Figure \ref{fig: GMM points v2} in the Supplement.
For each of the methods shown in Figures \ref{fig: GMM points} and \ref{fig: GMM points v2}, tuning parameters were varied and the overall performance was captured in Figure \ref{fig: GMM}.
It was observed that for (a-c) the choice of numerical optimisation method was the most influential tuning parameter, with the simpler Monte Carlo-based method being most successful.
The kernels $k_1,k_2$ were seen to perform well, but in (a,b,d) the kernel $k_3$ was sometimes seen to fail.

A subjectively-selected exemplar was extracted for each method, and these `best' results for each method are overlaid in Figure \ref{fig: GMM 2}.
The total number of points was limited to $n = 100$.
In terms of our proposed methods, two qualitative regimes were observed:
(i) For low computational budget $\log n_{\text{eval}} \leq 7$, the standard Monte Carlo method performed best.
(ii) For a larger computational budget $7 < \log n_{\text{eval}}$, greedy Stein points were not out-performed.

Note that KSD and SVGD are based on the log target and its gradient, whilst for MED the target $p(x)$ itself is required.
As a result, numerical instabilities were sometimes encountered with MED.

Next, we turned our attention to two important posterior approximation problems that occur in the real world.

\begin{figure}[t!]
\centering
\includegraphics[width=0.45\textwidth,clip,trim = 0.5cm 6.8cm 0cm 1.6cm]{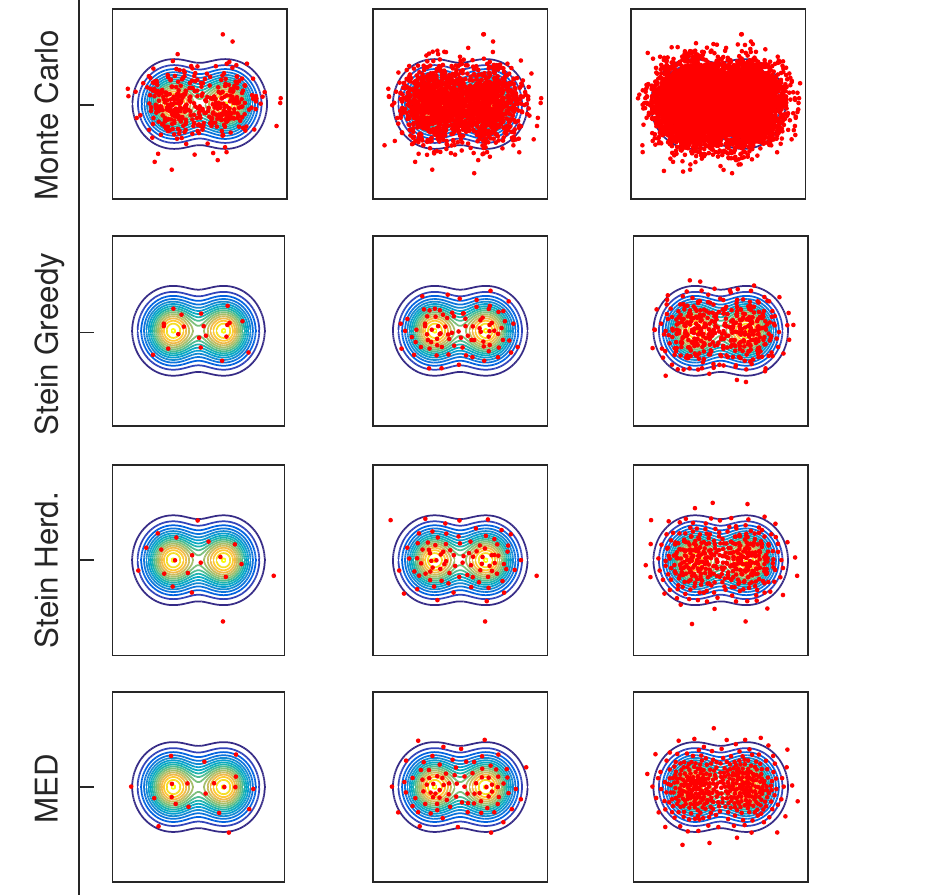}
\includegraphics[width=0.45\textwidth,clip,trim = 0.5cm 0.8cm 0cm 13cm]{figures/gensprpts_2c.pdf}
\caption{Typical point sets obtained in the Gaussian mixture test. [Here the left border of each sub-plot is aligned to the exact value of $\log n_{\mathrm{eval}}$ spent to obtain each point set.]
}
\label{fig: GMM points}
\end{figure}

\begin{figure}[t!]
\centering
\begin{subfigure}[b]{0.23\textwidth}
		\centering
        \includegraphics[width=\textwidth]{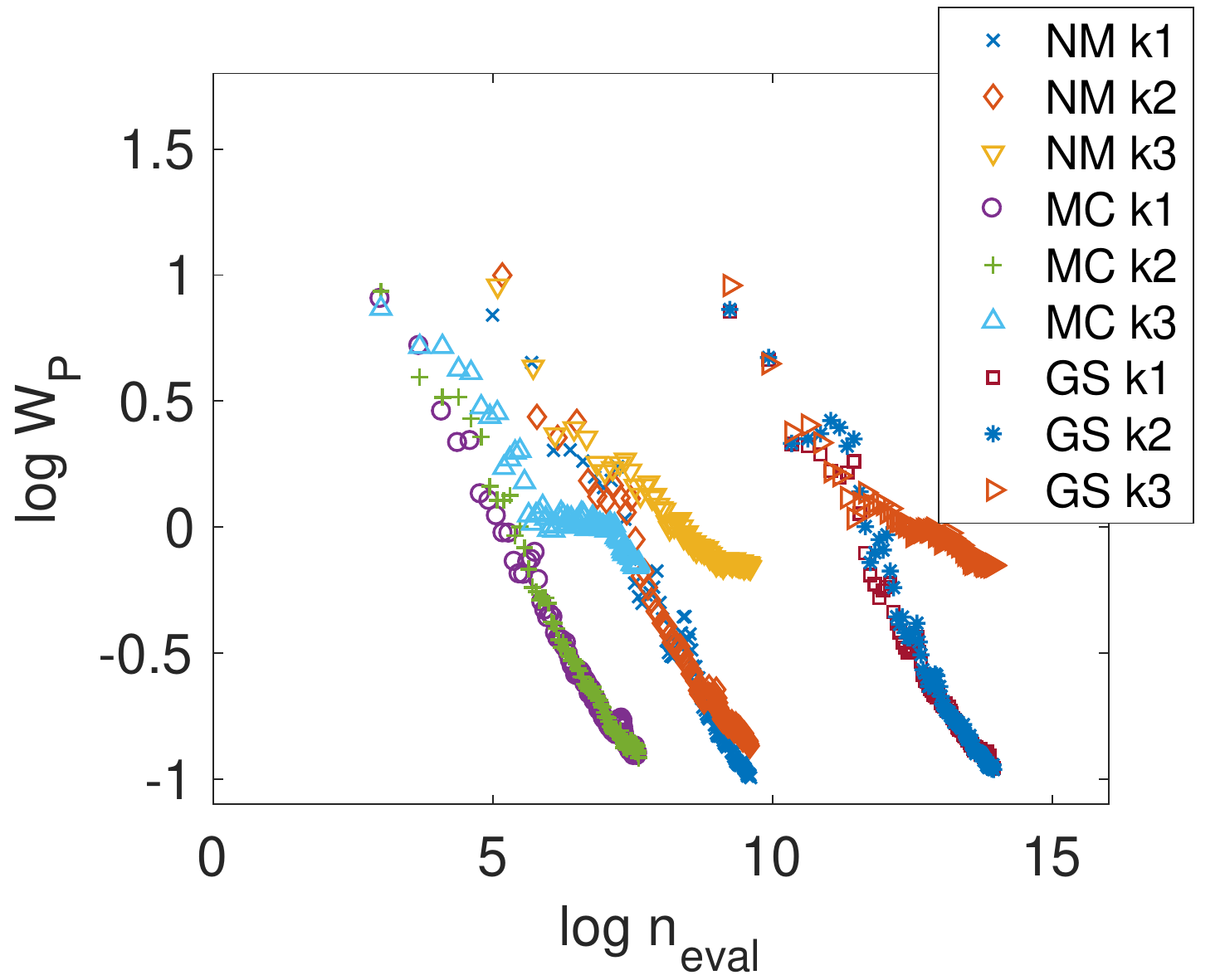}
        \caption{Stein Points (Greedy)}
        \label{fig: GMM greedy}
    \end{subfigure}
\begin{subfigure}[b]{0.23\textwidth}
		\centering
        \includegraphics[width=\textwidth]{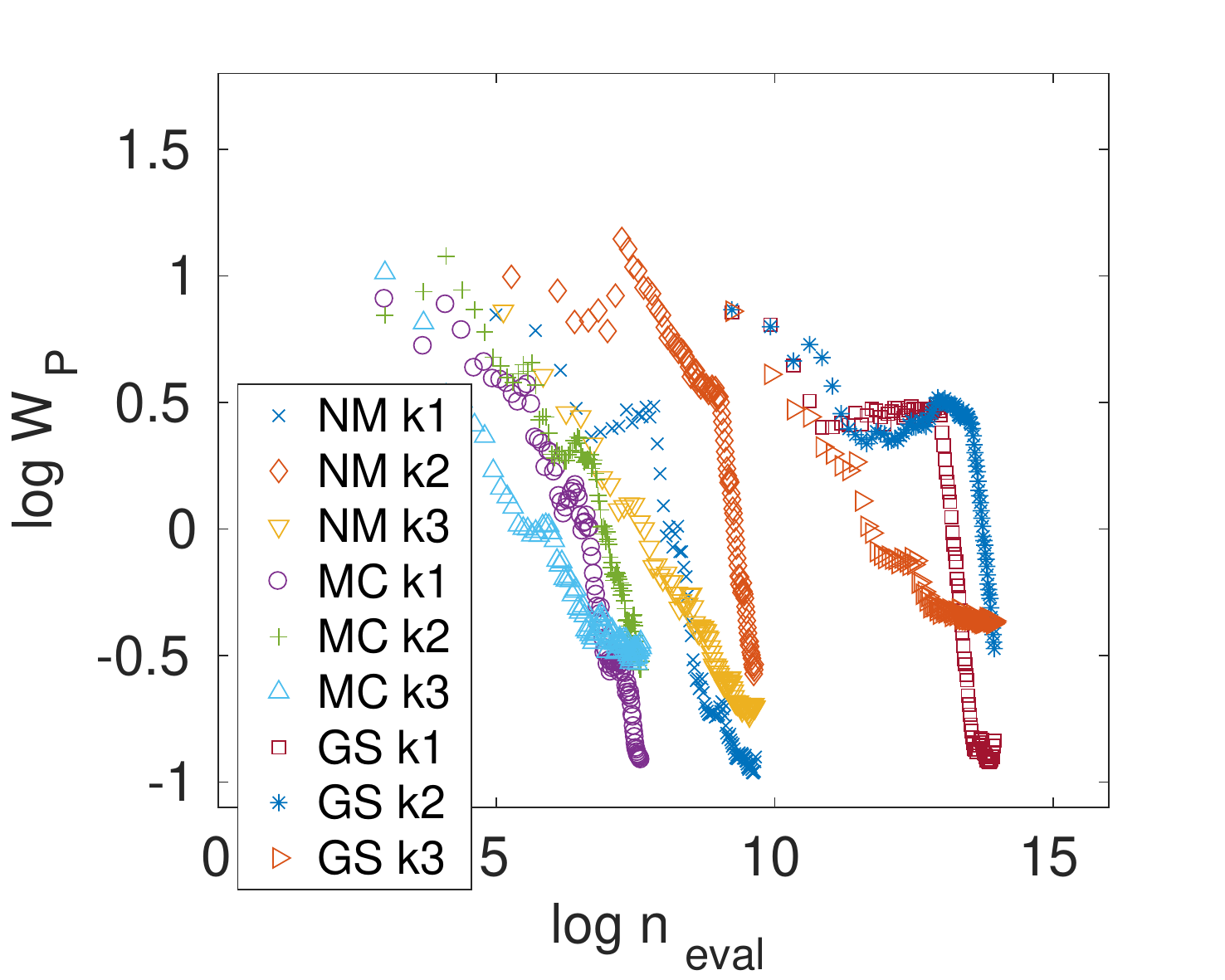}
        \caption{Stein Points (Herding)}
        \label{fig: GMM herding}
    \end{subfigure}

\begin{subfigure}[b]{0.23\textwidth}
		\centering
        \includegraphics[width=\textwidth]{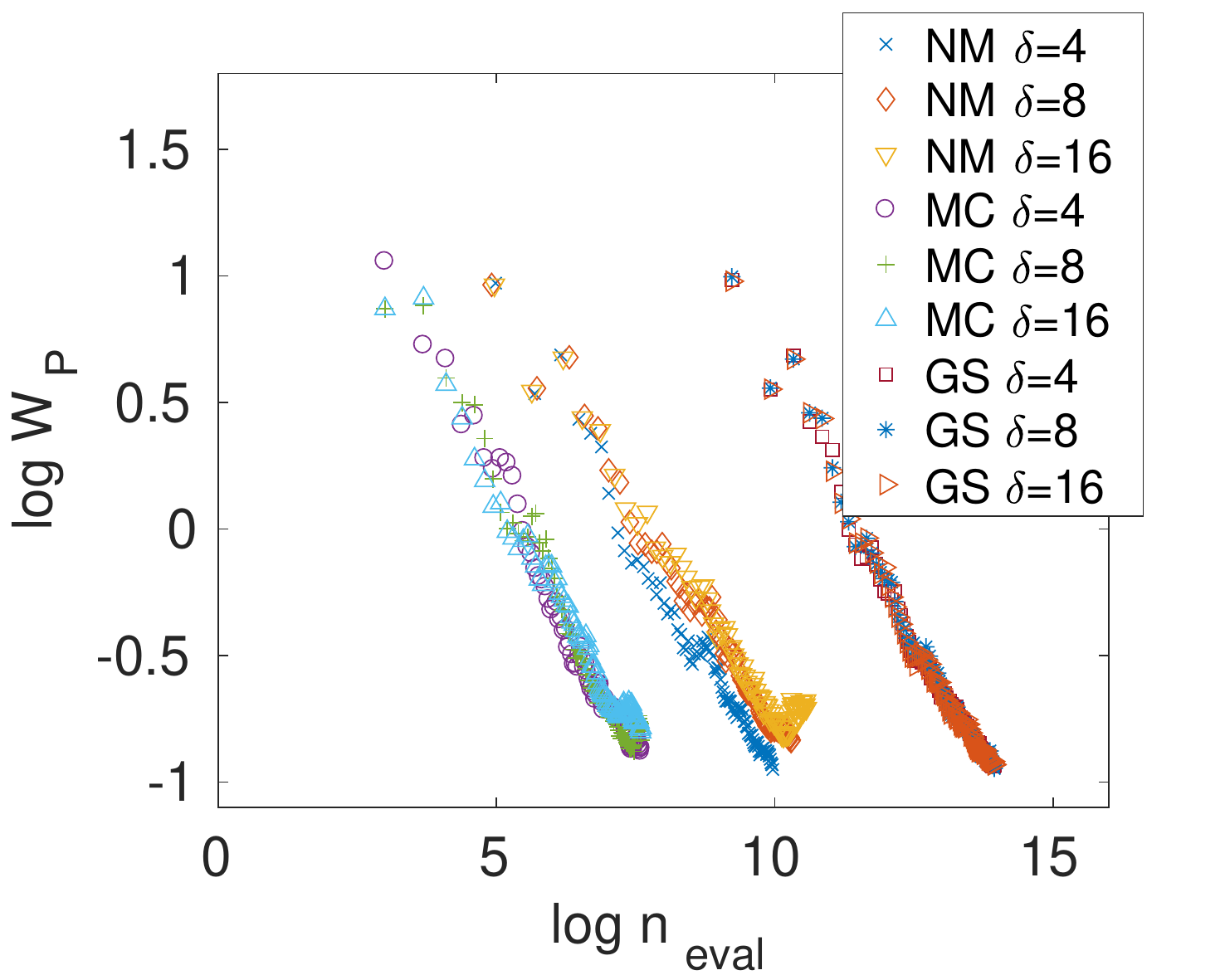}
        \caption{MED}
        \label{fig: GMM MED}
    \end{subfigure}
\begin{subfigure}[b]{0.23\textwidth}
		\centering
        \includegraphics[width=\textwidth]{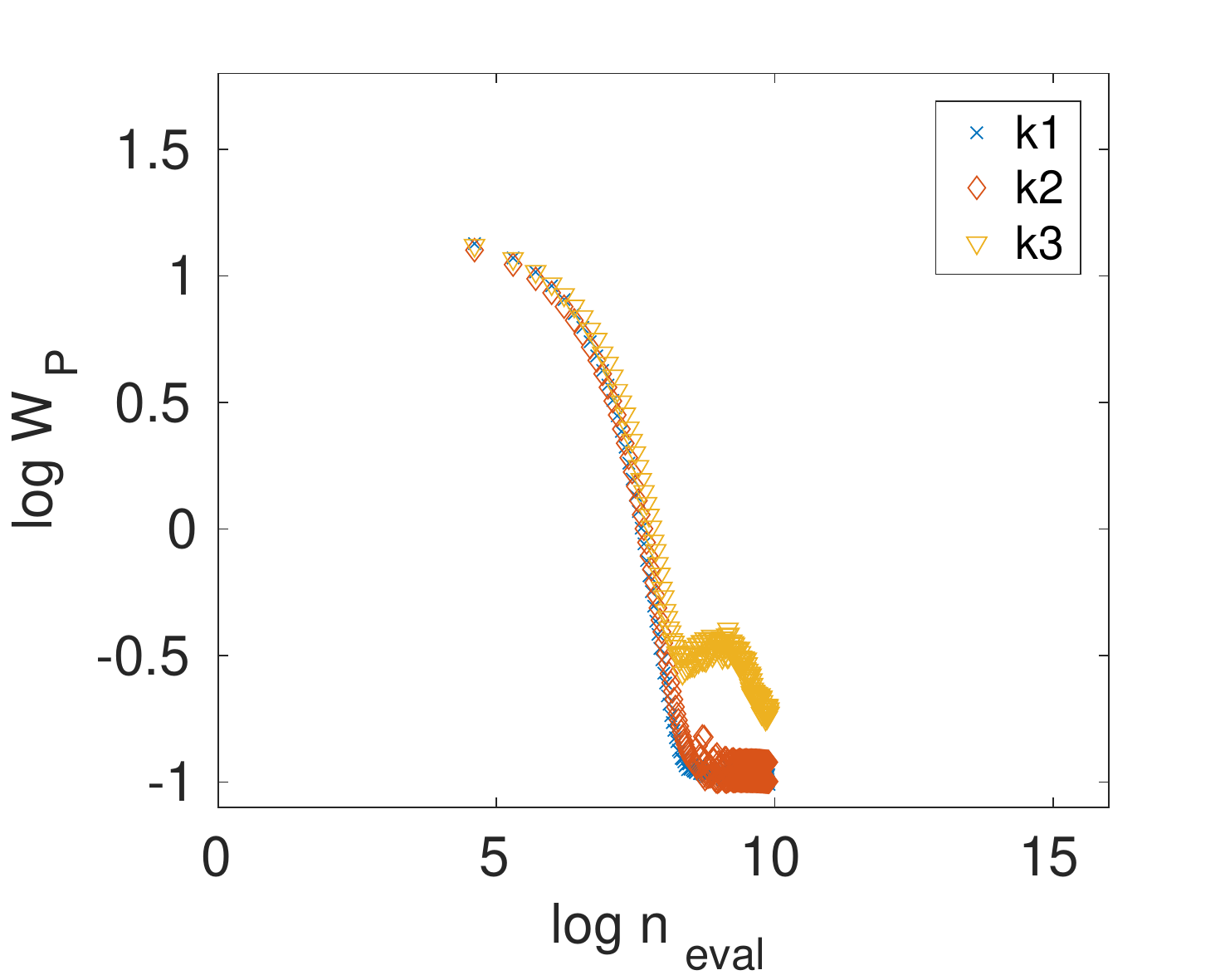}
        \caption{SVGD}
        \label{fig: GMM SVGD}
    \end{subfigure}

\caption{Results for the Gaussian mixture test.
[Here $n = 100$.
x-axis: log of the number $n_{\text{eval}}$ of model evaluations that were used.
y-axis: log of the Wasserstein distance $W_P(\{x_i\}_{i=1}^n)$ obtained.
Kernel parameters $\alpha$, $\beta$ were optimised according to $W_P$ in all cases, with sensitivities reported in Fig. \ref{fig: GP params} of the Supplement.]}
\label{fig: GMM}
\end{figure}

\begin{figure}[t!]
\centering
\includegraphics[width=0.49\textwidth]{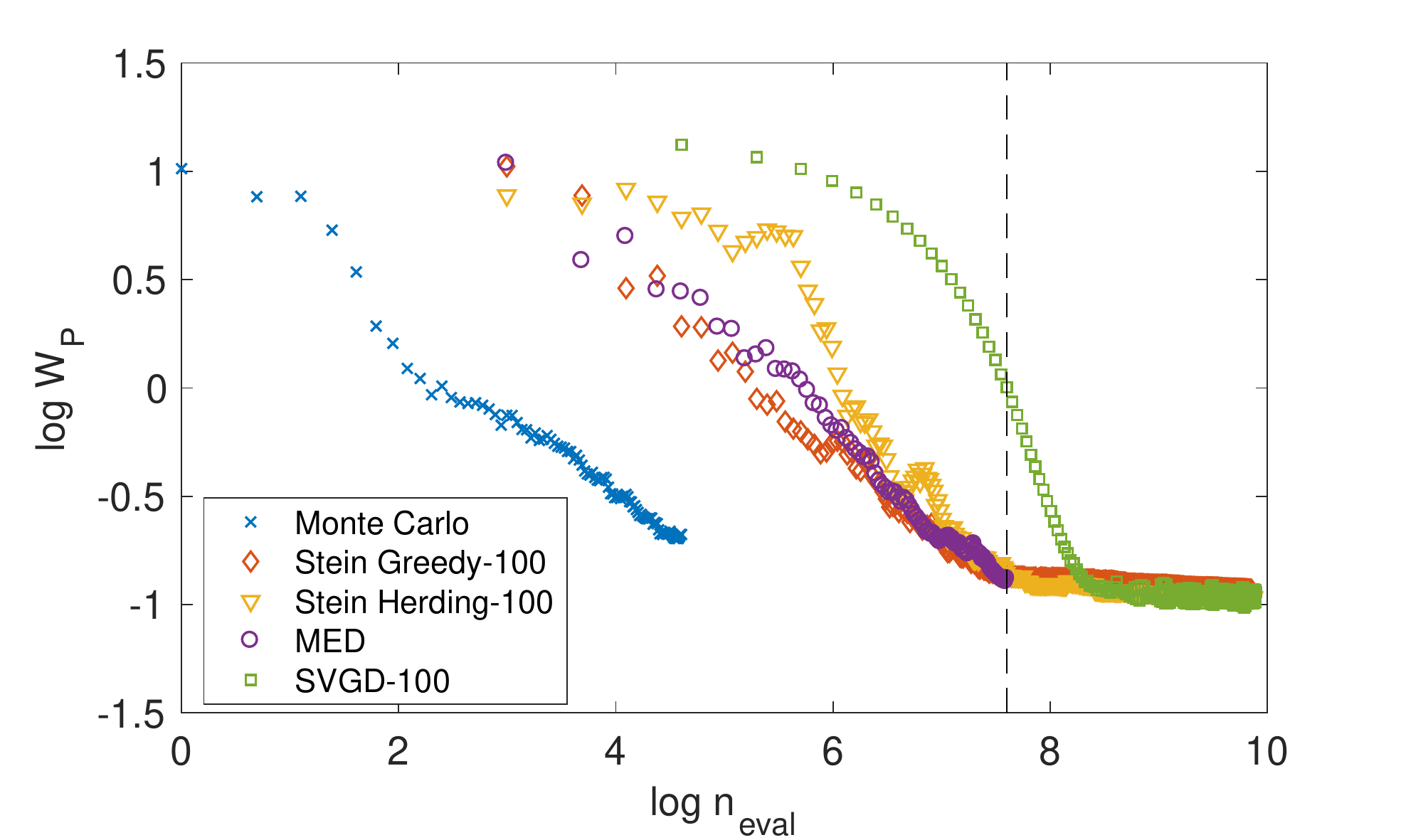}

\caption{Combined results for the Gaussian mixture test.
[Here $n = 100$.
x-axis: log of the number $n_{\text{eval}}$ of model evaluations that were used.
y-axis: log of the the Wasserstein distance $W_P(\{x_i\}_{i=1}^n)$ obtained.
Tuning parameters were selected to minimise $W_P$, as described in the main text.
The dashed line indicates the point at which $n$ Stein Points have been generated; block coordinate descent is performed thereafter to satisfy the $n$ point budget constraint.]}
\label{fig: GMM 2}
\end{figure}

\subsection{Gaussian Process Regression Model} \label{subsec: GP model}

The Gaussian process (GP) model is a popular choice for uncertainty quantification in the non-parametric regression context \citep{Rasmussen2006}.
The data $\mathcal{D} = \{(x_i,y_i)\}_{i=1}^n$ that we considered are from a light detection and ranging (LIDAR) experiment \citep{Ruppert2003}. They consist of 221 realisations of an independent scalar variable $x_i$ (distances travelled before the light is reflected back to its source) and a dependent scalar variable $y_i$ (log-ratios of received light from two laser sources); these were modelled as
$
y_i  =  g(x_i) + \epsilon_i, \text{ for } \epsilon_i \stackrel{\text{i.i.d.}}{\sim} \mathcal{N}(0,\sigma^2)
$
and a known value of $\sigma$.
The unknown regression function $g$ is modelled as a centred GP with covariance function $\text{cov}(x,x') = \theta_1 \exp( - \theta_2 (x - x')^2)$.
The hyper-parameters $\theta_1,\theta_2 > 0$ determine the suitability of the GP model, but appropriate values will be unknown in general.
In this experiment we re-parametrised $\phi_i = \log \theta_i$ and placed a standard multivariate Cauchy prior on $\phi = (\phi_1, \phi_2)$, defined on $X = \mathbb{R}^2$.
The task is thus to approximate the conditional distribution $p(\phi | \mathcal{D})$.
This problem is motivated by the computation of posterior predictive marginal distributions $p(y^* | x^* , \mathcal{D})$ for a new input $x^*$, which is defined as the integral $\int p(y^* | x^* , \phi, \mathcal{D}) p (\phi | \mathcal{D}) \mathrm{d}\phi$.
Note that the density $p(\phi | \mathcal{D})$ can be differentiated, and an explicit formula is provided in \citet[][Eqn. 5.9]{Rasmussen2006}.

For each class of method, `best' tuning parameters were selected and these are presented on the same plot in Figure \ref{fig: GP results}.
In addition, typical point sets provided by each method are presented in Figures \ref{fig: GP points} and \ref{fig: GP points v2} in the Supplement.
MED was not included because the method exhibited severe numerical instability on this task, as earlier discussed.
Results indicated three qualitative regimes where, respectively, Monte Carlo, greedy Stein points and SVGD provided the best performance for fixed cost.

\subsection{IGARCH Model} \label{subsec: IGARCH}

The integrated generalised autoregressive conditional heteroskedasticity (IGARCH) model is widely-used to describe financial time series $(y_t)$ with time-varying volatility $(\sigma_t)$ \citep{Taylor2011}.
The model is as follows:
\begin{eqnarray*}
y_t & = & \sigma_t \epsilon_t, \hspace{40pt} \epsilon_t \stackrel{\text{i.i.d.}}{\sim} \mathcal{N}(0,1) \\
\sigma_t^2 & = & \theta_1 + \theta_2 y_{t-1}^2 + (1-\theta_2) \sigma_{t-1}^2
\end{eqnarray*}
with parameters $\theta = (\theta_1,\theta_2)$, $\theta_1 > 0$ and $0 < \theta_2 < 1$.
The data $y = (y_t)$ that we considered were 2,000 daily percentage returns of the S\&P 500 stock index (from December 6, 2005 to November 14, 2013), and an improper uniform prior was placed on $\theta$.
Thus the task was to approximate the posterior $p(\theta | y)$.
Note that, whilst the domain $X = \mathbb{R}_+ \times (0,1)$ is bounded, for these data the posterior density is negligible on the boundary $\partial X$.
This ensures that Eqn. \ref{eq: stein property} holds essentially to machine precision; see also the discussion in \citet[][Section 3.2]{Oates2018}.
For the IGARCH model, gradients $\nabla \log p(\theta | y)$ can be obtained as the solution of a recursive system of equations for $\partial \sigma_t / \partial \theta_2$.

As before, the `best' performing of each class of method was selected and these are presented on the same plot in Figure \ref{fig: IGARCH results}.
In addition, typical point sets provided by each method are presented in Figures \ref{fig: IGARCH points} and \ref{fig: IGARCH points v2} in the Supplement.
(Numerical instability again prevented results for MED from being obtained.)
Results were consistent with the Gaussian mixture experiment, favouring either Monte Carlo or greedy Stein points depending on the computational budget.

\begin{figure*}[t!]
\centering
\begin{subfigure}[b]{0.49\textwidth}
	\centering
	\includegraphics[width=\textwidth]{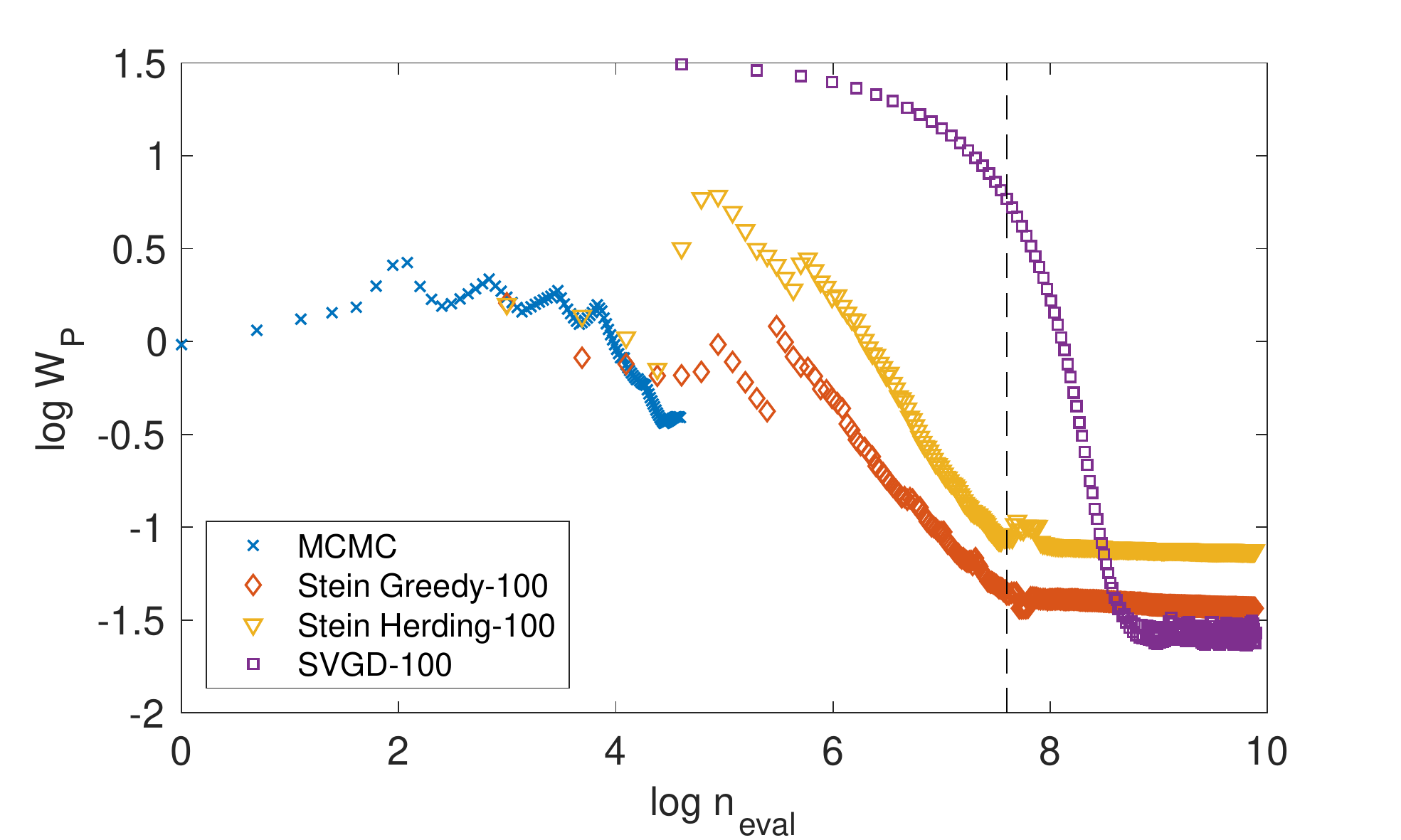}
    \caption{Gaussian Process Test}
	\label{fig: GP results}
\end{subfigure}
\begin{subfigure}[b]{0.49\textwidth}
	\centering
	\includegraphics[width=\textwidth]{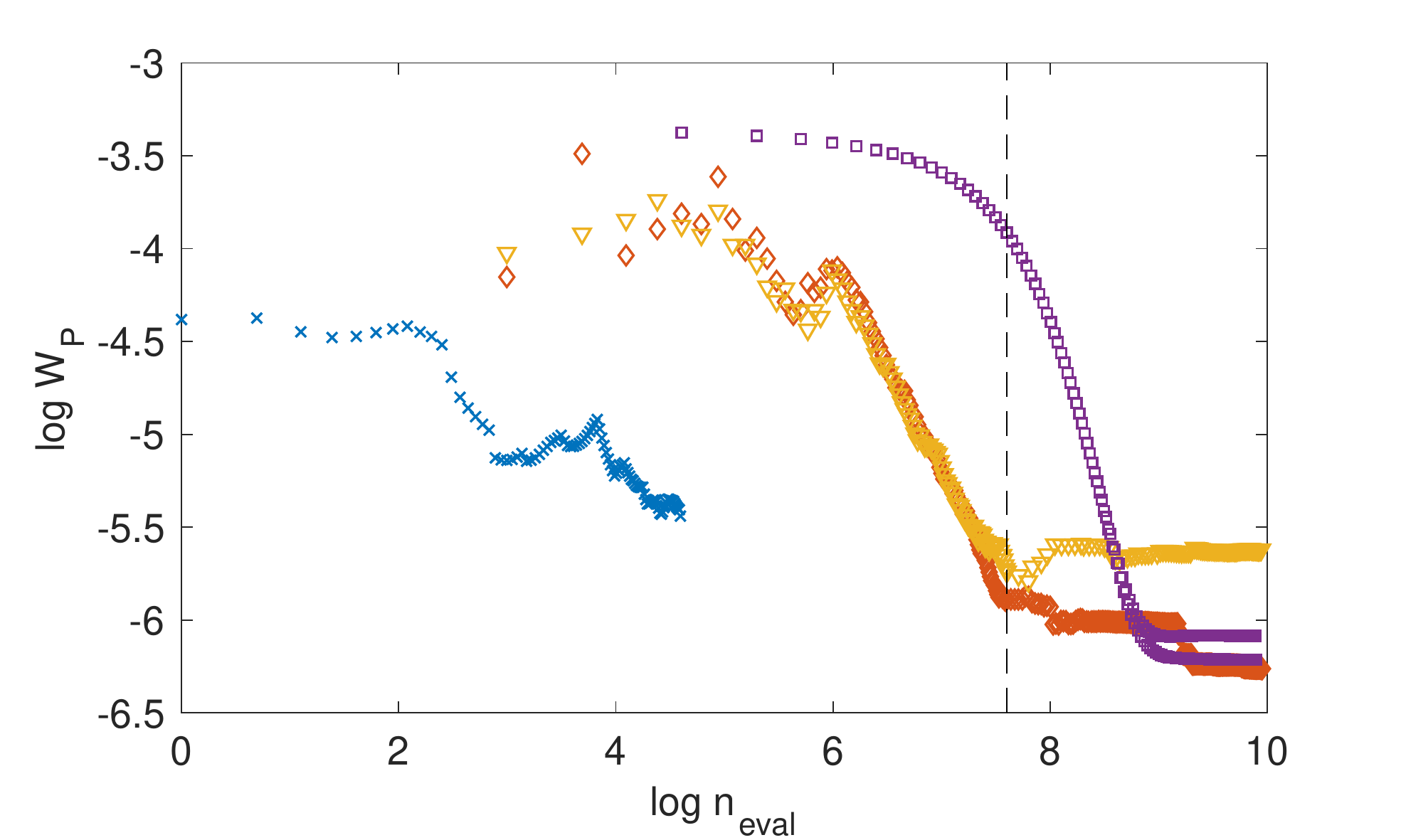}
    \caption{IGARCH Test}
	\label{fig: IGARCH results}
\end{subfigure}

\caption{Combined results for the (a) Gaussian process test and (b) IGARCH test.
[Here $n = 100$.
x-axis: log of the number $n_{\text{eval}}$ of model evaluations that were used.
y-axis: log of the Wasserstein distance $W_P(\{x_i\}_{i=1}^n)$ obtained.
Tuning parameters were selected to minimise $W_P$, as described in the main text.
The dashed line indicates the point at which $n$ Stein Points have been generated; block coordinate descent is performed thereafter to satisfy the $n$ point budget constraint.]}
\end{figure*}

\def\mbi#1{\boldsymbol{#1}} %
\def\mbf#1{\mathbf{#1}}
\def\mbb#1{\mathbb{#1}}
\def\mc#1{\mathcal{#1}}
\def\mrm#1{\mathrm{#1}}
\def\tbf#1{\textbf{#1}}
\def\tsc#1{\textsc{#1}}

\def\indic#1{\mbb{I}\left[{#1}\right]} %
\def\E{\mbb{E}} %
\def\Earg#1{\E\left[{#1}\right]}
\def\Esubarg#1#2{\E_{#1}\left[{#2}\right]}
\def\bigO#1{\mathcal{O}(#1)} %
\def\P{\mbb{P}} %
\def\Parg#1{\P\left({#1}\right)}
\def\Psubarg#1#2{\P_{#1}\left[{#2}\right]}

\def\norm#1{\left\|{#1}\right\|} %
\newcommand{\onenorm}[1]{\norm{#1}_1} %
\newcommand{\twonorm}[1]{\norm{#1}_2} %
\newcommand{\infnorm}[1]{\norm{#1}_{\infty}} %
\newcommand{\opnorm}[1]{\norm{#1}_{op}} %
\newcommand{\fronorm}[1]{\norm{#1}_{F}} %
\def\staticnorm#1{\|{#1}\|} %
\newcommand{\statictwonorm}[1]{\staticnorm{#1}_2} %
\newcommand{\inner}[2]{\langle{#1},{#2}\rangle} %
\newcommand{\binner}[2]{\left\langle{#1},{#2}\right\rangle} %

\newcommand{\qtext}[1]{\quad\text{#1}\quad} 
\def\defeq{\triangleq} %

\newcommand{\grad}{\nabla} %
\newcommand{\Hess}{\nabla^2} %
\newcommand{\lapl}{\triangle} %
\newcommand{\deriv}[2]{\frac{d #1}{d #2}} %
\newcommand{\pderiv}[2]{\frac{\partial #1}{\partial #2}} %

\def\supp#1{\mathrm{supp}({#1})}
\newcommand{\eps}[0]{\epsilon}
\newcommand{\textfrac}[2]{{\textstyle\frac{#1}{#2}}}
\newcommand{\ft}[1]{\mathcal{F}\{#1\}}

\def\Gronwall{Gr\"onwall\xspace}
\def\Holder{H\"older\xspace}
\def\Ito{It\^o\xspace}
\def\Nystrom{Nystr\"om\xspace}
\def\Schatten{Sch\"atten\xspace}
\def\Matern{Mat\'ern\xspace}

\def\reals{\mathbb{R}} %
\def\R{\mathbb{R}}
\def\integers{\mathbb{Z}} %
\def\Z{\mathbb{Z}}
\def\rationals{\mathbb{Q}} %
\def\Q{\mathbb{Q}}
\def\naturals{\mathbb{N}} %
\def\N{\mathbb{N}}
\def\complex{\mathbb{C}} %

\newtheorem{lemma}[theorem]{Lemma}

\def\Lap{\Delta}   %
\def\Re{\textbf{Re}}
\def\Im{\textbf{Im}}

\newcommand{\dt}{\,dt}
\newcommand{\ds}{\,ds}
\newcommand{\dx}{\,dx}
\newcommand{\dw}{\,dw}
\newcommand{\dy}{\,dy}
\newcommand{\dz}{\,dz}
\newcommand{\dr}{\,dr}

\newcommand{\algref}[1]{Algorithm~{\ref{alg:#1}}}
\newcommand{\appref}[1]{Appendix~{\ref{sec:#1}}}
\newcommand{\assumpref}[1]{Assumption~{\ref{assump:#1}}}
\newcommand{\assumpsref}[1]{Assumptions~{\ref{assump:#1}}}
\newcommand{\assumpssref}[1]{{\ref{assump:#1}}}
\newcommand{\chapref}[1]{Chapter~{\ref{chap:#1}}}
\newcommand{\cororef}[1]{Corollary~{\ref{cor:#1}}}
\newcommand{\conjref}[1]{Conjecture~{\ref{conj:#1}}}
\newcommand{\eqnref}[1]{\eqref{eqn:#1}}
\newcommand{\exref}[1]{Example~{\ref{ex:#1}}}
\newcommand{\figref}[1]{Figure~{\ref{fig:#1}}}
\newcommand{\lemref}[1]{Lemma~{\ref{lem:#1}}}
\newcommand{\lemsref}[1]{Lemmas~{\ref{lem:#1}}}
\newcommand{\lemssref}[1]{{\ref{lem:#1}}}
\newcommand{\defref}[1]{Definition~{\ref{def:#1}}}
\newcommand{\secref}[1]{Section~{\ref{sec:#1}}}
\newcommand{\secsref}[1]{Sections~{\ref{sec:#1}}}
\newcommand{\secssref}[1]{{\ref{sec:#1}}}
\newcommand{\propref}[1]{Proposition~{\ref{prop:#1}}}
\newcommand{\tabref}[1]{Table~{\ref{tab:#1}}}
\newcommand{\thmref}[1]{Theorem~{\ref{thm:#1}}}
\newcommand{\thmsref}[1]{Theorems~{\ref{thm:#1}}}
\newcommand{\thmssref}[1]{{\ref{thm:#1}}}
\newcommand{\algreflow}[1]{Algorithm~\lowercase{\ref{alg:#1}}}
\newcommand{\appreflow}[1]{Appendix~\lowercase{\ref{sec:#1}}}
\newcommand{\assumpreflow}[1]{Assumption~\lowercase{\ref{assump:#1}}}
\newcommand{\chapreflow}[1]{Chapter~\lowercase{\ref{chap:#1}}}
\newcommand{\cororeflow}[1]{Corollary~\lowercase{\ref{cor:#1}}}
\newcommand{\figreflow}[1]{Figure~\lowercase{\ref{fig:#1}}}
\newcommand{\lemreflow}[1]{Lemma~\lowercase{\ref{lem:#1}}}
\newcommand{\defreflow}[1]{Definition~\lowercase{\ref{def:#1}}}
\newcommand{\secreflow}[1]{Section~\lowercase{\ref{sec:#1}}}
\newcommand{\secsreflow}[1]{Sections~\lowercase{\ref{sec:#1}}}
\newcommand{\secssreflow}[1]{\lowercase{\ref{sec:#1}}}
\newcommand{\propreflow}[1]{Proposition~\lowercase{\ref{prop:#1}}}
\newcommand{\tabreflow}[1]{Table~\lowercase{\ref{tab:#1}}}
\newcommand{\thmreflow}[1]{Theorem~\lowercase{\ref{thm:#1}}}

\newcommand{\eqdist}{\stackrel{d}{=}}
\newcommand{\todist}{\stackrel{d}{\to}}
\newcommand{\iid}{\stackrel{\mathrm{iid}}{\sim}}
\newcommand{\ind}{\stackrel{\mathrm{ind}}{\sim}}
\newcommand{\eqd}{\stackrel{d}{=}}
\def\KL#1#2{\textnormal{KL}({#1}\Vert{#2})}
\newcommand\indep{\protect\mathpalette{\protect\independenT}{\perp}}
\def\independenT#1#2{\mathrel{\rlap{$#1#2$}\mkern4mu{#1#2}}}

\newcommand{\kset}{\mathcal{K}}
\newcommand{\pset}{\mathcal{P}}

\section{Theoretical Results} \label{sec: theory}
In this section we establish two important forms of theoretical guarantees: 
(1) discrepancy control, i.e., $D_{\mathcal{K}_0,P}(\{x_i\}_{i=1}^n) \to 0$ as $n\to\infty$ for our extensible Stein Point sequences
and
(2) distributional convergence control, i.e., for our kernel choices and appropriate choices of target, $D_{\mathcal{K}_0,P}(\{x_i\}_{i=1}^n) \to 0$ 
implies that the empirical distribution $\frac{1}{n}\sum_{i=1}^n\delta_{x_i}$ converges in distribution to $P$.

\subsection{Discrepancy Control}
Earlier work has shown that, when a kernel is uniformly bounded (i.e., $\sup_{x\in X} k_0(x,x) \leq R^2$), 
the greedy and kernel herding algorithms decrease the associated discrepancy
$D_{\mathcal{K}_0,P}$ at an $O(n^{-\frac{1}{2}})$ rate \citep{Lacoste-julien2015,Jones1992}.
We extend these results to cover all growing, $P$-sub-exponential kernels. 
\begin{definition}[$P$-sub-exponential reproducing kernel]
We say a reproducing kernel $k_0$ is \emph{$P$-sub-exponential} if
\[
\Psubarg{Z\sim P}{ k_0(Z,Z) \geq t } \leq c_1 e^{-c_2t}
\]
for some constants $c_1, c_2 > 0$ and all $t \geq 0$.
\end{definition}
Notably, any uniformly bounded reproducing kernel is $P$-sub-exponential, and, when $P$ is a sub-Gaussian distribution, 
any kernel with at most quadratic growth (i.e., $k_0(x,x) = O(\twonorm{x}^2)$) is also $P$-sub-exponential.
Our first result, proved in \secref{herding-convergence-proof}, shows that if we truncate the search domain suitably in each step, Stein Herding decreases the discrepancy at an $O(\sqrt{\log(n)/n})$ rate.
This result holds even if each point $x_i$ is selected suboptimally with error $\delta/2$.
This extra degree of freedom allows a user to conduct a grid search or a search over appropriately generated random points on each step \citep[see, e.g.,][]{Lacoste-julien2015} and still obtain a rate of convergence.
\begin{theorem}[Stein Herding Convergence]
\label{thm:herding-convergence}
Suppose $k_0$ with $k_{0,P} = 0$ is a $P$-sub-exponential reproducing kernel.
Then there exist constants $c_1, c_2 > 0$ depending only on $k_0$ and $P$ 
such that any point sequence $\{x_i\}_{i=1}^n$ satisfying 
\[
\textstyle
\sum_{i=1}^{j-1} k_0(x_i, x_j) \leq \frac{\delta}{2} 
	+ \displaystyle\min_{x\in X: k_0(x,x)\leq R_j^2}
	\textstyle\sum_{i=1}^{j-1} k_0(x_i, x) 
\]
with ${k_0(x_j,x_j)} \leq R_j^2\in[{2\log(j)/c_2},{2\log(n)/c_2}]$ for each $1 \leq j \leq n$ also satisfies 
$$\textstyle
D_{\mc{K}_0,P}(\{x_i\}_{i=1}^n) 
	\leq e^{\pi/2} \sqrt{\frac{2\log(n)}{c_2n} +\frac{c_1}{n} + \frac{\delta}{n}}.
$$
\end{theorem}
Our next result, proved in \secref{greedy-convergence-proof}, shows that Stein Greedy decreases the discrepancy at an $O(\sqrt{\log(n)/n})$ rate whether we choose to truncate ($R_j < \infty$) or not ($R_j = \infty$).  This highlights an advantage of the Stein Greedy algorithm over Stein Herding: the extra $k_0(x,x)/2$ term acts as a regularizer ensuring that no truncation is necessary. The result also accommodates points $x_i$ selected suboptimally with error $\delta/2$.
\begin{theorem}[Stein Greedy Convergence]
\label{thm:greedy-convergence}
Suppose $k_0$ with $k_{0,P} = 0$ is a $P$-sub-exponential reproducing kernel.
Then there exist constants $c_1, c_2 > 0$ depending only on $k_0$ and $P$ 
such that any point sequence $\{x_i\}_{i=1}^n$ satisfying 
\begin{align*}
\textstyle
\frac{k_0(x_j,x_j)}{2}&+\textstyle\sum_{i=1}^{j-1} k_0(x_i, x_j) \\
	&\leq \textstyle\frac{\delta}{2} + 
	\displaystyle\min_{x\in X: k_0(x,x)\leq R_j^2}\textstyle\frac{k_0(x,x)}{2}+\sum_{i=1}^{j-1} k_0(x_i, x) 
\end{align*}
with $\sqrt{2\log(j)/c_2}\leq R_j \leq\infty$ for each $1 \leq j \leq n$ also satisfies 
$$\textstyle
D_{\mc{K}_0,P}(\{x_i\}_{i=1}^n) 
	\leq e^{\pi/2} \sqrt{\frac{2\log(n)}{c_2n} +\frac{c_1}{n} + \frac{\delta}{n}}.
$$
\end{theorem}

\subsection{Distributional Convergence Control}
To present our final results, we overload notation to define the KSD associated with any probability measure $\mu$:
\begin{eqnarray*}
D_{\mathcal{K}_0,P}\left(\mu \right) & = & \sqrt{\Esubarg{(Z,Z')\sim \mu\times\mu}{k_0(Z,Z')}} .
\end{eqnarray*}
Our original $D_{\mathcal{K}_0,P}$ definition (Eq.~\ref{eqn:ksd}) for a point set $\{x_i\}_{i=1}^n$ is recovered when $\mu$ is the empirical measure $\frac{1}{n}\sum_{i=1}^n \delta_{x_i}$.
We also write $\mu_m \Rightarrow P$ to indicate that a sequence of probability measures $(\mu_m)_{m=1}^\infty$ converges in distribution to $P$.

\citet[][Thm. 8]{Gorham2017} showed that KSDs with IMQ base kernel \ref{enum:imq} and Langevin Stein operator $\mc{T}_P$ control distributional convergence 
 whenever $P$ belongs to the set $\pset$ of distantly dissipative distributions (i.e., $\inner{\grad \log p(x) - \grad \log p(y)}{x-y} \leq -\kappa\twonorm{x-y}^2 + C$ for some $C\geq 0, \kappa >0$) with Lipschitz $\grad \log p$.
Surprisingly, Gaussian, \Matern, and other kernels with light tails do not satisfy this property \citep[][Thm. 6]{Gorham2017}.

Our next theorem establishes distributional convergence control for our newly introduced log inverse kernel \ref{enum:log-imq}.
\begin{theorem}[Log Inverse KSD Controls Convergence]\label{thm:log-inverse-implies-weak-convergence}
Suppose $P \in \pset$.
Consider a Stein reproducing kernel $k_0 = \mc{T}_P\overline{\mathcal{T}_P}k_2$
with Langevin operator $\mc{T}_P$ and base kernel
$k_2(x,x') = (\alpha + \log(1 + \twonorm{x-x'}^2))^{\beta}$ for
$\alpha > 0$ and $\beta < 0$. If $D_{\mathcal{K}_0,P}\left(\mu_m \right) \to
0$, then $\mu_m \Rightarrow P$.
\end{theorem}

Our final theorem, proved in \secref{imq-score-implies-weak-convergence-proof}, guarantees distributional convergence control for the new IMQ score kernel \ref{enum:imq-score} under the additional assumption that $\log p$ is strictly concave.
\begin{theorem}[IMQ Score KSD Controls Convergence]\label{thm:imq-score-implies-weak-convergence}
Suppose $P \in \pset$ has strictly concave log density. 
Consider a Stein reproducing kernel $k_0 = \mc{T}_P\overline{\mathcal{T}_P}k_3$
with Langevin operator $\mc{T}_P$ and base kernel
$k_3(x,x') = (c^2 + \twonorm{\grad\log p(x)-\grad\log p(x')}^2)^{\beta}$ for $c>0$ and
$\beta\in (-1,0)$. If $D_{\mathcal{K}_0,P}\left(\mu_m \right) \to 0$, then $\mu_m \Rightarrow P$.
\end{theorem}

\section{Conclusion} \label{sec: conclusion}

This paper proposed and studied Stein Points, extensible point sequences rooted in minimisation of a KSD, building on the recent theoretical work of \cite{Gorham2017}.
Although we focused on KSD to limit scope, our methods could in fact be applied to any computable Stein discrepancy, even those not based on reproducing kernels \citep[see, e.g.,][]{Gorham2015,Gorham2016}.
Stein Points provide an interesting counterpoint to other recent work focussing on point sequences \citep{Joseph2015,Joseph2017} and point sets \citep{Liu2016a,Liu2017}.
Moreover, when $X$ is a finite set $\{y_i\}_{i=1}^N$ (e.g., an inexpensive initial point set generated by MCMC), Stein Points provide a compact and convergent approximation to the optimally weighted probability measure $\sum_{i=1}^N w_i\delta_{y_i}$ with minimum KSD (see Section \ref{subsec: fixed set} for more details).

Theoretical results were provided which guarantee the asymptotic correctness of our methods.
However, we were only able to establish an $O(\sqrt{\log(n)/n})$ rate, which leaves a theoretical gap between the faster convergence that was sometimes empirically observed.
Relatedly, the $O(n^2)$ computational cost could be reduced to $O(n)$ by using finite-dimensional kernels \citep[see, e.g.,][]{Jitkrittum2017}, but the associated distributional convergence control results must first be developed.

Our experiments were relatively comprehensive, but we did not consider other Stein operators, nor higher-dimensional or non-Euclidean manifolds $X$.
Related methods not considered in this work include those based on optimal transport \cite{Marzouk2016} and self-avoiding particle-based samplers \citep{Robert2003}.
The comparison against these methods is left for future work.

\section*{Acknowledgements}
WYC was supported by the ARC Centre of Excellence for Mathematical and Statistical Frontiers.
FXB was supported by EPSRC [EP/L016710/1, EP/R018413/1].
CJO was supported by the Lloyd's Register Foundation programme on data-centric engineering at the Alan Turing Institute, UK.
This material was based upon work partially supported by the National Science Foundation under Grant DMS-1127914 to the Statistical and Applied Mathematical Sciences Institute. Any opinions, findings, and conclusions or recommendations expressed in this material are those of the author(s) and do not necessarily reflect the views of the National Science Foundation.
\nocite{langley00}

\bibliography{bibliography}
\bibliographystyle{icml2018}

\newpage
\onecolumn
\appendix
{\bf \Large Supplement}

This electronic supplement is organised as follows:
In Section \ref{sup: proof sec} proofs for the theoretical results in the main text are provided.
In Section \ref{sup: benchmark} we provide details for the two existing methods (MED, SVGD) that formed our experimental benchmark.
Then, in Section \ref{sup: additional numerics}, we provide additional numerical results that elaborate on those reported in the main text.

\paragraph{Code}

Code to reproduce these experiments is available from:
\begin{center}
\url{github.com/wilson-ye-chen/stein_points}
\end{center}

\section{Proof of Theoretical Results in the Main Text} \label{sup: proof sec}

\subsection{Proofs of \thmsref{herding-convergence} and \thmssref{greedy-convergence}: Stein Herding and Stein Greedy Convergence}
We will show that both \thmref{herding-convergence} and \thmref{greedy-convergence} follow from the following unified Stein Point convergence result, proved in \secref{stein-point-convergence-proof}.
\begin{theorem}[Stein Point Convergence]
\label{thm:stein-point-convergence}
Suppose $k_0$ with $k_{0,P} = 0$ is a $P$-sub-exponential reproducing kernel.
Then there exist constants $c_1, c_2 > 0$ depending only on $k$ and $P$ 
such that any point sequence $\{x_i\}_{i=1}^n$ satisfying 
\begin{align*}
\frac{k_0(x_j,x_j)}{2}&+\displaystyle\sum_{i=1}^{j-1} k_0(x_i, x_j) 
	\leq \frac{\delta}{2} + \frac{S_j^2}{2}
	+ \displaystyle\min_{x\in X: k_0(x,x)\leq S_j^2}\displaystyle\sum_{i=1}^{j-1} k_0(x_i, x) 
\end{align*}
with $S_j \in [\sqrt{2\log(j)/c_2}, \sqrt{2\log(n)/c_2}]$ for each $1 \leq j \leq n$ and $\delta \geq 0$ also satisfies 
$$\textstyle
D_{\mc{K}_0,P}(\{x_i\}_{i=1}^n) 
	\leq e^{\pi/2} \sqrt{\frac{2\log(n)}{c_2n} +\frac{c_1}{n} + \frac{\delta}{n}}.
$$
\end{theorem}

\subsubsection{Proof of \thmref{herding-convergence}: Stein Herding Convergence}
\label{sec:herding-convergence-proof}
Instantiate the constants $c_1,c_2>0$ from \thmref{stein-point-convergence}, and 
consider any point sequence $\{x_i\}_{i=1}^n$ satisfying 
\[
\sum_{i=1}^{j-1} k_0(x_i, x_j) \leq \frac{\delta}{2} 
	+ \displaystyle\min_{x\in X: k_0(x,x)\leq R_j^2}
	\textstyle\sum_{i=1}^{j-1} k_0(x_i, x) 
\]
with ${k_0(x_j,x_j)} \leq R_j^2\in[{2\log(j)/c_2},{2\log(n)/c_2}]$.
We immediately have
\begin{align*}
\frac{k_0(x_j,x_j)}{2}&+\displaystyle\sum_{i=1}^{j-1} k_0(x_i, x_j) 
	\leq \frac{\delta}{2} + \frac{R_j^2}{2}
	+ \displaystyle\min_{x\in X: k_0(x,x)\leq R_j^2}\displaystyle\sum_{i=1}^{j-1} k_0(x_i, x)
\end{align*}
so the desired conclusion follows from \thmref{stein-point-convergence}.

\subsubsection{Proof of \thmref{greedy-convergence}: Stein Greedy Convergence}
\label{sec:greedy-convergence-proof}
Instantiate the constants $c_1,c_2>0$ from \thmref{stein-point-convergence}, and 
consider any point sequence $\{x_i\}_{i=1}^n$ satisfying 
\begin{align*}
\frac{k_0(x_j,x_j)}{2}&+\displaystyle\sum_{i=1}^{j-1} k_0(x_i, x_j) 
	\leq \frac{\delta}{2} + 
	\displaystyle\min_{x\in X: k_0(x,x)\leq R_j^2}\frac{k_0(x,x)}{2}+\displaystyle\sum_{i=1}^{j-1} k_0(x_i, x) 
\end{align*}
with $S_j = \sqrt{2\log(j)/c_2}\leq R_j \leq\infty$ for each $1 \leq j \leq n$.
We immediately have
\begin{align*}
\frac{k_0(x_j,x_j)}{2}&+\displaystyle\sum_{i=1}^{j-1} k_0(x_i, x_j) 
	\leq \frac{\delta}{2} + 
	\displaystyle\min_{x\in X: k_0(x,x)\leq S_j^2}\frac{k_0(x,x)}{2}+\displaystyle\sum_{i=1}^{j-1} k_0(x_i, x) 
	\leq \frac{\delta}{2} + \frac{S_j^2}{2}
	+ \displaystyle\min_{x\in X: k_0(x,x)\leq S_j^2}\displaystyle\sum_{i=1}^{j-1} k_0(x_i, x),
\end{align*}
so the desired conclusion follows from \thmref{stein-point-convergence}.

\subsubsection{Proof of \thmref{stein-point-convergence}: Stein Point Convergence}
\label{sec:stein-point-convergence-proof}
Our high-level strategy is to show that, when $k_0$ is $P$-sub-exponential, 
optimizing over a suitably truncated search space on each step is sufficient to optimize the discrepancy globally.
To obtain an explicit rate of convergence, we adapt the greedy approximation error analysis of \citet{Jones1992}, which applies to uniformly bounded kernels.
We begin by fixing any sequence of truncation levels $(S_j)_{j=1}^\infty$ with each $S_j \in [0,\infty)$, defining the truncation sets $B_j = \{ x\in X: k_0(x,x) \leq S_j^2\}$,
and letting $\mathcal{M}_j$ denote the convex hull
of $\{ k_0(x,\cdot) \}_{x\in B_j}$.
Next we identify a truncation-optimal $h_j \in \argmin_{f \in \mathcal{M}_j} J(f)$.
Now, fix any point sequence $\{x_i\}_{i=1}^n$ satisfying 
\begin{align*}
\frac{k_0(x_j,x_j)}{2}&+\displaystyle\sum_{i=1}^{j-1} k_0(x_i, x_j) 
	\leq \frac{\delta}{2} + \frac{S_j^2}{2}
	+ \displaystyle\min_{x\in X: k_0(x,x)\leq S_j^2}\displaystyle\sum_{i=1}^{j-1} k_0(x_i, x)
\end{align*}
for some approximation level $\delta \geq 0$ and each $1 \leq j \leq n$.
In the remainder, we will recursively bound the discrepancy of this point sequence in terms of each $S_j$ and $\norm{h_j}_{\kset_0}$, bound each $\norm{h_j}_{\kset_0}$ in terms of $S_j$ using the $P$-sub-exponential tails of $k_0$, and show that an appropriate setting of each $S_j$ delivers the advertised claim.

\paragraph{Bounding discrepancy}
For each $j$, let $f_j = \frac{1}{j}\sum_{i=1}^j k_0(x_i,\cdot)$ and $\eps_j = D_{\mc{K}_0,P}(\{x_i\}_{i=1}^j) = \| f_j\|_{\mc{K}_0}$.
By Cauchy-Schwarz and the arithmetic-geometric mean inequality, we have the estimates
\begin{align*}
n^2 \eps_n^2 - \delta 
	&= k_0(x_n,x_n) + (n-1)^2 \eps_{n-1}^2 + 2(n-1)f_{n-1}(x_n) - \delta \\
	&\leq S_n^2 + (n-1)^2 \eps_{n-1}^2 + 2(n-1)\min_{x \in B_n} f_{n-1}(x)\\
	&= S_n^2 + (n-1)^2 \eps_{n-1}^2 + 2(n-1)\inf_{f \in \mc{M}_n} \inner{f}{f_{n-1}}_{\mc{K}_0}\\
	&\leq S_n^2 + (n-1)^2 \eps_{n-1}^2 + 2(n-1)\inner{h_n}{f_{n-1}}_{\mc{K}_0}\\
	&\leq S_n^2 + (n-1)^2 \eps_{n-1}^2 + n^2\norm{h_n}_{\mc{K}_0}^2 + \frac{(n-1)^2}{n^2}\eps_{n-1}^2.
\end{align*}
Unrolling the recursion, we obtain
\begin{align*}
n^2 \eps_n^2 
	&\leq \sum_{i=0}^{n-1} (S_{n-i}^2 + \delta + \norm{h_{n-i}}_{\mc{K}_0}^2 (n-i)^2) \prod_{j=1}^i (1+1/(n-j+1)^2).
\end{align*}
Moreover, the products in this expression are uniformly bounded in $i$ as
\[
\log(\prod_{j=1}^i (1+1/(n-j+1)^2))
	= \sum_{j=1}^i\log( (1+1/(n-j+1)^2))
	\leq \int_0^\infty \log(1+1/x^2)\dx = \pi.
\]
Therefore, 
\begin{align*}
n^2 \eps_n^2 
	& \leq e^\pi \sum_{i=1}^{n}S_i^2 + \delta + i^2\norm{h_i}_{\mc{K}_0}^2.
\end{align*}

\paragraph{Bounding $\norm{h_i}_{\kset_0}$}
To bound each $\norm{h_i}_{\kset_0}$, we consider the truncated mean embeddings
\begin{align*}
k_{i}^- &:= \int k_0(x,\cdot) \indic{k_0(x,x) \leq S_i^2} \mathrm{d}P(x) \qtext{and} \\
k_{i}^+ &:= \int k_0(x,\cdot) \indic{k_0(x,x) > S_i^2} \mathrm{d}P(x) = k_P - k_i^-.
\end{align*}
Since $k_P = 0$, we have $\norm{k_i^+}_{\mc{K}_0} = \norm{k_i^-}_{\mc{K}_0}$.
Moreover, since $k_i^- \in \mc{M}_i$, we deduce that
\begin{align*}
\norm{h_i}_{\mc{K}_0}^2
	&\leq \norm{k_i^-}_{\mc{K}_0}^2 = \norm{k_i^+}_{\mc{K}_0}^2\\
	&= \iint k_0(x,y) \indic{k_0(x,x) > S_i^2} \mathrm{d}P(x) \indic{k_0(y,y) > S_i^2} \mathrm{d}P(y) \\
	&\leq \left(\int \sqrt{k_0(x,x)}\indic{k_0(x,x) > S_i^2} \mathrm{d}P(x) \right)^2 \\
	&\leq \int k_0(x,x)\indic{k_0(x,x) > S_i^2} \mathrm{d}P(x)
\end{align*}
where the final two inequalities follow by Cauchy-Schwarz and Jensen's inequality.

Let $Y = k_0(Z,Z)$ for $Z \sim P$.
We will bound the tail expectation in the final display by considering the biased random variable $Y^* = k_0(Z^*,Z^*)$
for $Z^*$ with density $\rho(z^*) = \frac{k_0(z^*,z^*) p(z^*)}{\E[Y]}$.
By \citep[Thm. 2.2]{Wainwright17}, since $Y$ is sub-exponential, there exists $c_0> 0$ such that
$\E[e^{\lambda Y}] < \infty$ for all $|\lambda| \leq c_0$.
For any $\lambda \neq 0$ with $|\lambda| \leq c_0/2$, we have, by the relation $x \leq e^x$,
\[
\E[e^{\lambda Y^*}] 
	= \E[e^{\lambda k_0(Z^*,Z^*)}] 
	= \frac{\E[k_0(Z,Z)e^{\lambda k_0(Z,Z)}]}{\E[Y]}
	= \frac{\E[\lambda Y e^{\lambda Y}] }{\lambda \E[Y]}
	\leq \frac{ \E[e^{2\lambda Y}] }{\lambda \E[Y]}
	< \infty.
\]
Hence, by \citep[Thm. 2.2]{Wainwright17}, $Y^*$ is also sub-exponential and satisfies, for some ${\tilde{c}_1}, c_2 > 0$,
$\Parg{ Y^* \geq t } \leq {\tilde{c}_1} e^{-c_2t}$ for all $t > 0$.

Applying this finding to the bounding of $h_i$, we obtain
\[
\norm{h_i}_{\mc{K}_0}^2
	\leq \int k_0(x,x)\indic{k_0(x,x) > S_i^2} \mathrm{d}P(x)
	= {\mathbb{E}[Y]} \int \indic{k_0(x,x) > S_i^2} \rho(x)\dx
	= {\mathbb{E}[Y]} \Parg{Y^* \geq S_i^2}
	\leq c_1 e^{-c_2S_i^2}
\]
{where $c_1 = \tilde{c}_1 \mathbb{E}[Y]$.}
Hence
\begin{align*}
D_{\mc{K}_0,P}(\{x_i\}_{i=1}^n) 
	&\leq e^{\pi/2} \sqrt{\frac{1}{n^2}\sum_{i=1}^{n}S_i^2 + \delta + i^2c_1 e^{-c_2S_i^2}}.
\end{align*}

\paragraph{Setting each $S_i$}
By choosing $S_i \in [\sqrt{2\log(i)/c_2},\sqrt{2\log(n)/c_2}]$ for each $i$ we obtain 
\begin{align*}
D_{\mc{K}_0,P}(\{x_i\}_{i=1}^n) 
	&\leq e^{\pi/2} \sqrt{\frac{1}{n^2}\sum_{i=1}^{n}\frac{2\log(n)}{c_2} + \delta + c_1}
	\leq e^{\pi/2} \sqrt{\frac{2\log(n)}{c_2 n} + \frac{\delta}{n} + \frac{c_1}{n}}.
\end{align*}

\subsection{Proof of \thmref{log-inverse-implies-weak-convergence}: Log Inverse KSD Controls Convergence}
Fix any $\alpha > 0$ and $\beta < 0$.
Our proof will leverage \citep[][Thm. 7]{Gorham2017}. This requires demonstrating two
separate properties for the log inverse kernel: first, the log inverse function $\Phi(z) \defeq (\alpha + \log(1 +
\twonorm{z}^2))^{\beta}$ has a nonvanishing
generalized Fourier transform, and second, whenever
$D_{\mathcal{K}_0,P}\left(\mu_m \right)\to 0$, the measures $\mu_m$ are
uniformly tight. We will repeatedly use the notation $\gamma(r) \defeq (\alpha +
\log(1 + r))^{\beta}$ and $\phi(r) \defeq \gamma(r^2)$ throughout the
proof. Moveover, we will use $\hat{f}$ to denote the
(generalized) Fourier transform of a function $f$, and $V_d$ will represent
the volume of the unit Euclidean ball in $d$ dimensions. Finally, we write $f^{(m)}$ for the $m$-th derivative of any sufficiently differentiable function $f:\reals\to\reals$.

To demonstrate the first property, we begin with the following lemma.

\begin{lemma}[Log Inverse Function Is Completely Monotone]\label{lem:log-inverse-is-completely-monotone}
Fix any $\alpha > 0$ and $\beta < 0$.
The function $\gamma(r) \defeq (\alpha +\log(1 + r))^{\beta}$ is \emph{completely monotone}, i.e.,
$\gamma\in C^{\infty}$ and $(-1)^m \gamma^{(m)}(r) \ge 0$ for all
$m\in\mbb{N}_0$ and all $r \ge 0$,
and hence
the function $k_2:\mbb{R}^d\times\mbb{R}^d\to\mbb{R}$ given by
$k_2(x,x') \defeq \gamma(\twonorm{x-x'}^2)$ is a
kernel function for all dimensions $d\in\mbb{N}$.
\end{lemma}

\begin{proof}
By \citep[Theorem  7.13]{Wendland2004} we know that $\Phi$ is positive semidefinite for
all dimensions $d\in\mbb{N}$ if and only if $\gamma$ is completely monotone. Thus it remains to show that $\gamma$ is
completely monotone.

Since $\alpha > 0$, $\gamma(r) > 0$ for all $r \ge 0$. To verify $(-1)^{m}\gamma^{(m)}(r)\ge 0$ for all $m\ge
1$, we will proceed by induction. Let us suppose that for some $m\ge 1$,
\begin{align}\label{eqn:inverse-log-radial-derivs}
\gamma^{(m)}(r) = (-1)^{m} \sum_{l=1}^{m} c_{l,m} (\alpha + \log(1 + r))^{\beta-l}
(1 + r)^{-m}
\end{align}
where each $c_{l,m}\in\mbb{R}$ is positive. Taking another derivative
yields
\begin{align*}
  \gamma^{(m+1)}(r)
  &= (-1)^{m+1} \sum_{l=1}^{m+1} c_{l,m+1} (\alpha + \log(1 + r))^{\beta-l} (1 + r)^{-m-1},
\end{align*}
where $c_{1, m+1} \defeq m\, c_{1,m}$, $c_{l,m+1} \defeq m\, c_{l, m} +
(l-\beta-1)\, c_{l-1,m}$ for $l > 1$ and $c_{l,m}\defeq 0$ for all $l > m$,
completing the induction step.

As for the base case, notice $\gamma'(r) = \beta(\alpha + \log(1 + r))^{\beta-1} (1 +
r)^{-1}$, which establishes the identity for $l = 1$ by setting
$c_{1,1}\defeq -\beta$. The conclusion of this
proof by induction implies $(-1)^m \gamma^{(m)}(r)\ge 0$ for all $m$ and all
$r \ge 0$. By \eqnref{inverse-log-radial-derivs}, $\gamma\in C^{\infty}$,
establishing the lemma.
\end{proof}

Knowing that $\gamma$ is a completely monotone function, we can now demonstrate
$\hat{\Phi}$ has a nonvanishing generalized Fourier transform.

\begin{lemma}[Log Inverse Function Has Nonvanishing GFT]\label{lem:log-inverse-nonvanishing-gft}
Consider the function $\Phi:\reals^d\to\reals$ given by $\Phi(z) = (\alpha + \log(1 +
\twonorm{z}^2))^{\beta}$ for some $\alpha > 0$ and $\beta < 0$. Its
generalized Fourier transform $\hat{\Phi}(w)$ is radial, nonvanishing, and
continuous for $w\neq 0$. Moreover, $\hat{\Phi}(w)\to 0$ as $\twonorm{w}\to\infty$.
\end{lemma}

\begin{proof}
We will first use induction to prove an intermediate result
that states for any $m\in\mbb{N}_0$,
\begin{align}\label{eqn:lap-expansion}
\Lap^m \Phi(z) = \sum_{(u,v)\in S_m} \tau_{u,v} \twonorm{z}^{2v} \gamma^{(u)}(\twonorm{z}^2)
\end{align}
where $\tau_{u,v} > 0$ are positive reals, $S_m = \{(u, v)\in\mbb{N}_0^2
\,\mid\, v\le u - m, u\le 2m\}$ and $\gamma(r)\defeq (\alpha + \log(1 + r))^{\beta}$.

Note for the base case $m=0$, the claim above for $\Lap^0 \Phi = \Phi$
clearly holds. Now suppose it holds from some $m\in\mbb{N}_0$.
If $A:\mbb{R}^d\to\mbb{R}$ is a function that can decomposed as
$A(z) \defeq f(\twonorm{z}^2)\,g(\twonorm{z}^2)$
where $f,g \in C^{\infty}([0,\infty))$, then we have
\begin{align}\label{eqn:laplacian-product-rule}
\Lap A(z)
  &= \left [2d g'(\twonorm{z}^2) + 4\twonorm{z}^2 g''(\twonorm{z}^2)\right ]
  f(\twonorm{z}^2) +
  \left [2d g(\twonorm{z}^2) + 4\twonorm{z}^2 g'(\twonorm{z}^2)\right ]
  f'(\twonorm{z}^2) +
  4\twonorm{z}^2 g(\twonorm{z}^2) f''(\twonorm{z}^2).
\end{align}
Consider each term in the decomposition of $\Lap^m \Phi(z)$ from the induction
hypothesis. If we let $g(r) = r^{v}$ and $f(r) = \phi^{(u)}(r)$, we see that
each term from \eqnref{laplacian-product-rule} is of the form $\tau_{u,v}'
\twonorm{z}^{2v'} \phi^{(u')}(\twonorm{z}^2)$ where the values for $(u',
v')$ are $(u, v-1), (u, v-1), (u+1, v), (u+1, v), (u+2, v+1)$
respectively. Notice that when $v=0$ or $v=1$, the first or second
derivative of $g$ will be zero and these terms may disappear altogether. Thus all
these tuples will lie in $S_{m+1}$ for any $(u,v)\in S_m$, and so we must have
$\Lap^{m+1} \Phi(z)$ satisfies the induction hypothesis as well, completing
the proof by induction.

Now we can prove the lemma. Suppose $2m \ge d$. Then by the triangle inequality and a
radial substitution \cite{Baker1999},
\begin{align*}
  \int_{\mbb{R}^d} |\Lap^m \Phi(z)|\dz
    \le \sum_{(u,v)\in S_m} \int_{\mbb{R}^d} \tau_{u,v} \twonorm{z}^{2v}
    |\phi^{(u)}(\twonorm{z}^2)| \dz
    = d\,V_d \sum_{(u,v)\in S_m} \int_{0}^{\infty} \tau_{u,v} r^{2v + d - 1}
    |\phi^{(u)}(r^2)| \dr.
\end{align*}
Because $|\phi^{(u)}(r)| = O(r^{-u}\log^{\beta-1}(r))$ as $r\to \infty$ for $u\in\mbb{N}$ by
\eqnref{inverse-log-radial-derivs}, we see that each integrand above is $O(r^{2(v-u)
  + d - 1} \log^{\beta-1}(r))$. But since $v \le u - m$, this will imply that
each integrand is $O(r^{-2m + d - 1} \log^{\beta-1}(r))$, which is integrable for
large $r$ yielding $\Lap^m \Phi \in L^1(\mbb{R}^d)$.

By \citep[Lemma 4.34]{Christmann2008} and the fact that positive definiteness is preserved by summation,
we have $\Lap^m \Phi$ is a positive definite
function. This along with the fact that $\Lap^m \Phi \in L^1(\reals^d)$
allows us to invoke \citep[Theorem 6.11]{Wendland2004} and \citep[Theorem
  6.18]{Wendland2004} to obtain $\widehat{\Lap^m \Phi}$ is continuous,
radial and nonvanishing. Moreover, $\Lap^m \Phi$ belonging to
$L^1(\reals^d)$ implies its Fourier transform belongs to $L^{\infty}(\reals^d)$.
The lemma follows by noticing $\widehat{\Lap^m
  \Phi}(w) = \twonorm{w}^{2m}\hat{\Phi}(w)$, i.e.,
$\hat{\Phi}(w) = \twonorm{w}^{-2m} \widehat{\Lap^m \Phi}(w)$
for all $w\neq 0$.
\end{proof}

We now need to demonstrate the second property to complete the proof of
\thmref{log-inverse-implies-weak-convergence}, but in order to do so, we
first will establish the lemma below.  By
\lemref{log-inverse-nonvanishing-gft}, we know $\hat{\Phi}$ is radial and
thus can write $\hat{\Phi}(w) = \phi_{\wedge}(\twonorm{w})$ for some
continuous function $\phi_{\wedge}:(0,\infty)\to (0, \infty)$. Our first
priority will be to lower bound $\phi_{\wedge}$ near the origin.

\begin{lemma}[Log Inverse GFT Lower Bound]\label{lem:log-inverse-gft-lower-bound}
If $\Phi$ is the log inverse function on $\reals^d$ from
\lemref{log-inverse-nonvanishing-gft},
then $\liminf_{r\to 0^{+}}
r^{d}(\alpha + \log(1+1/r^2))^{-\beta+1}\phi_{\wedge}(r) > 0$ where
$\hat{\Phi}(w) = \phi_{\wedge}(\twonorm{w})$ for all $w\neq 0$.
\end{lemma}

\begin{proof}
First we will show that $\phi_{\wedge}$ is strictly decreasing. Since
$r \mapsto (\alpha + \log(1 + r))^{\beta}$ was shown to be completely
monotone in \lemref{log-inverse-is-completely-monotone}, by \citep[Theorem
  7.14]{Wendland2004} we must have
$\Phi(z) = \int_0^{\infty} e^{-t \twonorm{z}^2}\partial v(t)$
for some finite, non-negative Borel measure $v$ on $[0,\infty)$ that is not
  concentrated at zero. Let $(\varphi_m)_{m=1}^\infty$ be a sequence of Schwartz functions \citep[Definition
  5.17]{Wendland2004} defined on
  $\reals^d$. Then, for each $m$, both $\hat{\varphi}_m$ and $\Phi \hat{\varphi}_m$ are also Schwartz
functions, and thus
\begin{align*}
\int_{\reals^d} \left [\int_0^{\infty} |e^{-t \twonorm{x}^2} \hat{\varphi}_m(x)|
\partial v(t) \right ] \dx
  = \int_{\reals^d} \left [\int_0^{\infty} e^{-t \twonorm{x}^2} |\hat{\varphi}_m(x)|
    \partial v(t) \right ] \dx
  = \int_{\reals^d} \Phi(x) |\hat{\varphi}_m(x)| \dx < \infty,
\end{align*}
as all Schwartz functions are integrable.  This allows us to use Fubini's
theorem in conjunction with Plancherel's Theorem to argue
\begin{align*}
  \int_{\reals^d} \hat{\Phi}(w)\varphi_m(w)\dw
  = \int_{\reals^d} \Phi(x)\hat{\varphi}_m(x)\dx
  &= \int_{\reals^d} \left [\int_0^{\infty} e^{-t \twonorm{x}^2}\partial
    v(t)\right ] \hat{\varphi}_m(x)\dx \\
  &= \int_{\reals^d} \int_0^{\infty} e^{-t \twonorm{x}^2} \hat{\varphi}_m(x) \partial
    v(t) \dx \\
  &= \int_0^{\infty}\int_{\reals^d}  e^{-t \twonorm{x}^2} \hat{\varphi}_m(x)\dx
  \partial v(t) \\
  &= \int_0^{\infty} \int_{\reals^d} (2t)^{d/2} e^{-\frac{1}{4t} \twonorm{w}^2}\varphi_m(w)\dw
  \partial v_+(t) + v_0 \varphi_m(0),
\end{align*}
where we have used the decomposition $v \defeq v_+ + v_0 \delta_0$ for  $v_0
\geq 0$ and $v_+$ non-zero and absolutely continuous with respect to
Lebesgue measure on $[0,\infty)$. Let $\mathcal{B}:\reals^d\to\reals$ be a bump
function, e.g., $\mathcal{B}(x) \defeq \mathcal{Z}^{-1}
\exp\{-1/(1-\twonorm{x}^2)\}\indic{\twonorm{x}<1}$ where $\mathcal{Z}$ is the
normalization constant chosen such that $\int_{\reals^d} \mathcal{B}(x)\dx
= 1$. Then let us define $\varphi_m:\reals^d\to\reals$
via the mapping $\varphi_m(w) \defeq m^d \mathcal{B}(m(w-w_0e_1)) - m^d \mathcal{B}(m(w-w_1e_1))$,
where $0 < w_0 < w_1$ and $e_1\in\reals^d$ is the first standard basis vector. Then $\int_{\reals^d}
\hat{\Phi}(w)\varphi_m(w)\dw \to \hat{\Phi}(w_0e_1) - \hat{\Phi}(w_1e_1) =
\phi_{\wedge}(w_0) - \phi_{\wedge}(w_1)$
since $\hat{\Phi}$ is a continuous in neighborhoods of $w_0e_1$ and $w_1e_1$
\citep[Theorem 5.22]{Wendland2004}.

Because $v_+$ cannot be the zero measure, there must be some finite interval
$[a_0,b_0]\subset (0,\infty)$ such that $v_+([a_0,b_0])>0$.
For each $t > 0$ and $m > \max(\frac{1}{w_0},\frac{2}{w_1 - w_0})$, we have
\[
A_m(t)\defeq \int_{\reals^d} (2t)^{d/2} e^{-\frac{1}{4t} \twonorm{w}^2}\varphi_m(w)\dw
	=  \int_{\reals^d} (2t)^{d/2} (e^{-\frac{1}{4t} \twonorm{w-w_0e_1}^2} - e^{-\frac{1}{4t} \twonorm{w-w_1e_1}^2})m^d \mathcal{B}(mw) \dw
	> 0,
\]
since $\twonorm{w-w_0e_1} < \twonorm{w-w_1e_1}$ when $\twonorm{w} < \min(\frac{w_0 - w_1}{2}, w_0)$.
Using \citep[Theorem 5.22]{Wendland2004} again,
we have $A_m(t)\to (2t)^{d/2}(e^{-\frac{1}{4t}w_0^2} -
e^{-\frac{1}{4t}w_1^2})$ as $m\to\infty$ for any $t > 0$.
Moreover, for all $t \in [a_0,b_0]$ and $m\geq 1$, we have
\begin{align}\label{eqn:schwartz-upper-bound}
|A_m(t)| \leq \int_{\reals^d} (2t)^{d/2} e^{-\frac{1}{4t}
  \twonorm{w-w_0e_1}^2}m^d \mathcal{B}(mw) \dw
  &\leq (2b_0)^{d/2}  \sup_{\twonorm{w} < 1} e^{-\frac{1}{4b_0}
  \twonorm{w-w_0e_1}^2} < \infty.
\end{align}
Hence, the dominated convergence theorem allows us to exchange the limit over $m$ and integral
over $t$ below to conclude
\begin{align*}
  &\phi_{\wedge}(w_0) - \phi_{\wedge}(w_1)
  = \lim_{m\to\infty} \int_{0}^{\infty} A_m(t)  \partial v_+(t) + v_0 \varphi_m(0)
  \geq \lim_{m\to\infty} \int_{a_0}^{b_0} A_m(t)  \partial v_+(t) 
  = \int_{a_0}^{b_0} \lim_{m\to\infty} A_m(t)  \partial v_+(t) \\
  &= \int_{a_0}^{b_0} (2t)^{d/2}(e^{-\frac{1}{4t}w_0^2} -
e^{-\frac{1}{4t}w_1^2}) \partial v_+(t)
  \ge v_+([a_0,b_0]) \min_{t\in [a_0,b_0]} \left \{(2t)^{d/2}
  (e^{-\frac{1}{4t}w_0^2} -
  e^{-\frac{1}{4t}w_1^2})\right \} > 0,
\end{align*}
showing $\phi_{\wedge}$ is strictly decreasing as claimed.

Suppose $\psi:[0,\infty)\to\reals$ is a $C^{\infty}$ function
with support $[a,b]$ for $0 < a < b$ such that $\psi(r) > 0$ for all $r \in (a,b)$ and
$\int_{0}^{\infty} \psi(r) \dr = 1$. Then because $\phi_{\wedge}$ is strictly decreasing,
by the mean value theorem we have
\begin{align}\label{eqn:sandwich-gft-near-origin}
\phi_{\wedge}(b/\lambda) \le \int_0^{\infty} \lambda
\phi_{\wedge}(r)\psi(\lambda r)\dr \le \phi_{\wedge}(a/\lambda)
\end{align}
for all $\lambda > 0$. If we assign $\Psi(w)\defeq \psi(\twonorm{w})$ to be
the radial continuation of $\psi$, by \cite{Baker1999} the quantity
sandwiched above becomes
\begin{align*}
\int_0^{\infty} \lambda \phi_{\wedge}(r)\psi(\lambda r)\dr
  = \int_0^{\infty} \phi_{\wedge}(s/\lambda)\psi(s)\ds
  = \frac{1}{d V_d}\int_{\reals^d}\hat{\Phi}(w/\lambda) \frac{\Psi(w)}{\twonorm{w}^{d-1}}\dw.
\end{align*}
Next suppose that $\xi:[0, \infty)\to\reals$ is a Schwartz function
satisfying $\xi^{(k)}(0) = 0$ for all integral $k\ge 0$, and
let $\Xi:\reals^d\to\reals$ given by $\Xi(x) \defeq \xi(\twonorm{x})$ be
the radial continuation of $\xi$. Then by Plancherel's Theorem,
scaling the input of a Fourier transform as in \citep[Theorem 5.16]{Wendland2004},
and the change to spherical coordinates in \cite{Baker1999}, for any
$\lambda > 0$, we have
\begin{align}\label{eqn:phi-xi}
\int_{\reals^d} \hat{\Phi}(w/\lambda) \Xi(w)\dw
  = \int_{\reals^d} \Phi(w) \hat{\Xi}(w/\lambda)\dw
  = d\,V_d\,\int_0^{\infty} r^{d-1} \phi(r) \xi_{\wedge}(r/\lambda) \dr
  = d\,V_d\,\lambda^d \int_0^{\infty} s^{d-1} \phi(\lambda s) \xi_{\wedge}(s) \,d s,
\end{align}
where $s = r/\lambda$ and $\xi_{\wedge}$ is the radial function associated with $\hat{\Xi}$,
i.e., $\hat{\Xi}(w) = \xi_{\wedge}(\twonorm{w})$ for all $w$. 

Let us define $\omega :[0,\infty)\to\reals$ by the mapping $\omega(t) \defeq
  (\alpha + t)^{\beta}$. Then by the mean value theorem and the fact that
  $\omega'$ is increasing, we have for all $s > 1$
\begin{align*}
-\omega'(\log(1 + \lambda^2s^2)) \le -\frac{\omega(\log(1 + \lambda^2s^2)) -
  \omega(\log(1 + \lambda^2))}{\log(1 + \lambda^2s^2) - \log(1 + \lambda^2)}
\le -\omega'(\log(1 + \lambda^2)).
\end{align*}
By rearranging terms, this implies for all $\lambda > 0$
\begin{align*}
  (-\beta)\left(\frac{\alpha + \log(1 + \lambda^2s^2)}{\alpha + \log(1 +
    \lambda^2)}\right )^{\beta-1}\log\left (\frac{1 + \lambda^2s^2}{1 + \lambda^2}\right )
  \le -\frac{\omega(\log(1 + \lambda^2s^2)) -
    \omega(\log(1 + \lambda^2))}
      {\omega(\log(1 + \lambda^2))(\alpha + \log(1 + \lambda^2))^{-1}}
  \le (-\beta)\log\left (\frac{1 + \lambda^2s^2}{1 + \lambda^2}\right ).
\end{align*}
Since $\log(\frac{1+\lambda^2s^2}{1+\lambda^2})\to 2\log s$ as
$\lambda\to\infty$, and the sandwiched term above is $-(\alpha +
\log(1+\lambda^2))(\phi(\lambda s)/\phi(\lambda) - 1)$, we have
$(\alpha + \log(1+\lambda^2))(\phi(\lambda s)/\phi(\lambda) - 1)\to 2\beta \log s$
as $\lambda\to\infty$ for all $s > 1$. The case for $s \in (0,1]$ is
analogous and yields the same asymptotic limit.

With this new asymptotic expansion in hand, we will revisit \eqnref{phi-xi}.
We have
\begin{align*}
\lambda^{-d}\phi(\lambda)^{-1}(\alpha + \log(1+\lambda^2))\int_{\reals^d}
\hat{\Phi}(w/\lambda) \Xi(w) \dw
 &= d\, V_d \phi(\lambda)^{-1}(\alpha + \log(1+\lambda^2))\int_{0}^{\infty}
 \phi(\lambda s) s^{d-1} \xi_{\wedge}(s)\ds \\
 &= d\, V_d\, (\alpha + \log(1+\lambda^2))\int_{0}^{\infty}
 \frac{\phi(\lambda s)}{\phi(\lambda)} s^{d-1} \xi_{\wedge}(s)\ds \\
 &= d\, V_d \int_{0}^{\infty}
 (\alpha + \log(1+\lambda^2)) \left [\frac{\phi(\lambda
     s)}{\phi(\lambda)} - 1 \right ] s^{d-1} \xi_{\wedge}(s)\ds.
\end{align*}
Notice that final integrand converges to $2\beta s^{d-1} (\log s) \xi_{\wedge}(s)$
pointwise for all $s \ge 0$ as $\lambda \to \infty$. Since $\xi_{\wedge}$ is
a Schwartz function on $[0,\infty)$, we can utilize the fact that $s\mapsto
\log s$ is integrable near the origin to reason that $s^{d-1} (\log s)
\xi_{\wedge}(s)$ is a Schwartz function as well, and thus integrable. Hence 
by the dominated convergence theorem, we have the integral above converges to
$2\, \beta\, d\, V_d \int_{0}^{\infty}  s^{d-1} (\log s) \xi_{\wedge}(s)\ds$
as $\lambda\to\infty$.

Now suppose we choose $\Xi(x)\defeq \twonorm{x}^{1-d}\Psi(x)$. By
\eqnref{sandwich-gft-near-origin} we have
\begin{align*}
\lim_{\lambda \to\infty}
\lambda^{-d}\phi(\lambda)^{-1}(\alpha + \log(1+\lambda^2)) \phi_{\wedge}(b/\lambda) \le
2\beta \int_0^{\infty} s^{d-1}(\log s)\xi_{\wedge}(s)\ds
\le \lim_{\lambda \to\infty} \lambda^{-d}\phi(\lambda)^{-1}(\alpha + \log(1+\lambda^2)) \phi_{\wedge}(a/\lambda).
\end{align*}
By \lemref{log-inverse-nonvanishing-gft}, we know $\phi_{\wedge}(r) > 0$ for
all $r > 0$, and thus the left-hand side above must be non-negative. Hence if we can
show for some choice of $\psi$ that the sandwiched term is non-zero, then the proof of
the lemma will follow from choosing $r = a/\lambda$.

Let us define $L(x) = \log \twonorm{x}$ with generalized Fourier transform $\hat{L}$. 
As usual, let
$l:[0,\infty)\to\reals$ and $l_{\wedge}:[0,\infty)\to\reals$ be the radial
    functions associated with $L$ and $\hat{L}$.  Notice that again by
    Plancherel's Theorem
\begin{align}\label{eqn:multivariate-log-gft-limit}
\int_0^{\infty} s^{d-1}(\log s)\xi_{\wedge}(s)\ds
  = \frac{1}{d V_d}\int_{\reals^d} \hat{\Xi}(w)L(w) \dw
  = \frac{1}{d V_d}\int_{\reals^d} \Xi(x)\hat{L}(x) \dx
  &= \frac{1}{d V_d}\int_{\reals^d}
  \frac{\Psi(x)}{\twonorm{x}^{d-1}}\hat{L}(x) \dx \notag \\
  &= \int_0^{\infty} \psi(r) l_{\wedge}(r)\dr.
\end{align}
Since we are free to choose $\psi$ to be any Schwartz function with support
$[a,b]$, if we could not find a function $\psi$ such that
the quantity in \eqnref{multivariate-log-gft-limit} is non-zero, this would
imply the support of $l_{\wedge}$ is a subset of $\{0\}$. But this would
mean $l_{\wedge}$ is some multiple of a point mass at zero, which would
imply $l$ is a constant function, a contradiction. Thus we must be able to
find some $\psi$ such that the integral above is non-zero, completing the
lemma.
\end{proof}

Fix any $a_0 > 0$
and $\alpha_0\in (0,\frac{1}{2})$.
Our strategy for showing the KSD controls tightness will mimic
\citep[Lem. 16]{Gorham2017}: we will
show that a bandlimited approximation of the function $g_j(x) = 2\alpha_0 x_j(a_0^2 +
\twonorm{x}^2)^{\alpha_0-1}$ belongs to the inverse log RKHS and thus enforces tightness.

First note that in the proof of \citep[Lem. 16]{Gorham2017}, it was shown
$h = \mc{T}_P g$ was a coercive, Lipschitz, and bounded-below function for $P\in\pset$. Moreover,
in the proof of \citep[Lem. 12]{Gorham2017}, a random vector $Y$ with
density $\rho(y)$ is constructed such that
the support of $\hat{\rho}$ belongs to $[-4,4]^d$ and also $\twonorm{Y}$ is
integrable. Consider the new function $g^{\circ}(x) \defeq \Earg{g(x +
  Y)}$ for all $x\in\reals^d$. By the convolution theorem, $\hat{g}^{\circ}_j
= \hat{g}_j\,\hat{\rho}$ and so $g^{\circ}_j$ is bandlimited for all $j$.
In the proof of \citep[Lem. 16]{Gorham2017}, $\hat{g}_j$ was shown to grow
asymptotically at the rate $(iw_j) \twonorm{w}^{-d-2\alpha_0}$ as
$\twonorm{w}\to 0$. Thus
\begin{align*}
  \sum_{j=1}^d \int_{\reals^d} \frac{\hat{g}^{\circ}_j(w)\overline{\hat{g}^{\circ}_j(w)}}{\hat{\Phi}(w)} \dw
    = \sum_{j=1}^d \int_{[-4,4]^d} \frac{\hat{g}_j(w)\overline{\hat{g}_j(w)}
      \hat{\rho}(w)^2}{\hat{\Phi}(w)} \dw
    &\le \kappa_0 \int_{[-4,4]^d}
    \frac{\twonorm{w}^{-2d-4\alpha_0+2}}{\hat{\Phi}(w)} \dw \\
    &\le \kappa_1 \int_0^{4\sqrt{d}} r^{-4\alpha_0+1}\log^{-\beta+1}(1 + r^{-2})\dr,
\end{align*}
for some constants $\kappa_0,\kappa_1 > 0$ where we used
\lemref{log-inverse-gft-lower-bound} in the final inequality.
This integral is finite for
all $\alpha_0 \in (0,\frac{1}{2})$ and any $\beta < 0$, which implies
$g^{\circ}$ is in the log inverse RKHS by \citep[Theorem 10.21]{Wendland2004}.

Finally, notice that via the argument proving \citep[Lemma
  12]{Gorham2017},
\begin{align*}
  \sup_{x\in\reals^d} |\mc{T}_P g^{\circ}(x) - h(x)|
  \le \frac{3d\log 2}{\pi}\left(\sup_{x\in\reals^d} \twonorm{\grad h(x)} +
  \sup_{x\in\reals^d} \opnorm{\grad^2\log p(x)} \cdot
  \sup_{x\in\reals^d} \twonorm{g(x)}\right) < \infty.
\end{align*}
Since $h$ is bounded below and coercive, these properties are inherited by
$\mc{T}_Pg^{\circ}$. This allows us to apply \citep[Lemma 17]{Gorham2017} to argue $D_{\mathcal{K}_0,P}\left(\mu_m \right)\to 0$
implies the measures $\mu_m$ are uniformly tight. Combining this with 
\lemref{log-inverse-nonvanishing-gft} allows us to utilize \citep[Theorem
  7]{Gorham2017} for the log inverse kernel, thereby concluding the proof.

\subsection{Proof of \thmref{imq-score-implies-weak-convergence}: IMQ Score KSD Convergence Control}
\label{sec:imq-score-implies-weak-convergence-proof}
For $b = \grad \log p$, introduce the alias $k_b = k_3$, let $\kset_b$ denote the RKHS of $k_b$, and let $C_c$ represent the set of continuous compactly supported functions on $X$.
Since $P\in\pset$, the proof of Thm. 13 in \citep{Gorham2017} shows that 
if, for each $h\in C^1\cap C_c$ and $\eps > 0$, there exists $h_\eps \in \kset_b$
such that $\sup_{x \in X} |(\mc{T}_Ph)(x) - (\mc{T}_Ph_\eps)(x)| \leq \eps$, then $\mu_m \Rightarrow P$
whenever $D_{\mathcal{K}_0,P}\left(\mu_m \right) \to 0$ and $(\mu_m)_{m=1}^\infty$ is uniformly tight.
Hence, to establish our result, it suffices to show 
(1) that, for each $h\in C^1\cap C_c$ and $\eps > 0$, there exists $h_\eps \in \kset_b$
such that $\sup_{x \in X} \max(\twonorm{\grad (h-h_\eps)(x)}, \twonorm{b(x) (h-h_\eps)(x)}) \leq \epsilon$
and 
(2) that $D_{\mathcal{K}_0,P}\left(\mu_m \right) \to 0$ implies $(\mu_m)_{m=1}^\infty$ is uniformly tight.
\subsubsection{Approximating $C^1 \cap C_c$ with $\kset_b$}
Fix any $f\in C^1\cap C_c$ and $\eps > 0$, and let $\kset$ denote the RKHS of $k(x,y) = (c^2 + \twonorm{x-y}^2)^{\beta}$.
Since $p$ is strictly log-concave, $b$ is invertible with $\det(\grad b(x))$ never zero.
Since $P \in \pset$, $b$ is Lipschitz.
By the following theorem, proved in \secref{composition-kernel-properties-proof}, it therefore suffices to show that there exists $f_\eps \in\kset$ such that $\sup_{x \in X} \max(\twonorm{\grad (f-f_\eps)(x)},  \twonorm{x (f-f_\eps)(x)})\leq \epsilon$.
\begin{theorem}[Composition Kernel Approximation]\label{thm:composition-kernel-properties}
For $b : X \to X$ invertible and $k$ a reproducing kernel on $X$ with induced RKHS $\kset$, define the composition kernel $k_b(x,y) = k(b(x), b(y))$
with induced RKHS $\kset_b$.
Suppose that, for each $f\in C^1\cap C_c$ and $\eps > 0$, there exists $f_\eps \in\kset$ such that \[\sup_{x \in X} \max(\twonorm{\grad (f-f_\eps)(x)},  \twonorm{x (f-f_\eps)(x)})\leq \epsilon.\]
If $b$ is Lipschitz and $\det(\grad b(x))$ is never zero, then, for each $h\in C^1\cap C_c$ and $\eps > 0$, there exists $h_\eps \in \kset_b$
such that \[\sup_{x \in X} \max(\twonorm{\grad (h-h_\eps)(x)}, \twonorm{b(x) (h-h_\eps)(x)}) \leq \epsilon.\]
\end{theorem}
Since the identity map $x\mapsto x$ is Lipschitz and $f \in L^2$ because it is continuous and compactly supported, \citep[Lem. 12]{Gorham2017} provides an explicit construction of $f_\eps \in \kset$ satisfying our desired property whenever $k(x,y) = \Phi(x-y)$ for $\Phi \in C^2$ with non-vanishing Fourier transform.  Our choice of IMQ $k$ satisfies these properties by \citep[Thm. 8.15]{Wendland2004}.

\subsubsection{Controlling Tightness}
Since $P$ is distantly dissipative,
\[
-\twonorm{b(x)}\twonorm{x}
	\leq \inner{b(x)}{x} 
	\leq -\kappa\twonorm{x}^2+C + \inner{b(0)}{x}
	\leq -\kappa\twonorm{x}^2+C + \twonorm{b(0)}\twonorm{x}
\]
by Cauchy-Schwarz.
Hence, $b$ is \emph{norm-coercive}, i.e., $\twonorm{b(x)}\to\infty$ whenever $\twonorm{x} \to\infty$.
Since $\grad b$ is bounded, our desired result follows from the following lemma
which guarantees tightness control on $b$ under weaker conditions.
\begin{lemma}[Coercive Score Kernel KSDs Control Tightness]
If $b : X \to X$ is norm coercive and differentiable, and $\grad_j b_j(x) = o(\twonorm{b(x)}^2)$ as $\twonorm{x}\to\infty$, then $\limsup_m D_{\mathcal{K}_0,P}\left(\mu_m \right) < \infty$ implies $(\mu_m)_{m=1}^\infty$ is tight.
\end{lemma}
\begin{proof}
Fix any $a >c/2$ and $\alpha \in(0,\frac{1}{2}(\beta + 1))$. 
The proof of \citep[Lem. 16]{Gorham2017} showed that the function $g_j(x) = 2\alpha x_j (a^2 + \twonorm{x}^2)^{\alpha - 1} \in \kset$ for each $j \in \{1,\dots, d\}$.
Hence $g_{b,j}(x) \defeq g_j(b(x)) \in \kset_b$ for each $j \in \{1,\dots, d\}$ by \lemref{composition-rkhs}.
By our assumptions on $\grad b$, we have
\begin{align*}
(\mc{T}_P g_{b})(x) 
	&= 2\alpha (\twonorm{b(x)}^2 (a^2 + \twonorm{b(x)}^2)^{\alpha - 1} 
	+ \sum_{j=1}^d \grad_j b_j(x)  (a^2 + \twonorm{b(x)}^2)^{\alpha - 1}
	+ b_j(x)^2 2(\alpha - 1)(a^2 + \twonorm{b(x)}^2)^{\alpha - 2} \grad_j b_j(x))\\
	&= 2\alpha \twonorm{b(x)}^2 (a^2 + \twonorm{b(x)}^2)^{\alpha - 1}  + o(\twonorm{b(x)}^{2\alpha}),
\end{align*}
so $\mc{T}_P g_{b}$ is coercive, and the proof of \citep[Lem. 17]{Gorham2017} therefore gives the result $(\mu_m)_{m=1}^\infty$ is uniformly tight whenever $\limsup_m D_{\mathcal{K}_0,P}\left(\mu_m \right)$ finite. 
\end{proof}

\subsection{Proof of \thmref{composition-kernel-properties}: Composition Kernel Approximation}
\label{sec:composition-kernel-properties-proof}

Let $c=b^{-1}$ represent the inverse of $b$, and for any function $f$ on $X$, let $f_c(y) = f(c(y))$ denote the composition of $f$ and $c$ so that $f_c(b(x)) = f(x)$.
The following lemma shows that $f_c$ inherits many of the properties of $f$ under suitable restrictions on $b$.
\begin{lemma}[Composition Properties] \label{lem:composition-properties}
For any function $f$ on $X$ and invertible function $b$ on $X$, define $f_c(y) = f(c(y))$ for $c = b^{-1}$.
The following properties hold.
\begin{enumerate}
	\item If $f$ has compact support and $b$ is continuous, then $f_c$ has compact support.
	\item If $f\in C^1$, $b \in C^1$, and $\det(\grad b(x))$ is never zero, then $f_c \in C^1$. 
\end{enumerate}
\end{lemma} 
\begin{proof}
We prove each claim in turn.
\begin{enumerate}
\item If $f$ is compactly supported and $b$ is continuous, then $\supp{f_c} = b(\supp{f})$ is also compact, since continuous functions are compact-preserving \citep[Prop. 1.8]{Joshi1983}.
\item If $f\in C^1$, $b \in C^1$, and $\det(\grad b(x))$ is never zero, then $c$ is continuous by the inverse function theorem \citep[Thm. 2-11]{Spivak1965}, $x\mapsto (\grad b(x))^{-1}$ is continuous,
and hence $\grad f_c(y) = (\grad c(y)) (\grad f)(c(y)) = ((\grad b)(c(y)))^{-1} (\grad f)(c(y))$ is continuous.
\end{enumerate}
\end{proof}   
Our next lemma exposes an important relationship between the RKHSes $\kset$ and $\kset_b$.
\begin{lemma}\label{lem:composition-rkhs}
Suppose $f$ is in the RKHS $\kset$ of a reproducing kernel $k$ on $X$ and $b : X \to X$ is invertible.
Then $f_b$ is in the RKHS $\kset_b$ of $k_b$ for $f_b(x) = f(b(x))$ and $k_b(x,y) = k(b(x), b(y))$.
\end{lemma}
\begin{proof}
Since $f\in\kset$, there exist $f_m = \sum_{j=1}^{J_m} a_{m,j} k(x_{m,j},\cdot)$ for $m\in\naturals, a_{m,j}\in\reals,$ and $x_{m,j}\in X$ such that $\lim_{m\to\infty}\norm{f_m - f}_{\kset} = 0$ and $\lim_{m\to\infty} f_m(x) = f(x)$ for all $x\in X$.
Now let $c=b^{-1}$, and define 
\[
f_{m,b}(x) = f_{m}(b(x)) = \sum_{j=1}^{J_m} a_{m,j} k(x_{m,j},b(x)) = \sum_{j=1}^{J_m} a_{m,j} k_b(c(x_{m,j}),x).
\]
Since $\kset_b = \overline{\{\sum_{j=1}^{J} a_{j} k_b(y_j,\cdot) : J\in\naturals, a_j\in\reals, y_j\in X\}}$,
each $f_{m,b} \in \kset_b$. 
Since $(f_m)_{m=1}^\infty$ is a Cauchy sequence, and $\inner{f_{m,b}}{f_{m',b}}_{\kset_b}^2 =  \sum_{j=1}^{J_m} a_{m,j} \sum_{j'=1}^{J_{m'}} a_{m',j'}k_b(c(x_{m,j}),c(x_{m',j'}))
= \inner{f_{m}}{f_{m'}}_{\kset}$ so that $\norm{ f_{m,b}-f_{m',b}}_{\kset_b} = \norm{ f_{m}-f_{m'}}_{\kset}$ for all $m,m'$,
the sequence $(f_{m,b})_{m=1}^\infty$ is also Cauchy and converges in $\norm{\cdot}_{\kset_b}$ to its pointwise limit $f_b$.
Since an RKHS is complete, $f_b \in \kset_b$.
\end{proof}
With our lemmata in hand, we now prove the advertised claim. %
Suppose $b$ is Lipschitz, $\det(\grad b(x))$ is never zero, and for each $f\in C^1\cap C_c$ and $\eps > 0$ there exists $f_\eps \in\kset$ such that $\sup_{x \in X} \max(\twonorm{\grad (f-f_\eps)(x)},  \twonorm{x (f-f_\eps)(x)})\leq \epsilon$.
Select any $h \in C^1\cap C_c$ and any $\eps > 0$.
By \lemref{composition-properties}, $h_c \in C^1\cap C_c$,
and hence there exists $h_{c,\eps} \in \kset$ such that 
$\sup_{y \in X} \max(\twonorm{\grad (h_c-h_{c,\eps})(y)},  \twonorm{y (h_c-h_{c,\eps})(y)})\leq \eps/\max(1,M_1(b))$.
Now define $h_\eps(x) = h_{c,\eps}(b(x))$ so that $h_\eps \in \kset_b$ by \lemref{composition-rkhs}.
We have $\sup_{x\in X} \twonorm{b(x)(h_\eps -h)(x)} \leq \sup_{y\in X} \twonorm{y(h_{c,\eps}-h_c)(y)} \leq \eps$, and  
\[
\sup_{x\in X}\twonorm{\grad h_\eps(x) - \grad h(x)} 
	= \sup_{x\in X}\twonorm{(\grad b(x)) ((\grad h_{c,\eps})(b(x)) - (\grad h_c)(b(x)))} 
	\leq M_1(b) \eps/\max(1,M_1(b)) \leq \eps.
\]

\section{Implementational Detail} \label{sup: benchmark}

\subsection{Benchmark Methods}

In this section we briefly describe the MED and SVGD methods used as our empirical benchmark, as well as the (block) coordinate descent method that was used in conjunction with Stein Points.

\subsubsection{Minimum Energy Designs} \label{subsec: MED describe}

The first class of method that we consider is due to \cite{Joseph2015}.
That work restricted attention to $X = [0,1]^d$ and constructed an energy functional:
\begin{eqnarray*}
\mathcal{E}_{\delta,P}(\{x_i\}_{i=1}^n) & := & \sum_{i \neq j} \left[ \frac{ p(x_i)^{-\frac{1}{2d}} p(x_j)^{-\frac{1}{2d}} }{\|x_i - x_j\|_2 } \right]^\delta 
\end{eqnarray*}
for some tuning parameter $\delta \in [1,\infty)$ to be specified.
In \cite{Joseph2017} the rule-of-thumb $\delta = 4d$ was recommended.
A heuristic argument in \cite{Joseph2015} suggests that the points $\{x_i\}_{i=1}^n$ that minimise $\mathcal{E}_{\delta,P}(\{x_i\}_{i=1}^n)$ form an empirical approximation that converges weakly to $P$.
The argument was recently made rigorous in \cite{Joseph2017}.

Minimisation of $\mathcal{E}_{\delta,P}$ does not require knowledge of how $p$ is normalised.
However, the actual minimisation of $\mathcal{E}_{\delta,P}$ can be difficult.
In \cite{Joseph2015} an extensible (greedy) method was considered, wherein the first point is selected as
$$
x_1 \in \argmax_{x \in X} \quad p(x)
$$
and subsequent points are selected as
\begin{eqnarray*}
x_n & \in & \argmin_{x \in X} \quad p(x)^{-\frac{\delta}{2d}} \sum_{i=1}^{n-1} \frac{p(x_i)^{-\frac{\delta}{2d}}}{\|x_i - x\|_2^\delta } .
\end{eqnarray*}
However, alternative approaches could easily be envisioned.
For instance, if $n$ were fixed then one could consider e.g. applying the Newton method for optimisation over the points $\{x_i\}_{i=1}^n$.

{\it Remark:} There is a connection between certain minimum energy methods and discrepancy measures in RKHS; see \citep{Sejdinovic2013}.

{\it Remark:} Several potential modifications to $\mathcal{E}_{\delta,P}$ were suggested in \cite{Joseph2017}, but that report appeared after this work was completed.
These could be explored in future work.

{\it Remark:} The MED objective function is typically numerically unstable due to the fact that the values of the density $p(\cdot)$ can be very small.
In contrast, our proposed methods operate on $\log p(\cdot)$ and its gradient, which is more numerically robust.

\subsubsection{Stein Variational Gradient Descent} \label{subset: SVGD describe}

The second method that we considered was due to \cite{Liu2016a,Liu2017} and recently generalised in \cite{Liu2017a}.
The idea starts by formulating a continuous version of gradient descent on $\mathcal{P}(X)$ with the Kullback-Leibler divergence $\text{KL}(\cdot || P)$ as a target.
To this end, restrict attention to $X = \mathbb{R}^d$ and consider the dynamics
\begin{eqnarray*}
S_f(x) & = & x + \epsilon f(x)
\end{eqnarray*}
parametrised by a function $f \in \mathcal{K}^d$.
For infinitesimal values of $\epsilon$ we can lift $S_f$ to a pushforward map on $\mathcal{P}(X)$; i.e. $Q \mapsto S_f Q$.
It was then shown in \cite{Liu2016a} that
\begin{eqnarray}
- \left. \frac{\mathrm{d}}{\mathrm{d} \epsilon} \mathrm{KL}(S_f Q || P) \right|_{\epsilon = 0} & = & \int \mathcal{T}_P f \; \mathrm{d}Q  \label{eq: KL diff}
\end{eqnarray}
where $\mathcal{T}_P$ is the Langevin Stein operator in Eqn. \ref{eq: langevin}.
Recall that this operator can be decomposed as $\mathcal{T}_P f = \sum_{j=1}^d \mathcal{T}_{P,j} f_j$ with $\mathcal{T}_{P,j} = \nabla_j + \nabla_j \log p$, where $\nabla_j$ denotes differentiation with respect to the $j$th coordinate in $X$.
Then the direction of fastest descent
\begin{eqnarray*}
f^*(\cdot) & := & \argmax_{f \in B(\mathcal{K}^d)} \; - \left. \frac{\mathrm{d}}{\mathrm{d} \epsilon} \mathrm{KL}(S_f Q || P) \right|_{\epsilon = 0}
\end{eqnarray*}
has a closed-form, with $j$th coordinate
\begin{eqnarray*}
f_j^*(\cdot ; Q) & = & \int \mathcal{T}_{P,j} k(x,\cdot) \; \mathrm{d}Q(x) .
\end{eqnarray*}
The algorithm proposed in \cite{Liu2016a} discretises this dynamics in both space $X$, through the use of $n$ points, and in time, through the use of a positive step size $\epsilon > 0$, leading to a sequence of empirical measures based on point sets $\{x_i^m\}_{i=1}^n$ for $m \in \mathbb{N}$.
Thus, given an initialisation $\{x_i^0\}_{i=1}^n$ of the points, at iteration $m \geq 1$ of the algorithm we update 
\begin{eqnarray*}
x_i^m & = & x_i^{m-1} + \epsilon f^*(x_i^{m-1} ; Q_n^m) 
\end{eqnarray*}
in parallel, where
\begin{eqnarray*}
Q_n^m & = & \frac{1}{n} \sum_{i=1}^n \delta_{x_i^{m-1}}
\end{eqnarray*}
is the empirical measure, at a computational cost of $O(n)$.
The output is the empirical measure $Q_n^m$.

{\it Remark:} The step size $\epsilon$ is a tuning parameter of the method.

{\it Remark:} At present there are not theoretical guarantees for this method.
Initial steps toward this goal are presented in \cite{Liu2017}.

\subsubsection{Block Coordinate Descent} \label{subset: block describe}

The Stein Point methods developed in the main text can be adapted to return a fixed number $n$ of points for a given finite computational budget by first iteratively generating a size $n$ point set, as described in the main text, and then performing (block) coordinate descent on this point set.
The (block) coordinate descent procedure is now described:

Fix an initial configuration $\{x_i^0\}_{i=1}^n$.
Then at iteration $m \geq 1$ of the algorithm, perform the following sequence of operations:
\begin{eqnarray*}
\forall i & & x_i^m \leftarrow x_i^{m-1} \quad \text{then:} \\
\text{for } i = 1,\dots,n & & x_i^m \leftarrow  \argmin_{x \in X} D_{\mathcal{K}_0,P}(\{x_j^m\}_{j \neq i} \cup \{x\})
\end{eqnarray*}
The output is the point set $\{x_i^m\}_{i=1}^n$.

{\it Remark:} The block coordinate descent method can equally be applied to MED; this was not considered in our empirical work.

{\it Remark:} Any numerical optimisation method can be used to solve the global optimisation problem in the inner loop.
In this work we considered the same three candidates in the main text; Monte Carlo, Nelder-Mead and grid search.
These are described next.

\subsection{Numerical Optimisation Methods}  \label{sup: optimisation methods}

Computation of the $n$th term in the proposed Stein Point sequences, given the previous $n-1$ terms, requires that a global optimisation is performed over $x_n \in X$.
The same is true for both MED and KSD in the coordinate descent context.
For all experiments reported in the main text, three different numerical methods were considered for this task, denoted \verb+NM+, \verb+MC+, \verb+GS+ in the main text.
In this section we provide full details for how these methods were implemented.

\subsubsection{Nelder-Mead}  \label{subsec: optim NM}

The Nelder-Mead (\verb+NM+) method \cite{Nelder1965} proceeds as in Algorithm \ref{alg: nm}.
The function $\mathrm{NM}$ takes the following inputs: $f$ is the objective function; $t$ is the iteration count; $n_{\mathrm{init}}$ is the number of initial points to be drawn from a proposal distribution; $n_{\mathrm{delay}}$ is the number of iterations after which the proposal distribution becomes adaptive; $\mu_0$ and $\Sigma_0$ are the mean vector and the covariance matrix of the initial proposal distribution; $\{x^{\mathrm{curr}}_j\}_{j=1}^{n_{\mathrm{curr}}}$ is the set of existing points; $\lambda$ is the variance of each mixture component of the adaptive proposal distribution; $l$ and $u$ are the lower- and upper-bounds of the search space. The non-adaptive initial proposal distribution is a truncated multivariate Gaussian $\mathcal{N}(\mu_0, \Sigma_0)$ whose support is bounded by the hypercube $[l,u]$. The adaptive proposal distribution is a truncated Gaussian mixture $\Pi(\{x^{\mathrm{curr}}_j\}_{j=1}^{n_{\mathrm{curr}}}, \lambda) := \frac{1}{n_{\mathrm{curr}}-1}\sum_{j=1}^{n_{\mathrm{curr}}-1}\mathcal{N}(x^{\mathrm{curr}}_j,\lambda I)$ with $\lambda > 0$ and support $[l,u]$. The expression $\mathrm{NelderMead}_{x}\left[f(x), x^{\mathrm{init}}_{i}, l, u\right]$ denotes the standard Nelder-Mead procedure for objective function $f$, initial point $x^{\mathrm{init}}_{i}$, and bound constraint $x \in [l,u]$. We use the symbol $\leftlsquigarrow$ to denote the assignment of a realised independent draw. The operator $\mathrm{trunc}_{l}^{u}[\cdot]$ bounds the support of a distribution by the hypercube $[l,u]$.

\begin{algorithm}[h!]
\caption{Nelder-Mead}
\label{alg: nm}
\begin{algorithmic}[1]
\INPUT $f$, $t$, $n_{\mathrm{init}}$, $n_{\mathrm{delay}}$, $\mu_0$, $\Sigma_0$, $\{x^{\mathrm{curr}}_j\}_{j=1}^{n_{\mathrm{curr}}}$, $\lambda$, $l$, $u$
\OUTPUT $x^*$
\FUNCTION{NM}
\FOR{$i \gets 1:n_{\mathrm{init}}$}
\IF{$t \le n_{\mathrm{delay}}$}
\STATE $x^{\mathrm{init}}_{i} \leftlsquigarrow \mathrm{trunc}_{l}^{u}\left[\mathcal{N}(\mu_0, \Sigma_0)\right]$
\ELSE
\STATE $x^{\mathrm{init}}_{i} \leftlsquigarrow \mathrm{trunc}_{l}^{u}\left[\Pi(\{x^{\mathrm{curr}}_j\}_{j=1}^{n_{\mathrm{curr}}}, \lambda)\right]$
\ENDIF
\STATE $x^{\mathrm{local}}_{i} \gets \mathrm{NelderMead}_{x}\left[f(x), x^{\mathrm{init}}_{i}, l, u\right]$
\ENDFOR
\STATE $i^* \gets \argmin_{i \in\{1 \ldots n_{\mathrm{init}}\}} f(x^{\mathrm{local}}_{i})$
\STATE $x^* \gets x^{\mathrm{local}}_{i^*}$
\ENDFUNCTION
\end{algorithmic}
\end{algorithm}

\FloatBarrier

\subsubsection{Monte Carlo}  \label{subsec: optim MC}

The Monte Carlo (\verb+MC+) optimisation method proceeds as in Algorithm \ref{alg: mc}.
The function $\mathrm{MC}$ takes the following inputs: $f$ is the objective function; $t$ is the iteration count; $n_{\mathrm{test}}$ is the number of test points to be drawn from a proposal distribution; $n_{\mathrm{delay}}$ is the number of iterations after which the proposal distribution becomes adaptive; $\mu_0$ and $\Sigma_0$ are the mean vector and the covariance matrix of the initial proposal distribution; $\{x^{\mathrm{curr}}_j\}_{j=1}^{n_{\mathrm{curr}}}$ is the set of existing points; $\lambda$ is the variance of each mixture component of the adaptive proposal distribution; $l$ and $u$ are the lower- and upper-bounds of the search space. The non-adaptive initial proposal distribution is a truncated multivariate Gaussian $\mathcal{N}(\mu_0, \Sigma_0)$ whose support is bounded by the hypercube $[l,u]$. The adaptive proposal distribution is a truncated Gaussian mixture $\Pi(\{x^{\mathrm{curr}}_j\}_{j=1}^{n_{\mathrm{curr}}}, \lambda) := \frac{1}{n_{\mathrm{curr}}-1}\sum_{j=1}^{n_{\mathrm{curr}}-1}\mathcal{N}(x^{\mathrm{curr}}_j,\lambda I)$ with $\lambda > 0$ and support $[l,u]$.

\begin{algorithm}[h!]
\caption{Monte Carlo}
\label{alg: mc}
\begin{algorithmic}[1]
\INPUT $f$, $t$, $n_{\mathrm{test}}$, $n_{\mathrm{delay}}$, $\mu_0$, $\Sigma_0$, $\{x^{\mathrm{curr}}_j\}_{j=1}^{n_{\mathrm{curr}}}$, $\lambda$, $l$, $u$
\OUTPUT $x^*$
\FUNCTION{MC}
\IF{$t \le n_{\mathrm{delay}}$}
\STATE $\{x^{\mathrm{test}}_{i}\}_{i=1}^{n_{\mathrm{test}}} \leftlsquigarrow \mathrm{trunc}_{l}^{u}\left[\mathcal{N}(\mu_0, \Sigma_0)\right]$
\ELSE
\STATE $\{x^{\mathrm{test}}_{i}\}_{i=1}^{n_{\mathrm{test}}} \leftlsquigarrow \mathrm{trunc}_{l}^{u}\left[\Pi(\{x^{\mathrm{curr}}_j\}_{j=1}^{n_{\mathrm{curr}}}, \lambda)\right]$
\ENDIF
\STATE $i^* \gets \argmin_{i \in\{1 \ldots n_{\mathrm{test}}\}} f(x^{\mathrm{test}}_{i})$
\STATE $x^* \gets x^{\mathrm{test}}_{i^*}$
\ENDFUNCTION
\end{algorithmic}
\end{algorithm}

\FloatBarrier

\subsubsection{Grid Search}  \label{subsec: optim GS}

The grid search (\verb+GS+) optimisation method proceeds as in Algorithm \ref{alg: gs}.
The function $\mathrm{GS}$ takes the following inputs: $f$ is the objective function; $t$ is the iteration count; $l$ and $u$ are the lower- and upper-bounds of the grid; $n_0$ is the initial grid size.

\begin{algorithm}[h!]
\caption{Grid Search}
\label{alg: gs}
\begin{algorithmic}[1]
\INPUT $f$, $t$, $l$, $u$, $n_0$
\OUTPUT $x^*$
\FUNCTION{GS}
\STATE $n_{\mathrm{grid}} \gets n_0 + \mathrm{Round}(\sqrt{t})$
\STATE $\delta_{\mathrm{grid}} \gets (u - l) / (n_{\mathrm{grid}} - 1)$
\STATE $X_{\mathrm{grid}} \gets \{l, l + \delta_{\mathrm{grid}}, \ldots, u\}^d$
\STATE $x^* \gets \argmin_{x \in X_{\mathrm{grid}}} f(x)$
\ENDFUNCTION
\end{algorithmic}
\end{algorithm}

\FloatBarrier

\subsection{Remark on Application to a Reference Point Set}\label{subsec: fixed set}

It is interesting to comment on the behaviour of our proposed methods in the case where $X$ is a finite set or the global optimisation over $X$ is replaced by a discrete optimisation over a pre-determined fixed set $Y = \{y_i\}_{i=1}^N \subseteq X$.
In this case it can be shown that:
\begin{itemize}
\item The algorithm after $n$ iterations will have selected $n$ points $\{y_{\pi(i)}\}_{i=1}^n$ with replacement from $Y$.
(Here $\pi(i)$ indexes the point that was selected at iteration $i$ of the algorithm.)
\item The empirical measure $\frac{1}{n} \sum_{i=1}^n \delta_{y_{\pi(i)}}$ can be expressed as $\sum_{i=1}^N w_i y_i$ for some weights $w_i$.
\item The weights $w_i$ converge to 
$$
(*) \; = \; \argmin_{\substack{w \geq 0 \\ w_1 + \dots + w_N = 1 }} \sqrt{ \frac{1}{N^2} \sum_{i,j = 1}^N w_i w_j k_0(y_i,y_j) } .
$$
\item At iteration $n$, it holds that $D_{\mathcal{K}_0,P}(\{y_{\pi(i)}\}_{i=1}^n) = (*) + O(\sqrt{\log(n) / n})$.
\end{itemize}
Thus in this scenario the algorithms that we have proposed act to ensure that these points are optimally weighted in the sense just described.

\section{Experimental Protocol and Additional Numerical Results}  \label{sup: additional numerics}

This section contains additional numerical results that elaborate on the three experiments reported in the main text.

\subsection{Gaussian Mixture Test} \label{subsec: ad numerics GMM}

Recall from the main text that the kernels $k_1$, $k_2$ and $k_3$ contain either one or two hyper-parameters that must be selected.
For each of the methods (a)-(f) reported in Figure \ref{fig: GMM} in the main text we optimised these parameters over a discrete set, with respect to an objective function of $W_P$ based on a point set of size $n = 100$ and the Nelder-Mead optimisation method. The set of possible values for $\alpha$ was $\{0.1\eta, 0.5\eta, \eta, 2\eta, 4\eta, 8\eta\}$, where $\eta$ is a problem dependent ``base scale" and chosen to be 1 for the Gaussian mixture test. The set of possible values for $\beta$ was $\{-0.1, -0.3, -0.5, -0.7, -0.9\}$.
The sensitivity of the reported results to the variation in hyper-parameters is shown, for the Gaussian mixture test, in Figure \ref{fig: GMM params}.
Point sets obtained under representatives of each method class are shown in Figure \ref{fig: GMM points v2}.

For all the global optimisation methods we imposed a bounding box $(-5, 5) \times (-5,5)$; for the Nelder-Mead method, we set $n_{\mathrm{init}} = 3$, $n_{\mathrm{delay}} = 20$, $\mu_0 = (0, 0)$, $\Sigma_0 = 25I$, and $\lambda = 1$; for the Monte Carlo method, we set $n_{\mathrm{test}} = 20$, $n_{\mathrm{delay}} = 20$, $\mu_0 = (0, 0)$, $\Sigma_0 = 25I$, and $\lambda = 1$; for the grid search, we set $n_0 = 100$.

For MED the tuning parameter $\delta$ was considered for $\delta = 4$, $\delta = 8$ or $\delta = 16$, with $\delta = 4d = 8$ being the recommendation in \cite{Joseph2017}.

For SVGD we set the initial point-set to be an equally spaced rectangular grid over the bounding box. Following \cite{Liu2016a}, the step-size $\epsilon$ for SVGD was determined by AdaGrad with a master step-size of 0.1 and a momentum factor of 0.9.

\begin{figure}[h]
\centering
\begin{subfigure}[b]{0.33\textwidth}
		\centering
        \includegraphics[width=\textwidth]{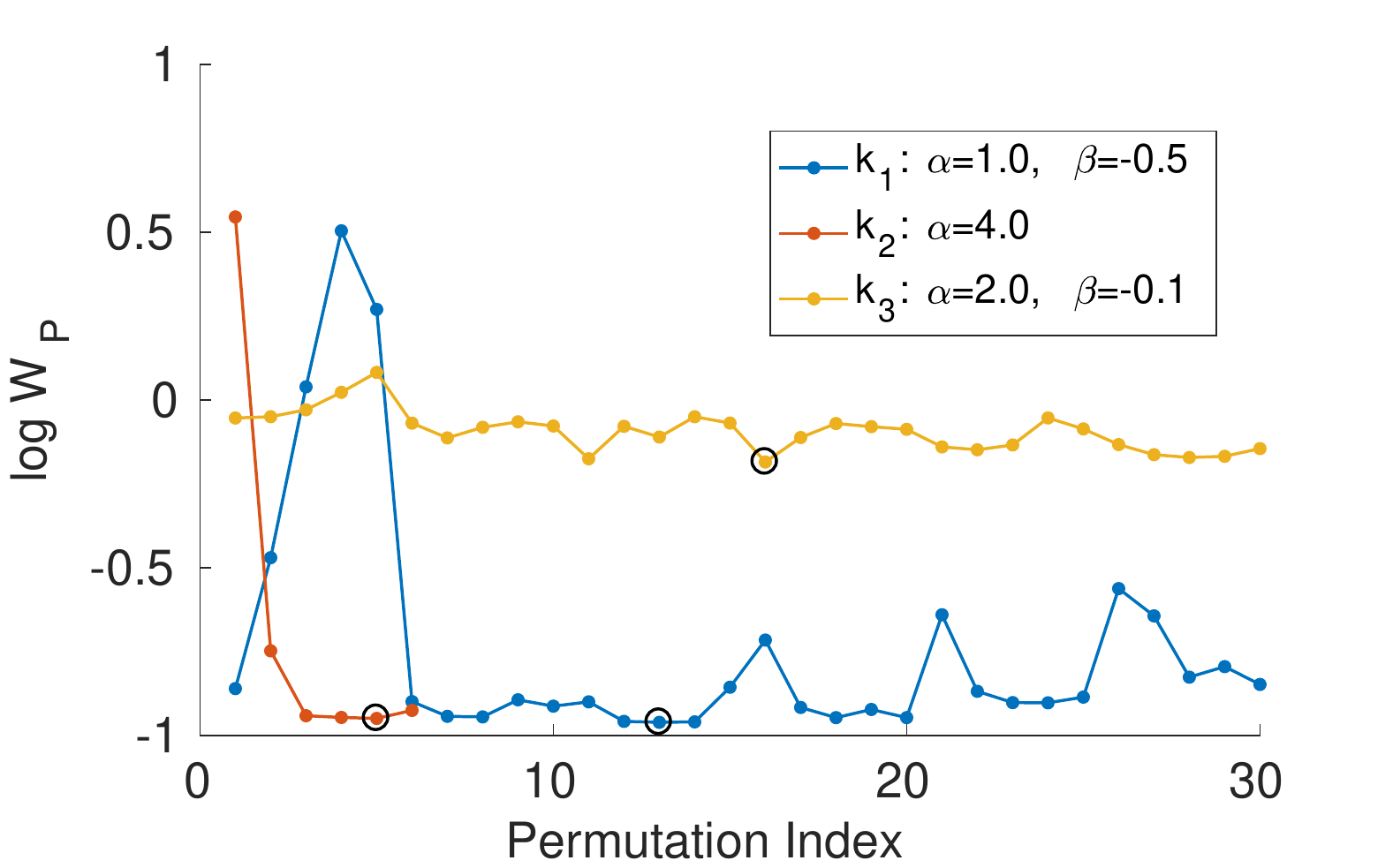}
        \caption{Stein Points (Greedy)}
        \label{fig: GMM params greedy}
    \end{subfigure}
\begin{subfigure}[b]{0.33\textwidth}
		\centering
        \includegraphics[width=\textwidth]{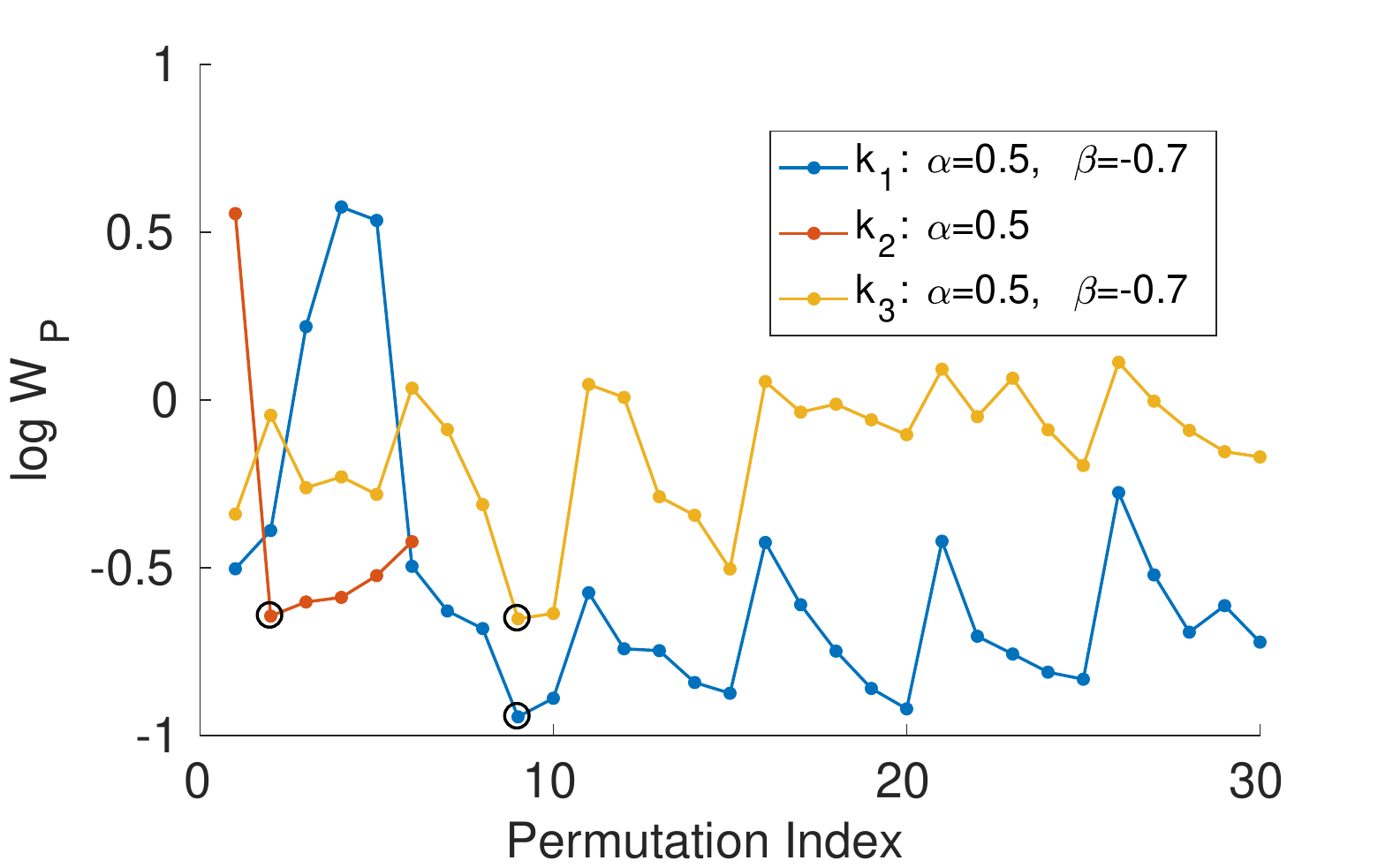}
        \caption{Stein Points (Herding)}
        \label{fig: GMM params herding}
    \end{subfigure}
\begin{subfigure}[b]{0.33\textwidth}
		\centering
        \includegraphics[width=\textwidth]{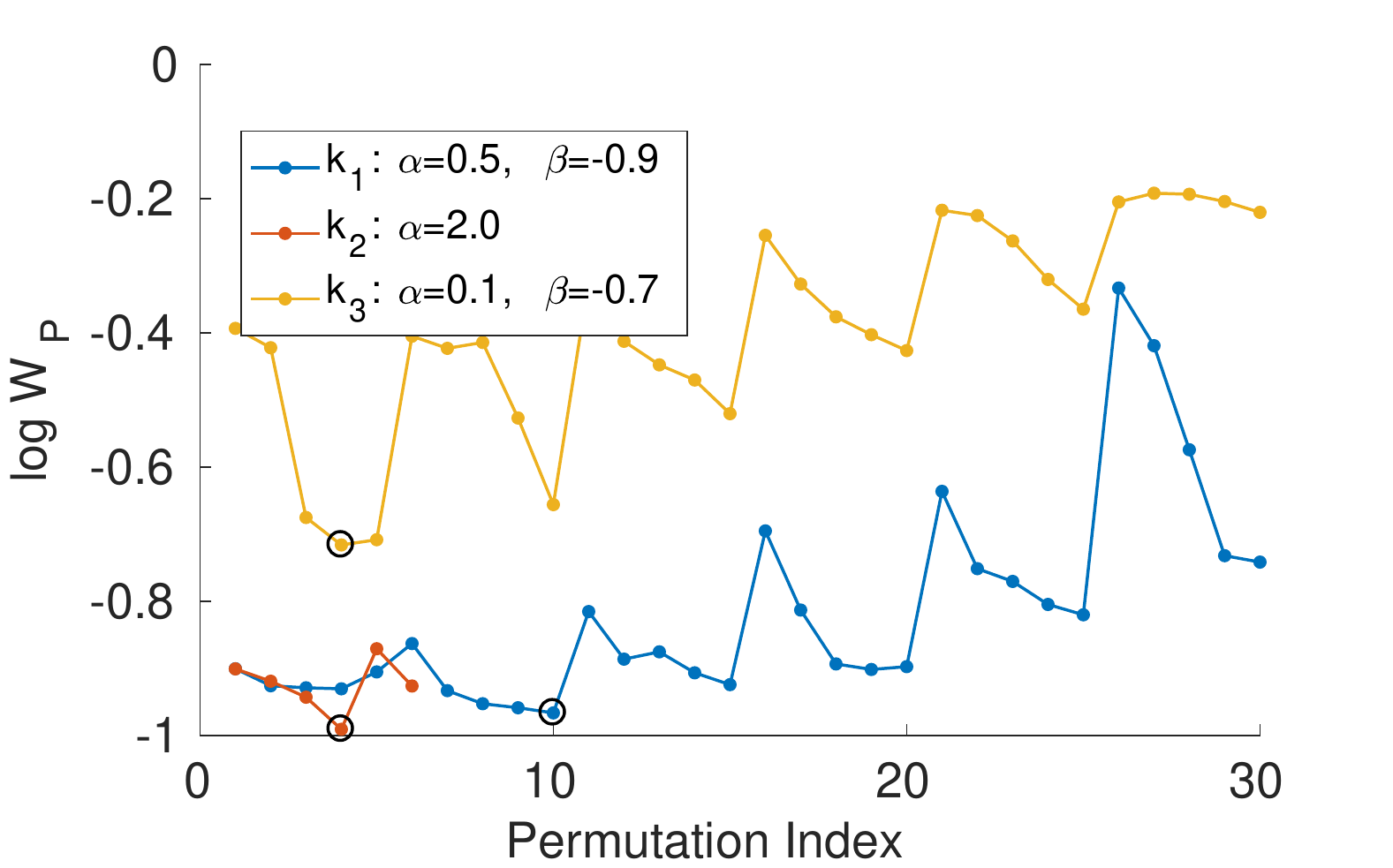}
        \caption{SVGD}
        \label{fig: GMM params SVGD}
    \end{subfigure}
\caption{Kernel parameter selection results for the Gaussian mixture test.
Parameters $\alpha,\beta$ in the kernels $k_1$, $k_2$, $k_3$ were optimised over a discrete set with respect to the Wasserstein distance $W_P$ for a point set of size $n = 100$.
The values $\log W_P$ (y-axis) are shown for all different configurations of parameters (x-axis) considered.
Optimal parameter configurations are circled and detailed in the legend.
}
\label{fig: GMM params}
\end{figure}

\begin{figure}[h]
\centering
\includegraphics[width=0.6\textwidth]{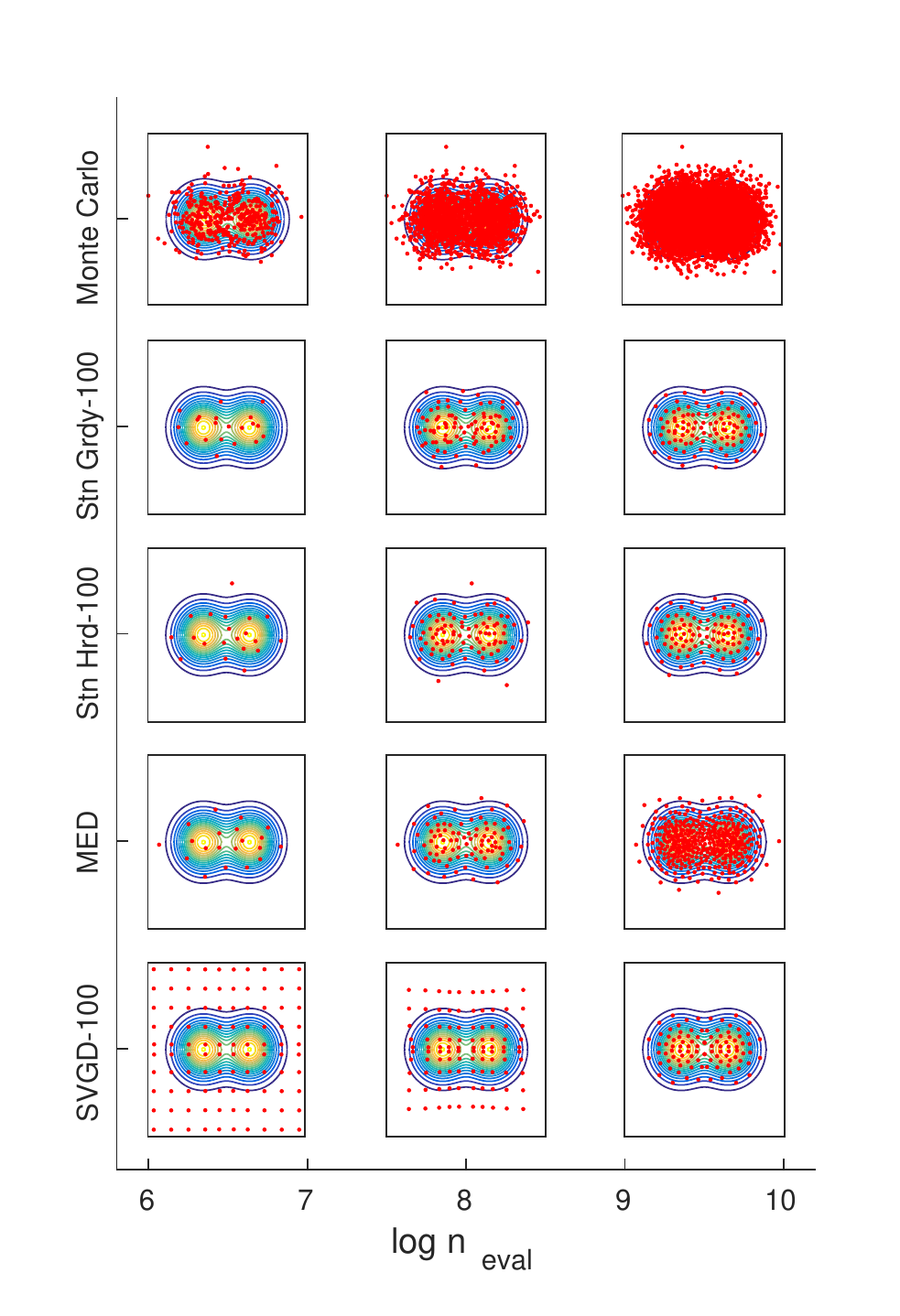}
\caption{Typical point sets obtained in the Gaussian mixture test, where the budget-constrained methods Stein Greedy-100 (Stn Grdy-100) and Stein Herding-100 (Stn Hrd-100) are considered. [Here each row corresponds to an algorithm, and each column corresponds to a chosen level of computational cost. The left border of each sub-plot is aligned to the exact value of $\log n_{\mathrm{eval}}$ spent to obtain each point-set.]
}
\label{fig: GMM points v2}
\end{figure}

\subsection{Gaussian Process Test} \label{subsec: ad numerics GP}

For the Gaussian process test, the base scale $\eta$ is also set to 1. The sensitivity of results to the selection of kernel parameters was reported in Figure \ref{fig: GP params}.
Point sets obtained under representatives of each method class are shown in Figures \ref{fig: GP points} and \ref{fig: GP points v2}.
Detailed results for each method considered are contained in Figure \ref{fig: GP}.

For all the global optimisation methods we imposed a bounding box of $(-5, 5) \times (-13,-7)$; for the Nelder-Mead method, we set $n_{\mathrm{init}} = 3$, $n_{\mathrm{delay}} = 20$, $\mu_0 = (0, -10)$, $\Sigma_0 = 25I$, and $\lambda = 1$; for the Monte Carlo method, we set $n_{\mathrm{test}} = 20$, $n_{\mathrm{delay}} = 20$, $\mu_0 = (0, -10)$, $\Sigma_0 = 25I$, and $\lambda = 1$; for the grid search, we set $n_0 = 100$.

For SVGD we set the initial point-set to be an equally spaced rectangular grid over the bounding box. Following \cite{Liu2016a}, the step-size $\epsilon$ for SVGD was determined by AdaGrad with a master step-size of 0.1 and a momentum factor of 0.9.

\begin{figure}[h]
\centering
\begin{subfigure}[b]{0.33\textwidth}
		\centering
        \includegraphics[width=\textwidth]{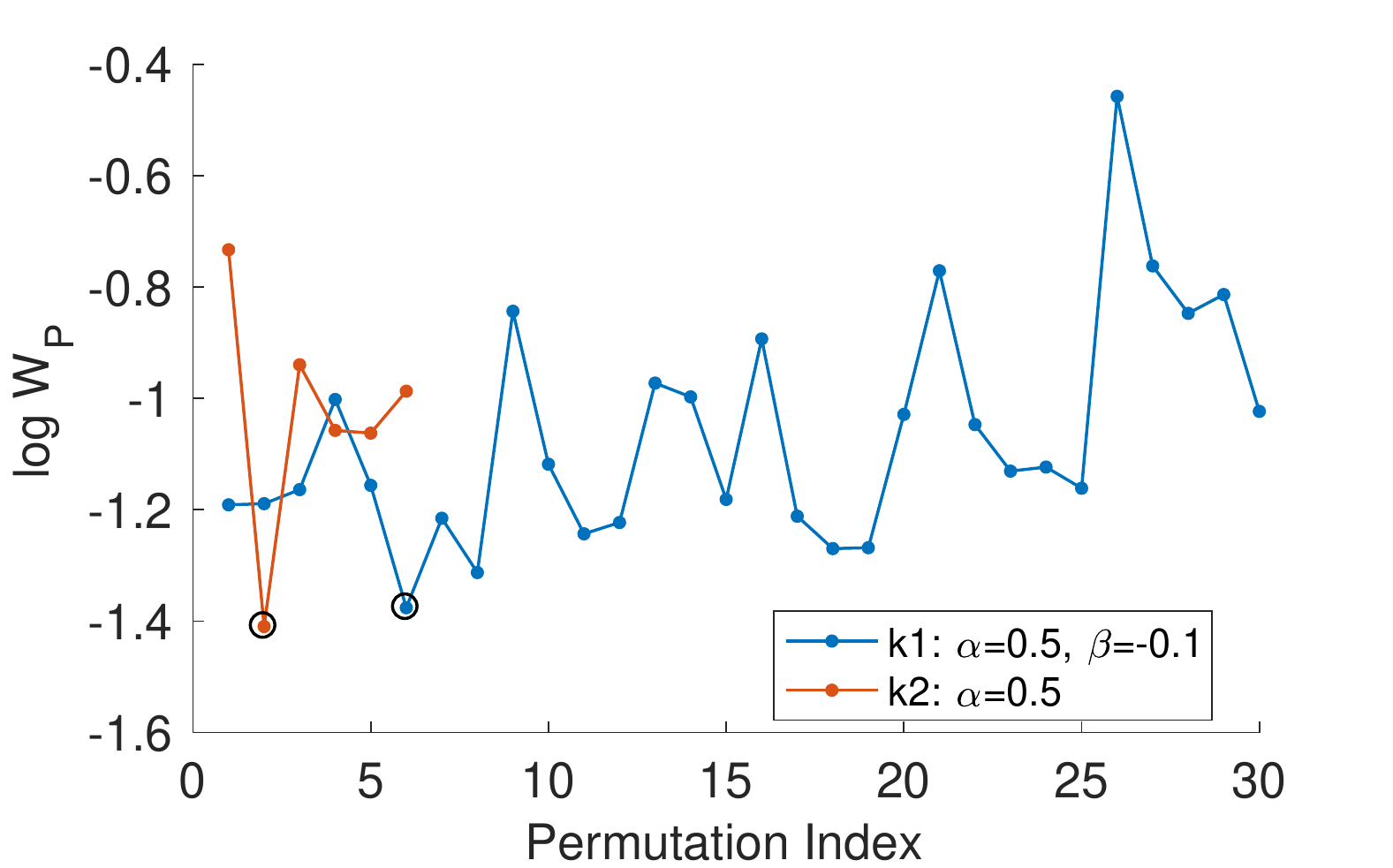}
        \caption{Stein Points (Greedy)}
        \label{fig: GP params greedy}
    \end{subfigure}
\begin{subfigure}[b]{0.33\textwidth}
		\centering
        \includegraphics[width=\textwidth]{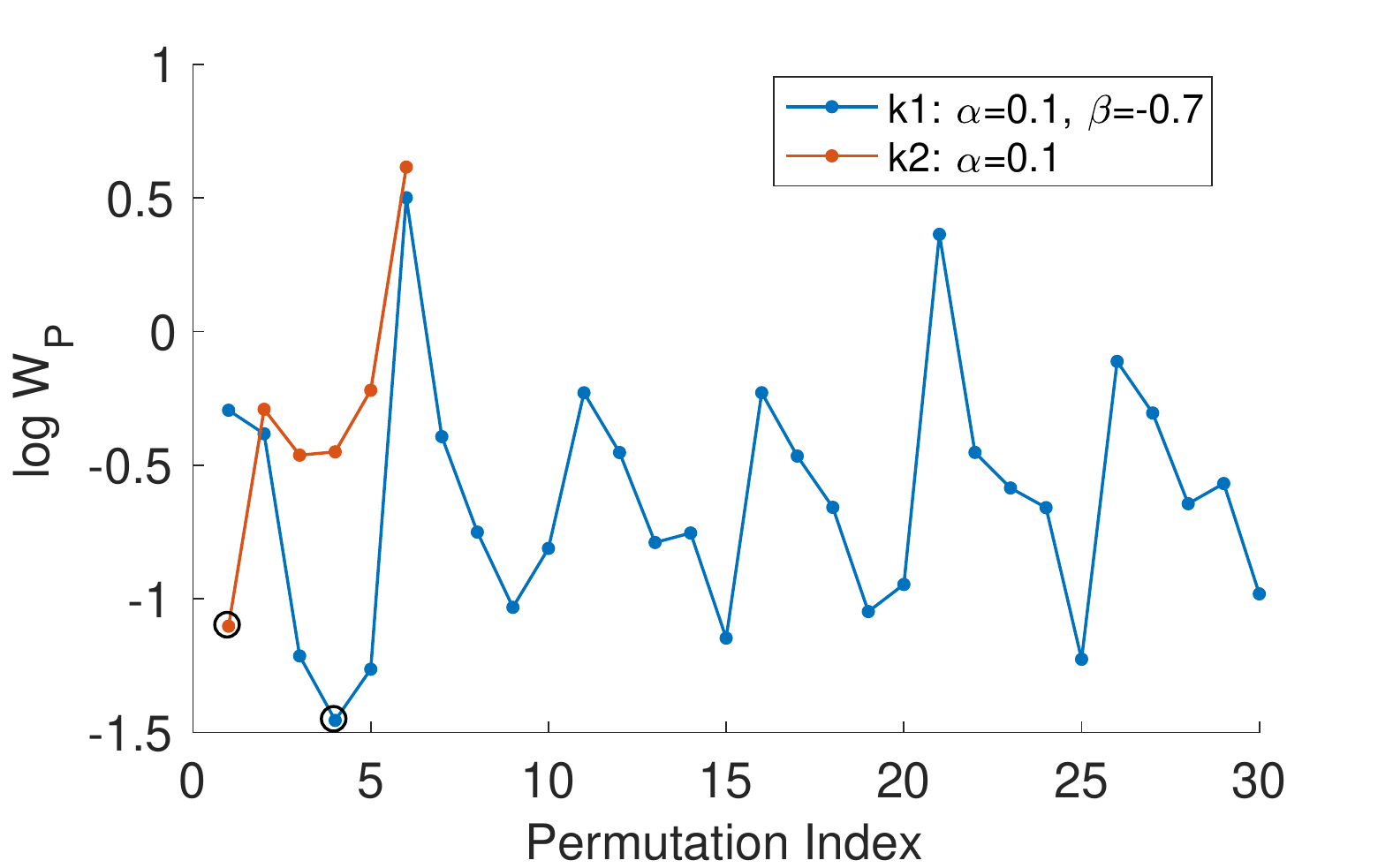}
        \caption{Stein Points (Herding)}
        \label{fig: GP params herding}
    \end{subfigure}
\begin{subfigure}[b]{0.33\textwidth}
		\centering
        \includegraphics[width=\textwidth]{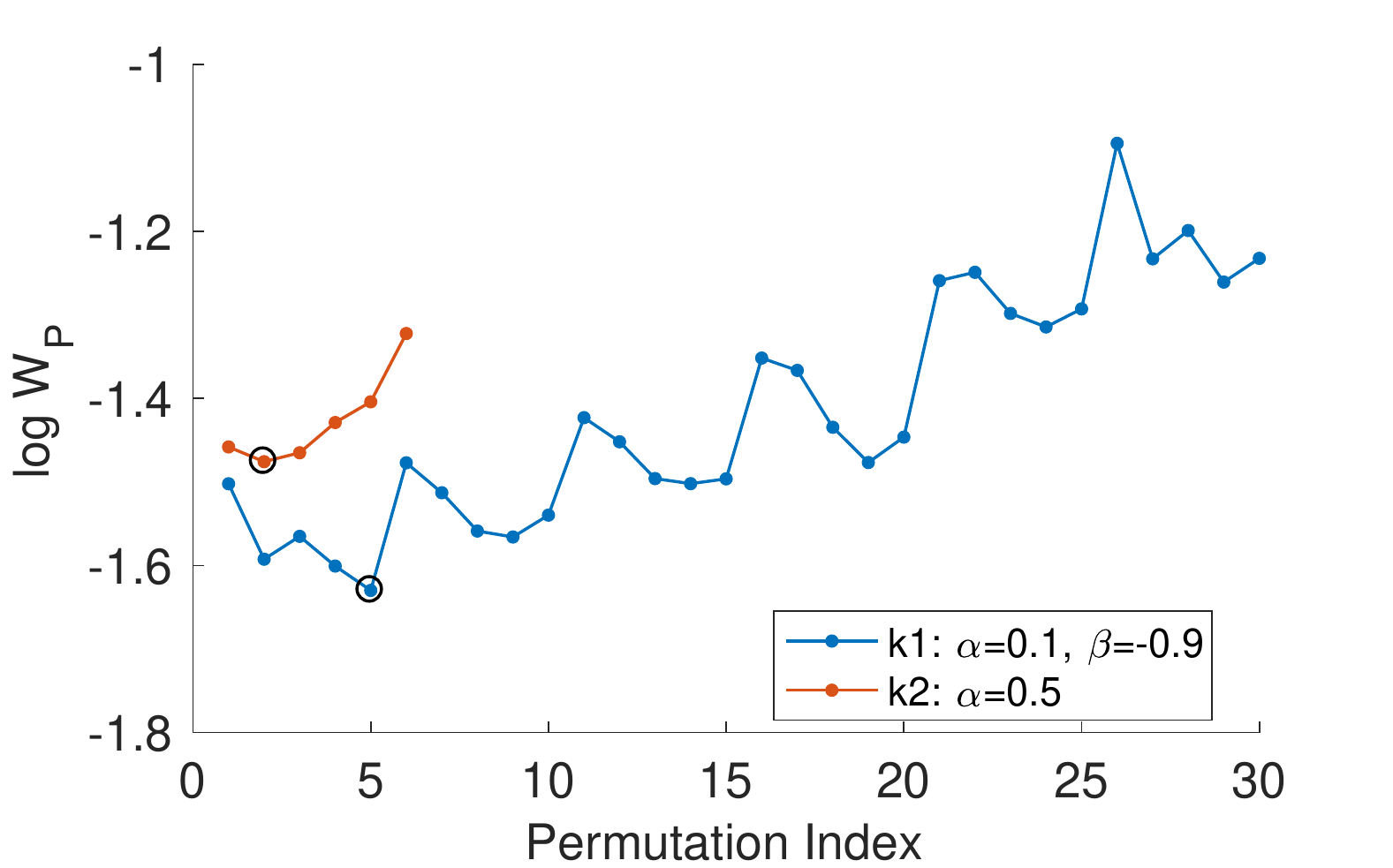}
        \caption{SVGD}
        \label{fig: GP params SVGD}
    \end{subfigure}
\caption{Kernel parameter selection results for the Gaussian process test.
Parameters $\alpha,\beta$ in the kernels $k_1$, $k_2$, $k_3$ were optimised over a discrete set with respect to the Wasserstein distance $W_P$ for a point set of size $n = 100$.
The values $\log W_P$ (y-axis) are shown for all different configurations of parameters (x-axis) considered.
Optimal parameter configurations are circled and detailed in the legend.}
\label{fig: GP params}
\end{figure}

\begin{figure}[h]
\centering
\includegraphics[width=0.6\textwidth,clip,trim = 0cm 6.5cm 0cm 0cm]{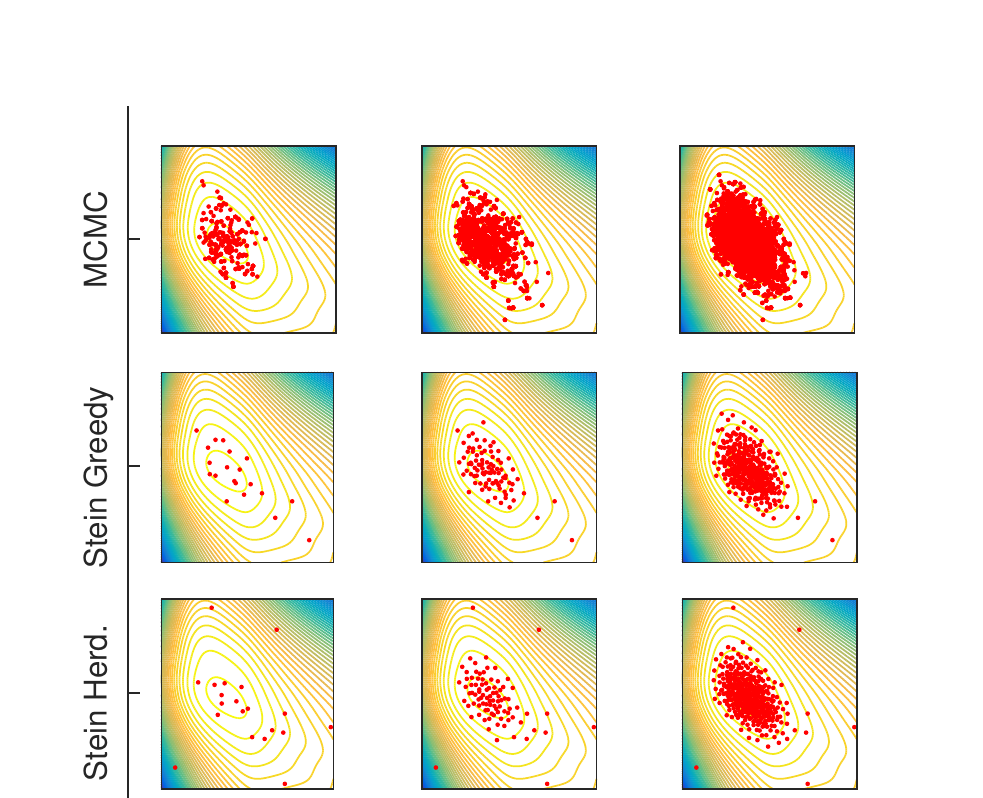}
\includegraphics[width=0.6\textwidth,clip,trim = 0cm 0cm 0cm 10.5cm]{figures/gensprpts_gp.pdf}
\caption{Typical point sets obtained in the Gaussian process test.
[Here each row corresponds to an algorithm, and each column corresponds to a chosen level of computational cost.
The left border of each sub-plot is aligned to the exact value of $\log n_{\mathrm{eval}}$ spent to obtain each point-set.
MCMC represents a random-walk Metropolis algorithm with a proposal distribution optimised according to acceptance rate.]}
\label{fig: GP points}
\end{figure}

\begin{figure}[h]
\centering
\includegraphics[width=0.6\textwidth]{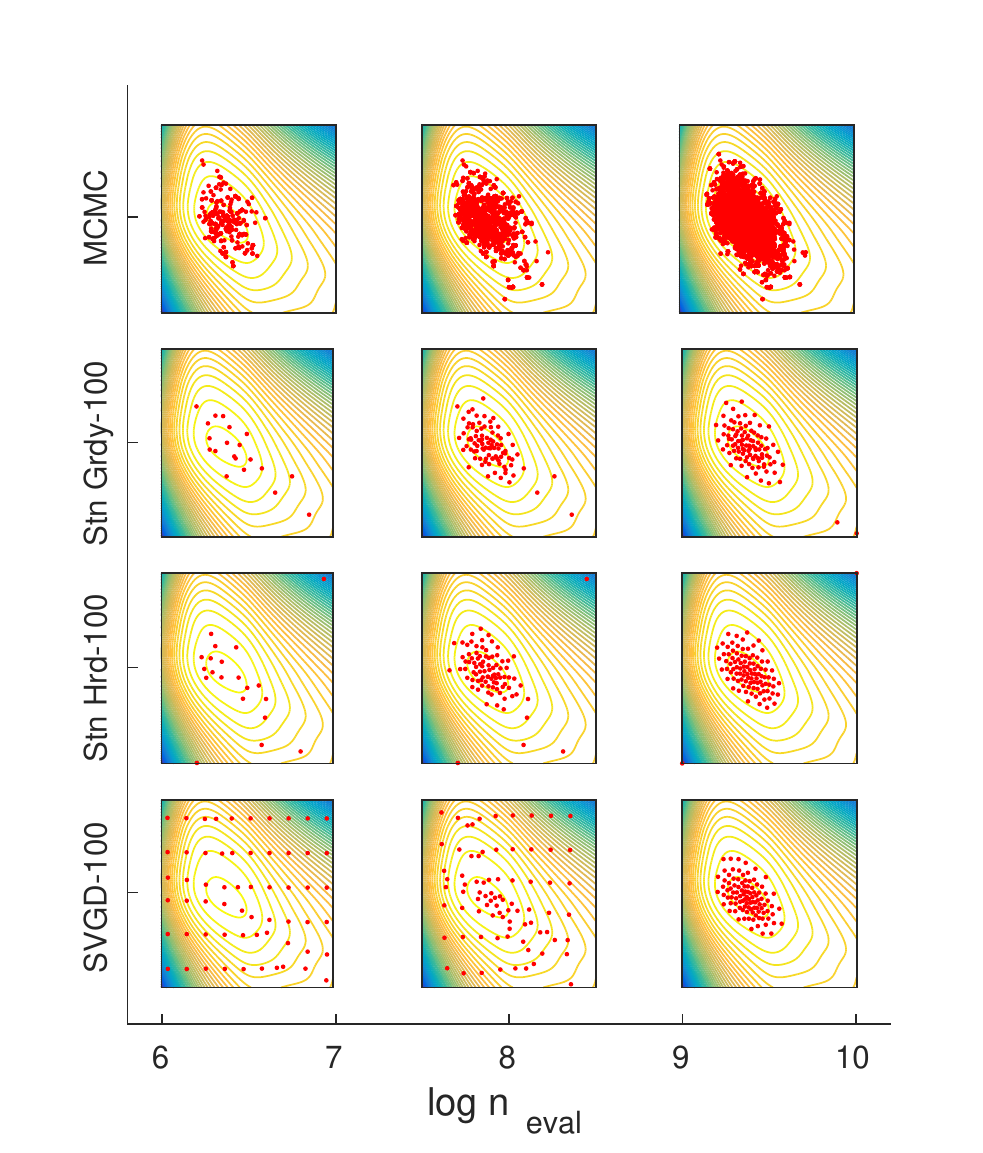}
\caption{Typical point sets obtained in the Gaussian process test, where the budget-constrained methods Stein Greedy-100 (Stn Grdy-100) and Stein Herding-100 (Stn Hrd-100) are considered. [Here each row corresponds to an algorithm, and each column corresponds to a chosen level of computational cost. The left border of each sub-plot is aligned to the exact value of $\log n_{\mathrm{eval}}$ spent to obtain each point-set.
MCMC represents a random-walk Metropolis algorithm with a proposal distribution optimised according to acceptance rate.]}
\label{fig: GP points v2}
\end{figure}

\begin{figure}[h]
\centering
\begin{subfigure}[b]{0.32\textwidth}
		\centering
        \includegraphics[width=\textwidth]{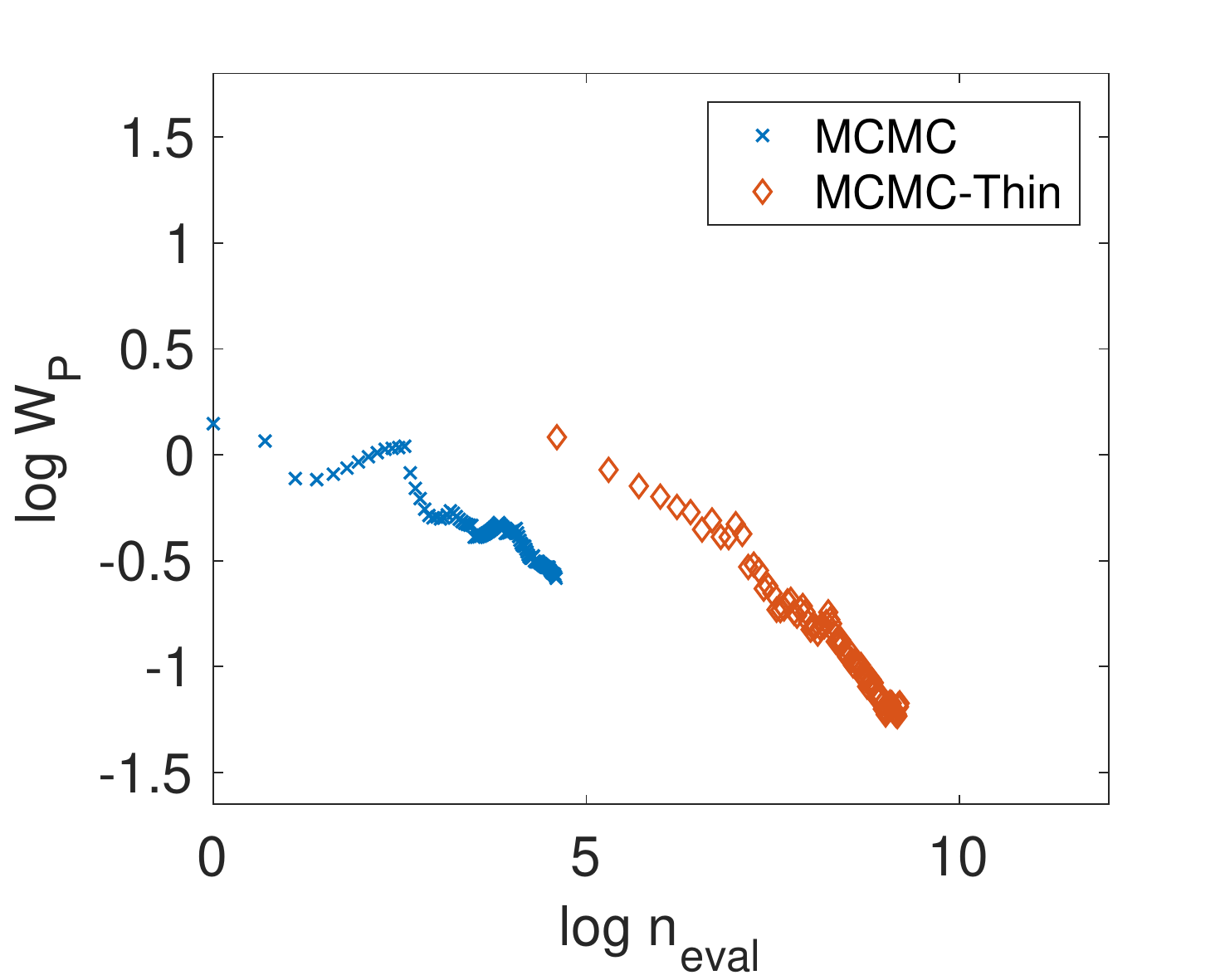}
        \caption{Monte Carlo}
        \label{fig: GP MC}
    \end{subfigure}
\begin{subfigure}[b]{0.32\textwidth}
		\centering
        \includegraphics[width=\textwidth]{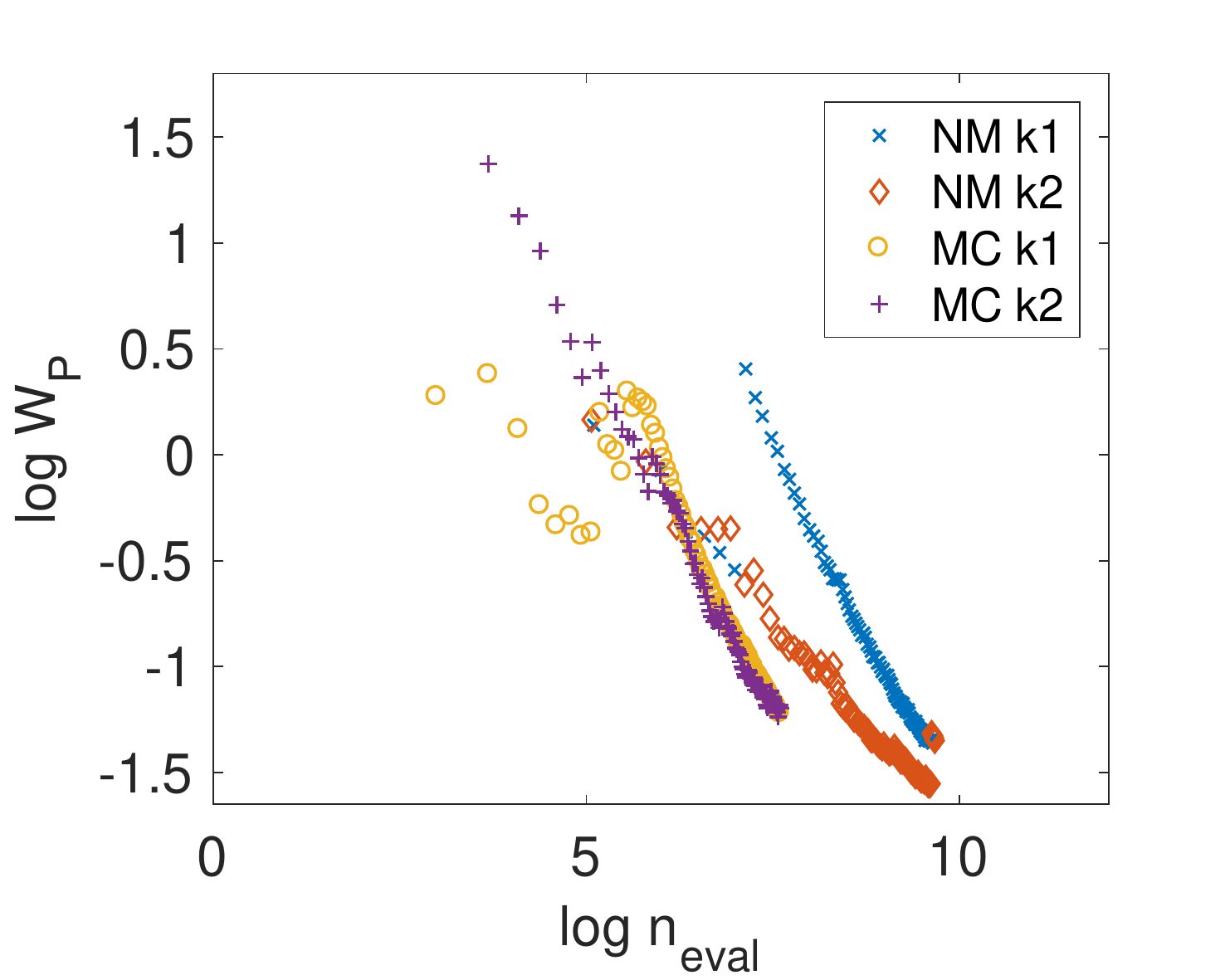}
        \caption{Stein Points (Greedy)}
        \label{fig: GP greedy}
    \end{subfigure}

\begin{subfigure}[b]{0.32\textwidth}
		\centering
        \includegraphics[width=\textwidth]{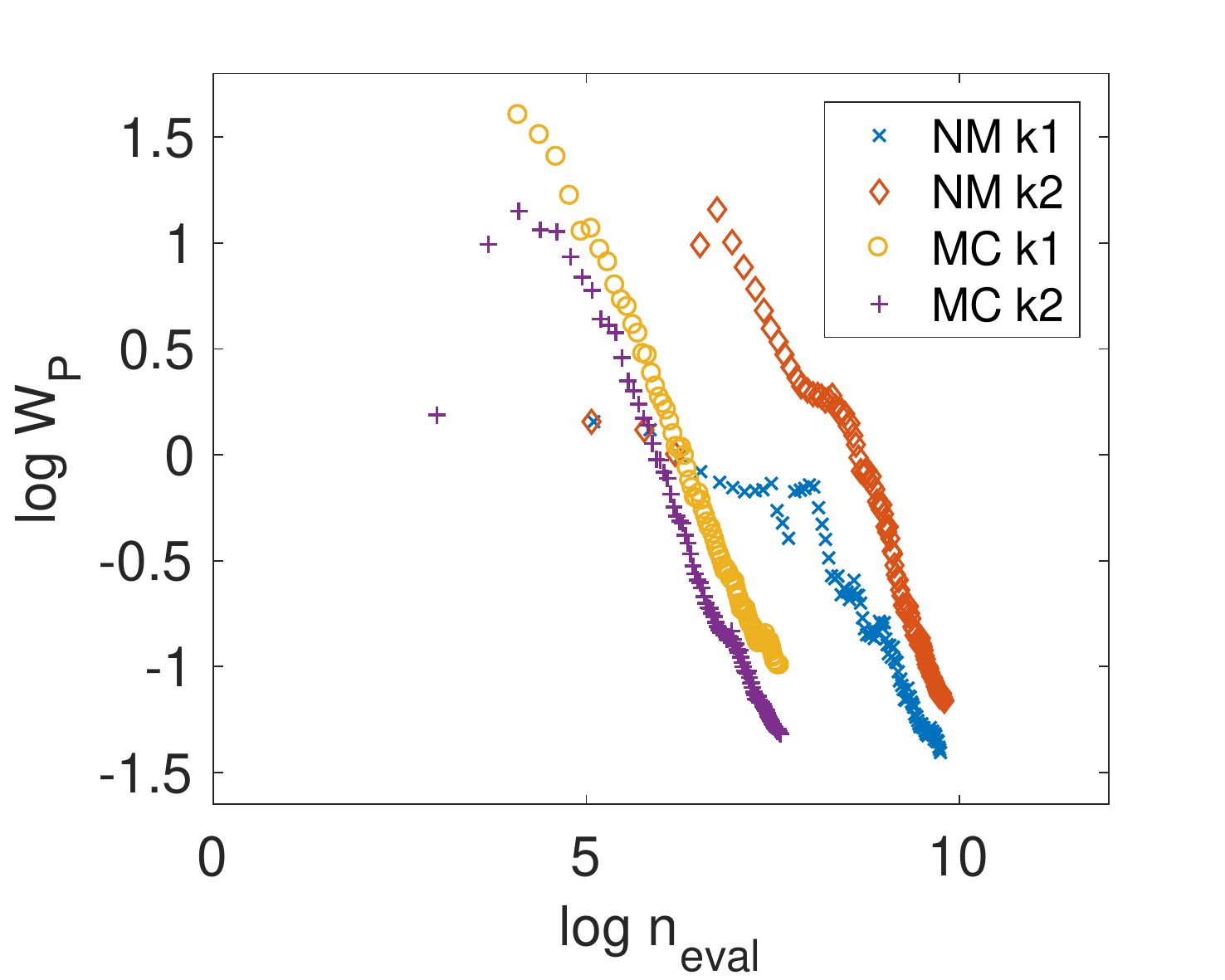}
        \caption{Stein Points (Herding)}
        \label{fig: GP herding}
    \end{subfigure}
\begin{subfigure}[b]{0.32\textwidth}
		\centering
        \includegraphics[width=\textwidth]{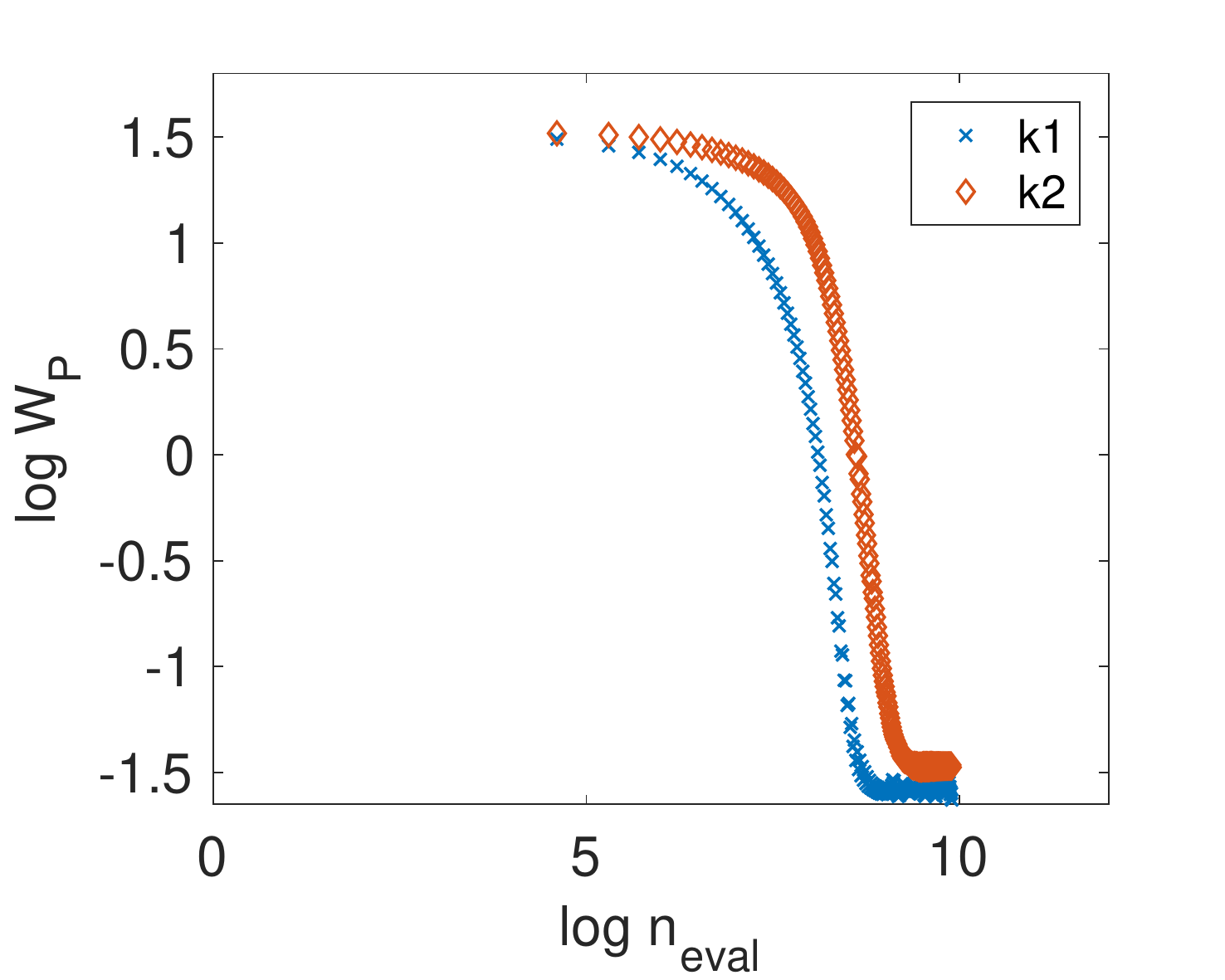}
        \caption{SVGD}
        \label{fig: GP SVGD}
    \end{subfigure}
\caption{Results for the Gaussian process test.
[Here $n = 100$.
x-axis: log of the number $n_{\text{eval}}$ of model evaluations that were used.
y-axis: log of the Wasserstein distance $W_P(\{x_i\}_{i=1}^n)$ obtained.
Kernel parameters $\alpha$, $\beta$ were optimised according to $W_P$.
In sub-figure \ref{fig: GP MC}, MCMC represents a random-walk Metropolis algorithm with a proposal distribution optimised according to acceptance rate.
MCMC-Thin represents a thinned chain by taking every 100th observation.]}
\label{fig: GP}
\end{figure}

\subsection{IGARCH Test} \label{subsec: ad numerics IGARCH}

For the IGARCH test, we choose the base scale $\eta$ to be 1e-5. The sensitivity of results to the selection of kernel parameters was reported in Figure \ref{fig: IGARCH params}.
Point sets obtained under representatives of each method class are shown in Figures \ref{fig: IGARCH points} and \ref{fig: IGARCH points v2}.
Detailed results for each method considered are contained in Figure \ref{fig: IGARCH}.

For all the global optimisation methods we impose a bounding box of $(0.002, 0.04) \times (0.05, 0.2)$; for the Nelder-Mead method, we set $n_{\mathrm{init}} = 3$, $n_{\mathrm{delay}} = 20$, $\mu_0 = (0.021, 0.125)$, $\Sigma_0 = \mathrm{diag}[(1\text{e-}4, 1\text{e-}3)]$, and $\lambda = 1\text{e-}5$; for the Monte Carlo method, we set $n_{\mathrm{test}} = 20$, $n_{\mathrm{delay}} = 20$, $\mu_0 = (0.021, 0.125)$, $\Sigma_0 = \mathrm{diag}[(1\text{e-}4, 1\text{e-}3)]$, and $\lambda = 1\text{e-}5$; for the grid search, we set $n_0 = 100$.

For SVGD we set the initial point-set to be an equally spaced rectangular grid over the bounding box. Following \cite{Liu2016a}, the step-size $\epsilon$ for SVGD was determined by AdaGrad with a master step-size of 1e-3 and a momentum factor of 0.9.

\begin{figure}[h]
\centering
\begin{subfigure}[b]{0.33\textwidth}
		\centering
        \includegraphics[width=\textwidth]{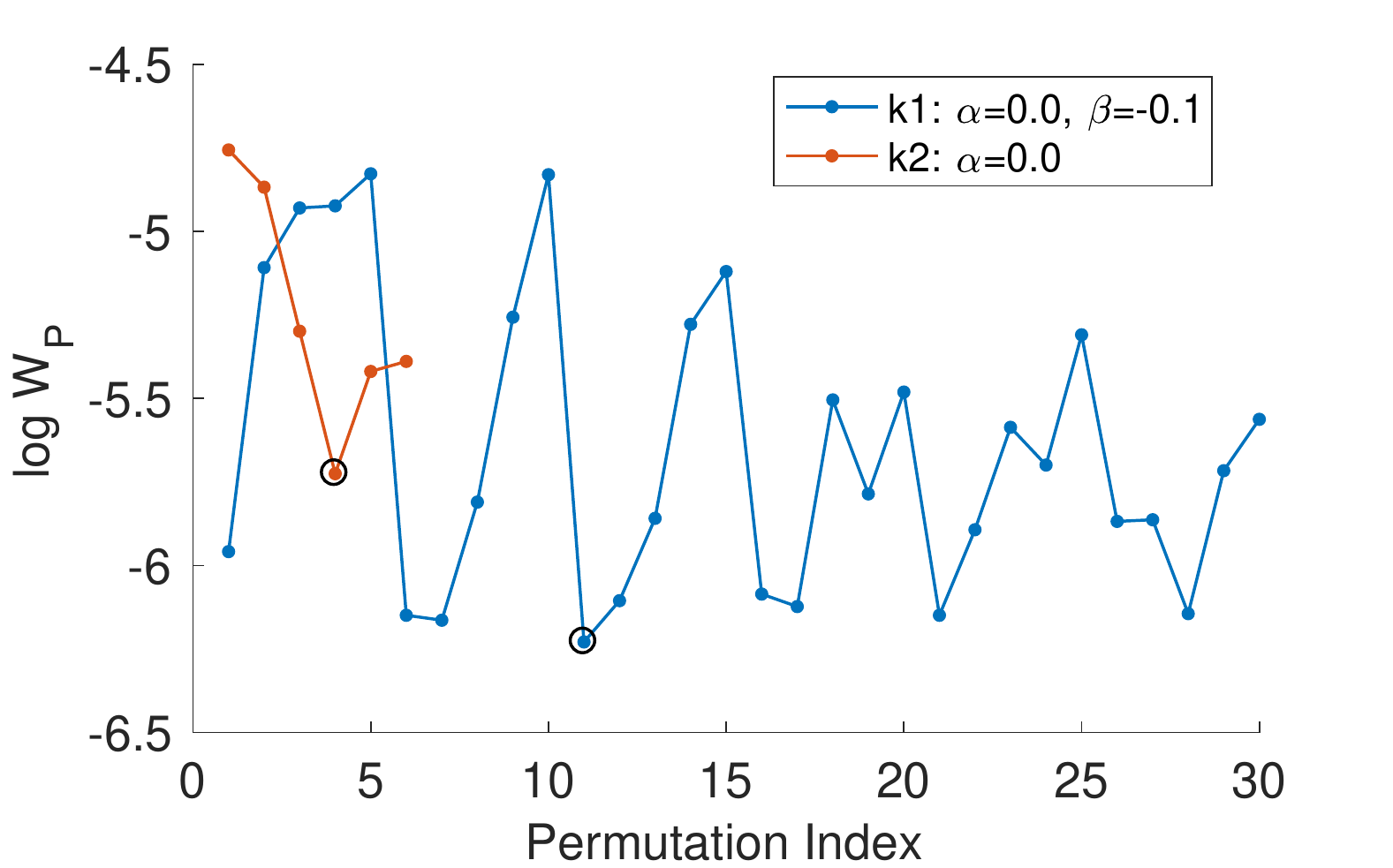}
        \caption{Stein Points (Greedy)}
        \label{fig: IGARCH params greedy}
    \end{subfigure}
\begin{subfigure}[b]{0.33\textwidth}
		\centering
        \includegraphics[width=\textwidth]{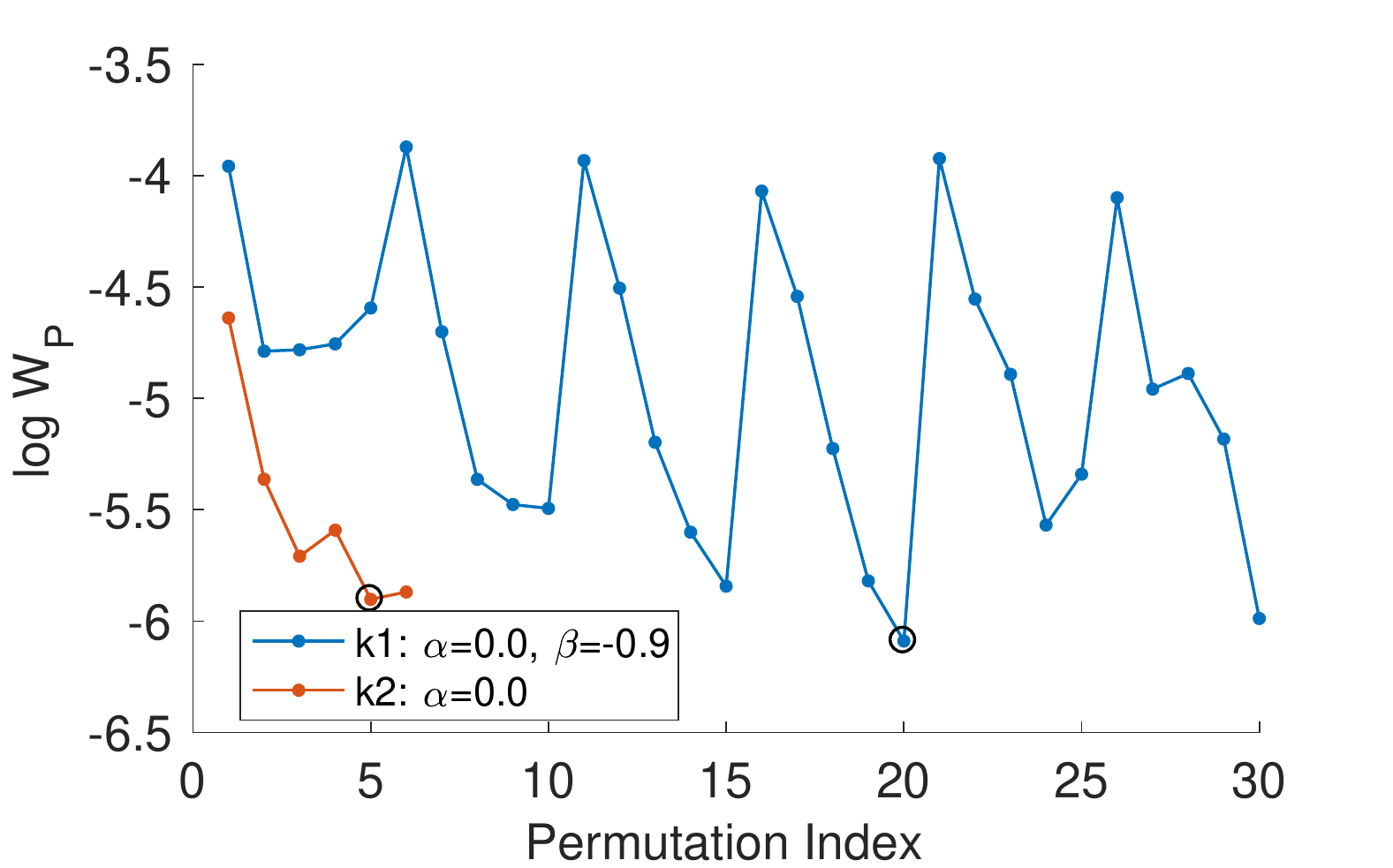}
        \caption{Stein Points (Herding)}
        \label{fig: IGARCH params herding}
    \end{subfigure}
\begin{subfigure}[b]{0.33\textwidth}
		\centering
        \includegraphics[width=\textwidth]{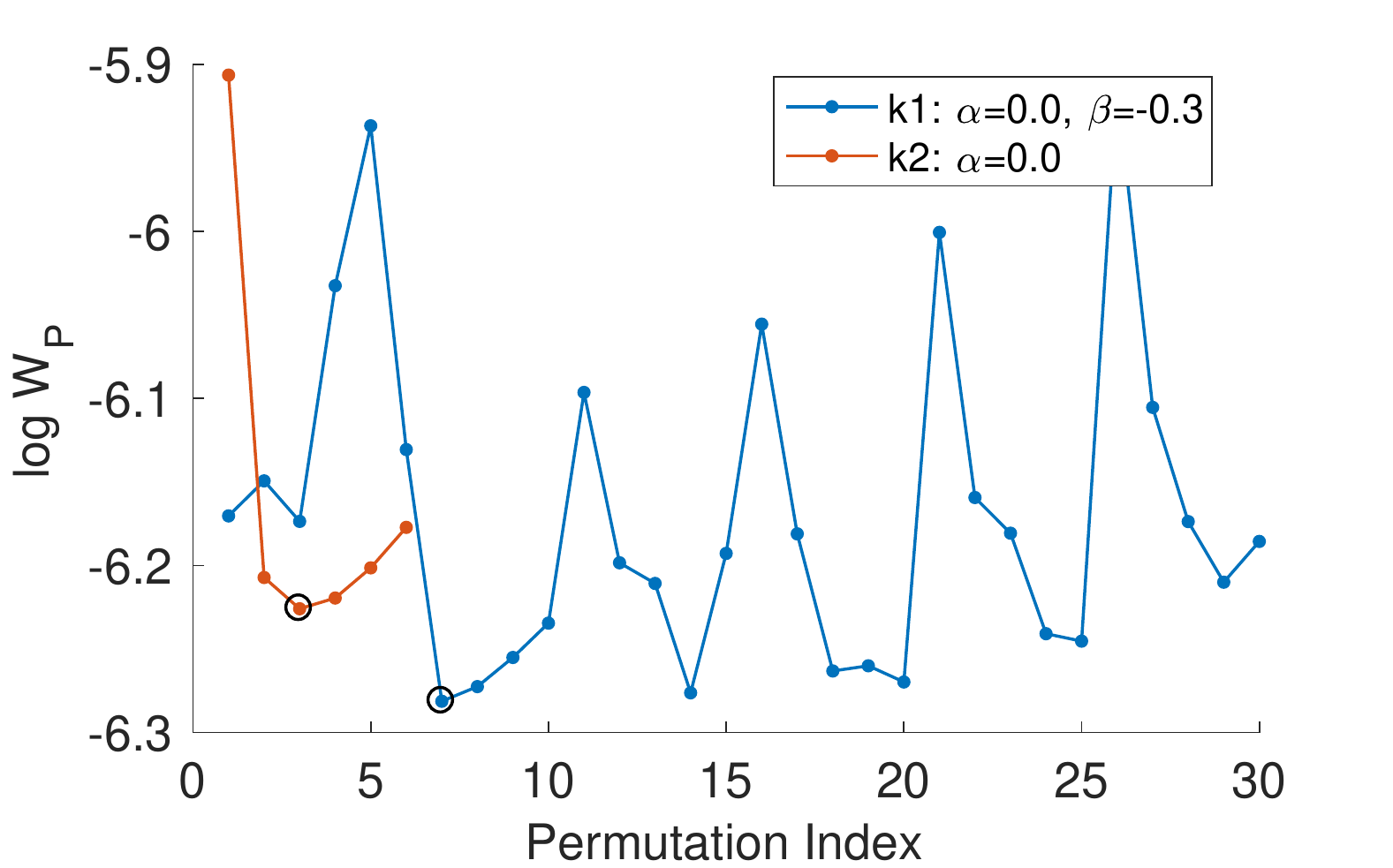}
        \caption{SVGD}
        \label{fig: IGARCH params SVGD}
    \end{subfigure}
\caption{Kernel parameter selection results for the IGARCH test.
Parameters $\alpha,\beta$ in the kernels $k_1$, $k_2$, $k_3$ were optimised over a discrete set with respect to the Wasserstein distance $W_P$ for a point set of size $n = 100$.
The values $\log W_P$ (y-axis) are shown for all different configurations of parameters (x-axis) considered.
Optimal parameter configurations are circled and detailed in the legend.}
\label{fig: IGARCH params}
\end{figure}

\begin{figure}[h]
\centering
\includegraphics[width=0.6\textwidth,clip,trim = 0cm 6.5cm 0cm 0cm]{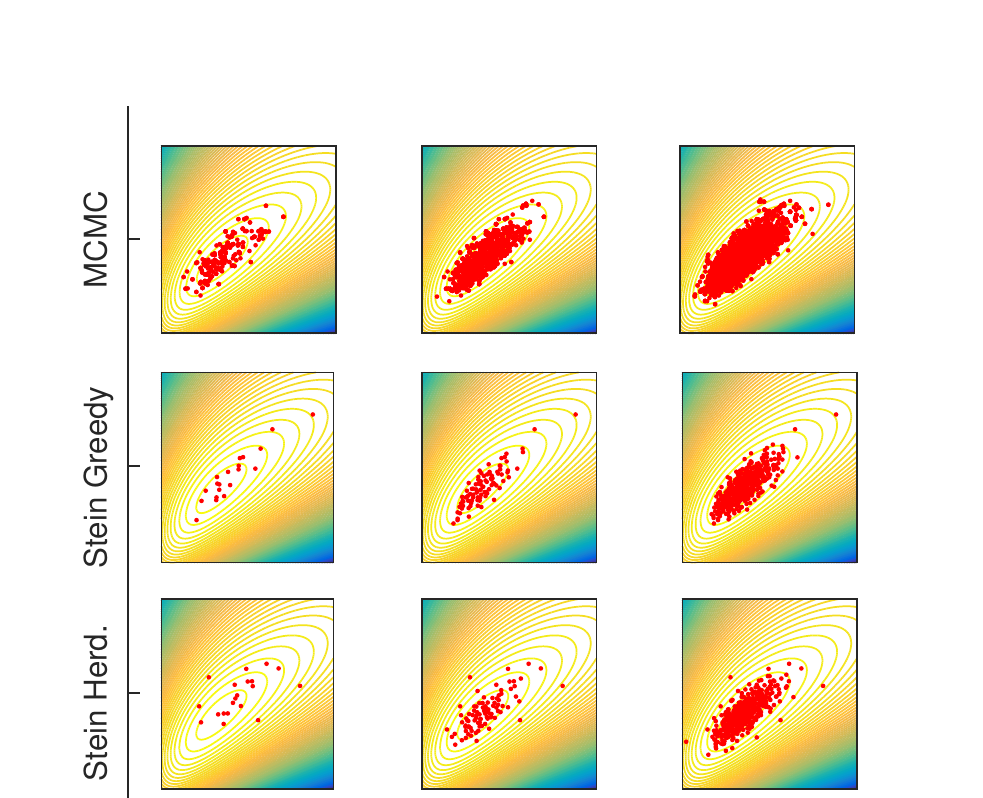}
\includegraphics[width=0.6\textwidth,clip,trim = 0cm 0cm 0cm 10.5cm]{figures/gensprpts_igarch.pdf}
\caption{Typical point sets obtained in the IGARCH test.
[Here each row corresponds to an algorithm, and each column corresponds to a chosen level of computational cost.
The left border of each sub-plot is aligned to the exact value of $\log n_{\mathrm{eval}}$ spent to obtain each point-set.
MCMC represents a random-walk Metropolis algorithm with a proposal distribution optimised according to acceptance rate.]}
\label{fig: IGARCH points}
\end{figure}

\begin{figure}[h]
\centering
\includegraphics[width=0.6\textwidth]{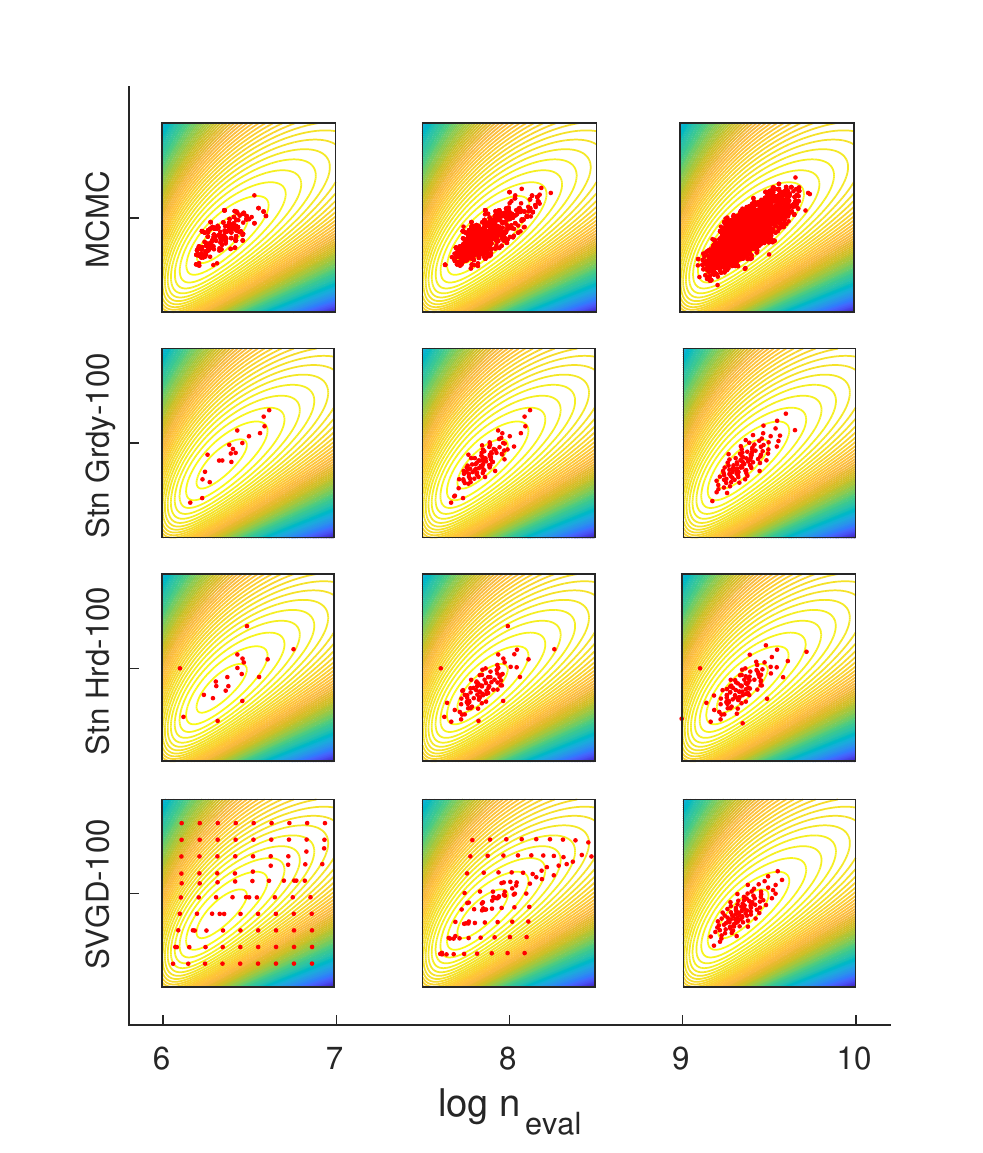}
\caption{Typical point sets obtained in the IGARCH test, where the budget-constrained methods Stein Greedy-100 (Stn Grdy-100) and Stein Herding-100 (Stn Hrd-100) are considered. [Here each row corresponds to an algorithm, and each column corresponds to a chosen level of computational cost. The left border of each sub-plot is aligned to the exact value of $\log n_{\mathrm{eval}}$ spent to obtain each point-set.
MCMC represents a random-walk Metropolis algorithm with a proposal distribution optimised according to acceptance rate.]}
\label{fig: IGARCH points v2}
\end{figure}

\begin{figure}[h]
\centering
\begin{subfigure}[b]{0.32\textwidth}
		\centering
        \includegraphics[width=\textwidth]{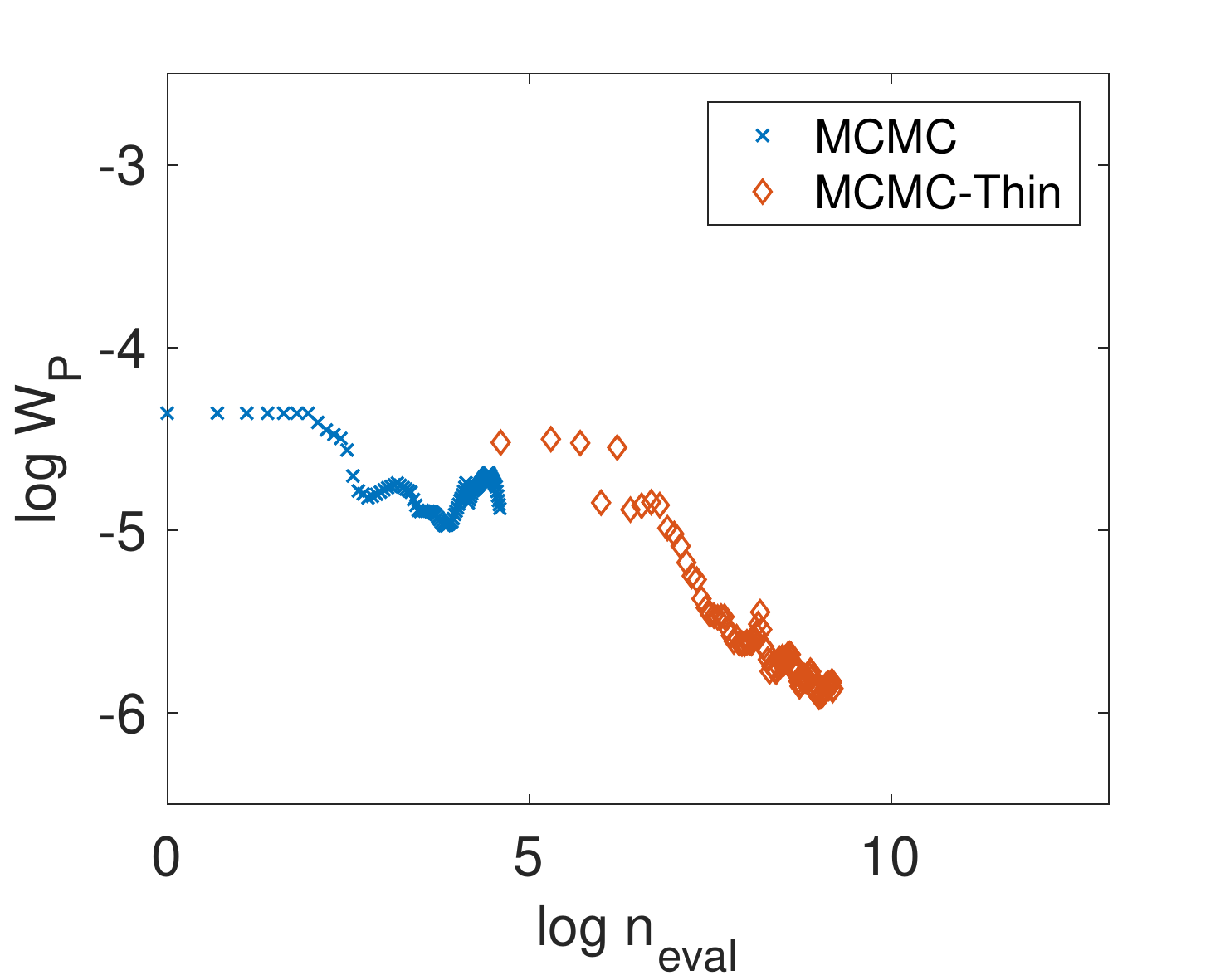}
        \caption{Monte Carlo}
        \label{fig: IGARCH MC}
    \end{subfigure}
\begin{subfigure}[b]{0.32\textwidth}
		\centering
        \includegraphics[width=\textwidth]{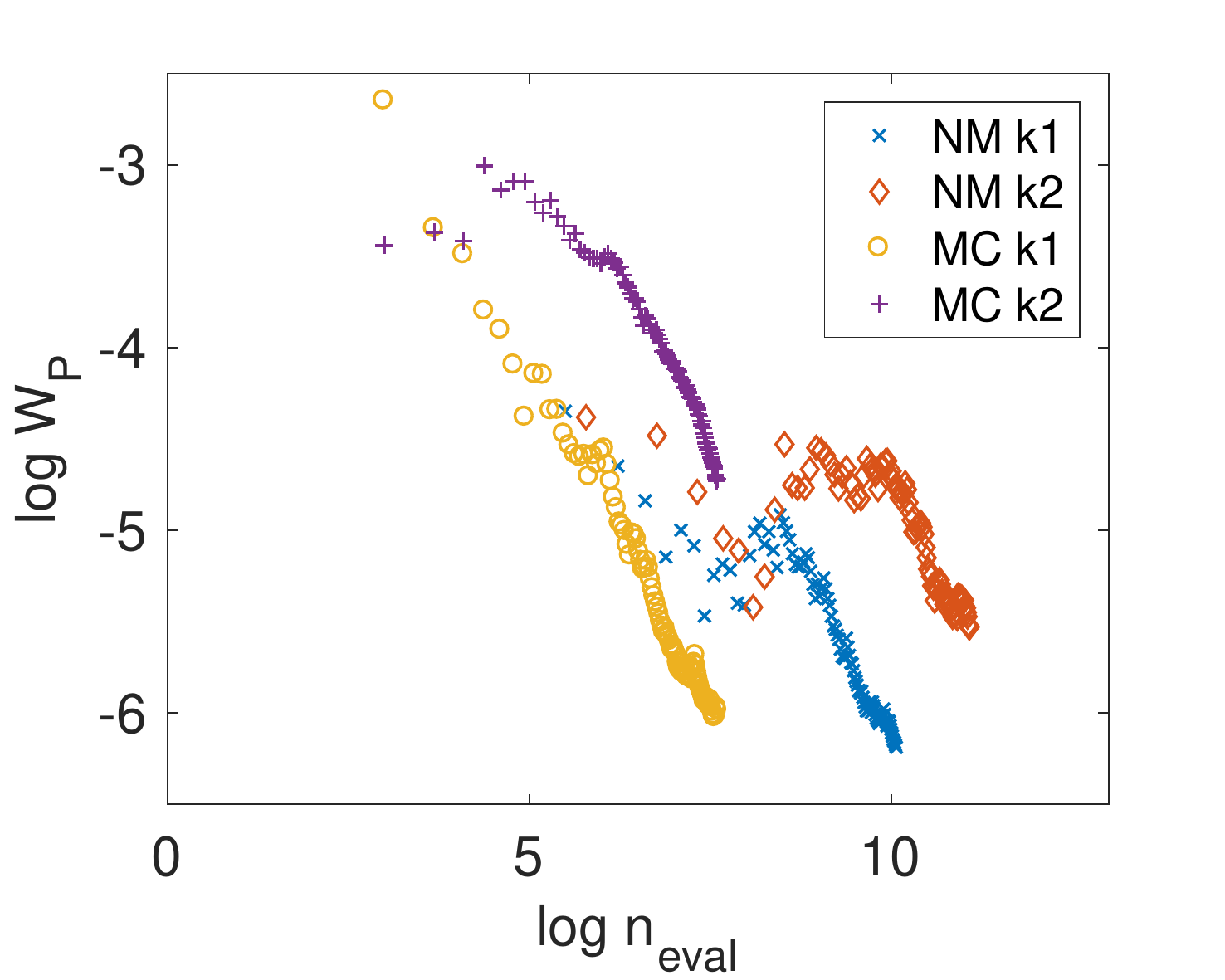}
        \caption{Stein Points (Greedy)}
        \label{fig: IGARCH greedy}
    \end{subfigure}

\begin{subfigure}[b]{0.32\textwidth}
		\centering
        \includegraphics[width=\textwidth]{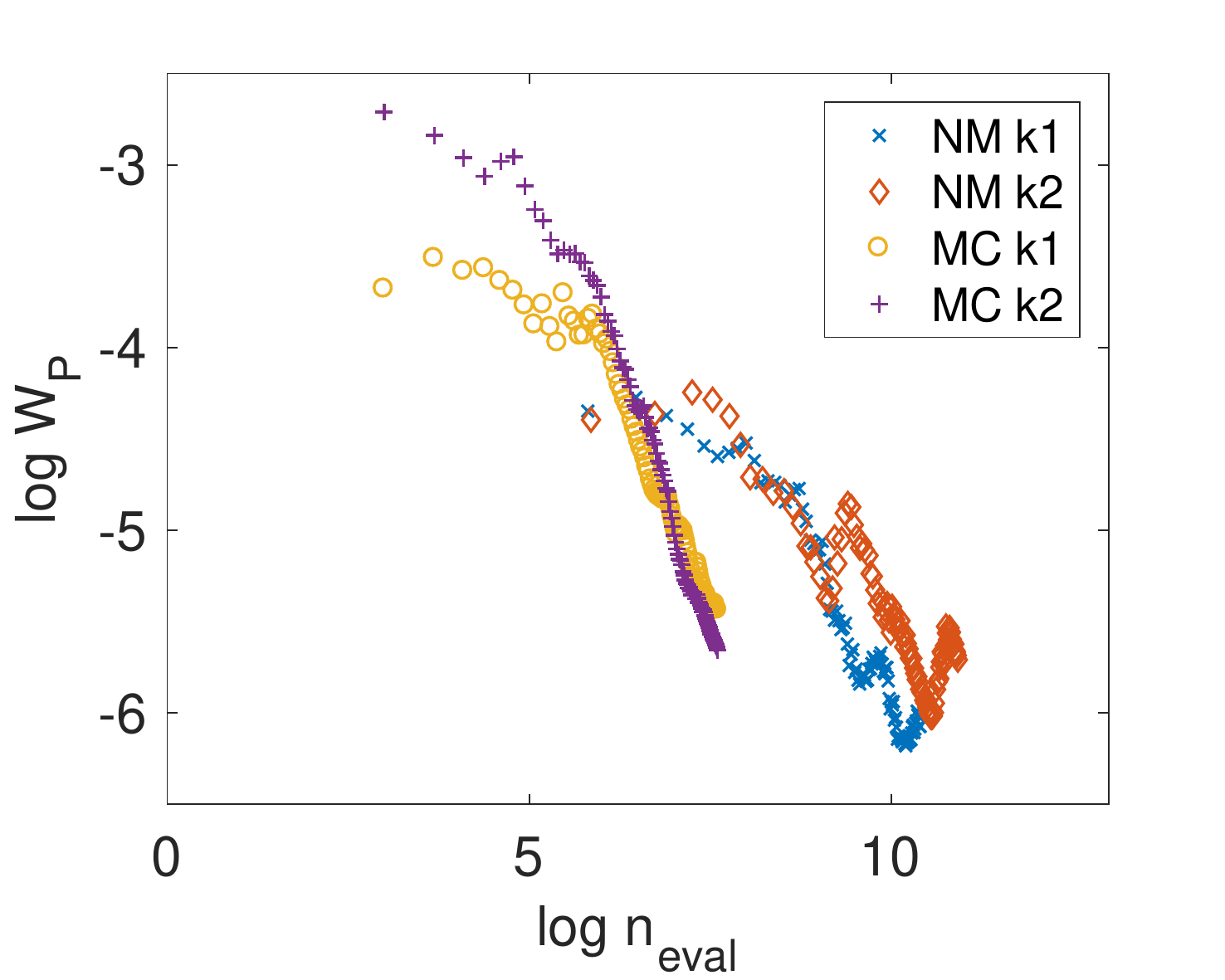}
        \caption{Stein Points (Herding)}
        \label{fig: IGARCH herding}
    \end{subfigure}
\begin{subfigure}[b]{0.32\textwidth}
		\centering
        \includegraphics[width=\textwidth]{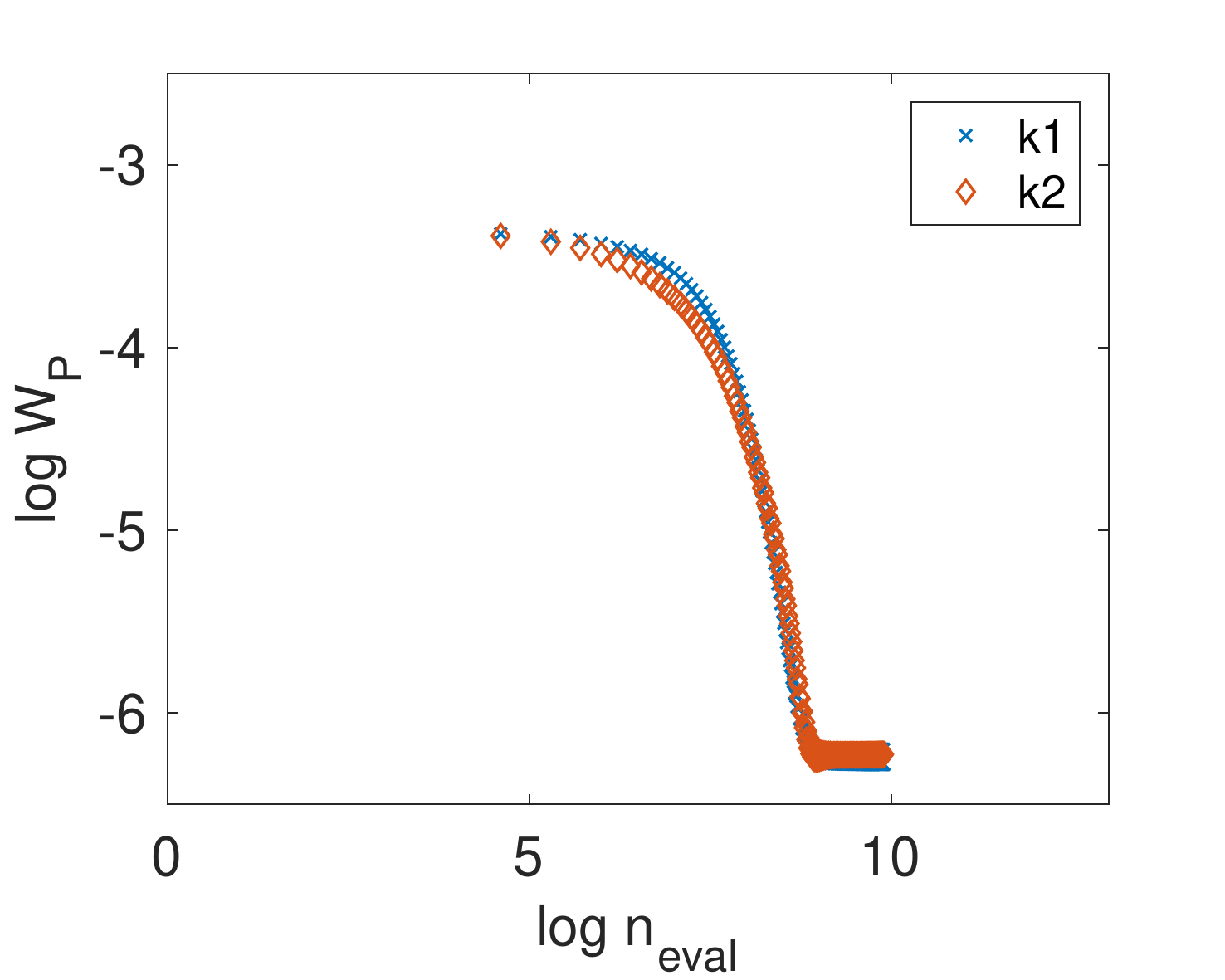}
        \caption{SVGD}
        \label{fig: IGARCH SVGD}
    \end{subfigure}
\caption{Results for the IGARCH test.
[Here $n = 100$.
x-axis: log of the number $n_{\text{eval}}$ of model evaluations that were used.
y-axis: log of the Wasserstein distance $W_P(\{x_i\}_{i=1}^n)$ obtained.
Kernel parameters $\alpha$, $\beta$ were optimised according to $W_P$.
In sub-figure \ref{fig: IGARCH MC}, MCMC represents a random-walk Metropolis algorithm with a proposal distribution optimised according to acceptance rate.
MCMC-Thin represents a thinned chain by taking every 100th observation.]}
\label{fig: IGARCH}
\end{figure}

\end{document}